\documentclass[12pt]{article}
\usepackage[dvipsnames,svgnames,x11names]{xcolor}
\usepackage{amsmath, amsthm, thm-restate}
\usepackage{xspace}
\usepackage{float}
\usepackage[T1]{fontenc}
\usepackage{mdframed}
\usepackage{mathtools}
\usepackage{enumitem}

\usepackage[linesnumbered,ruled]{algorithm2e}

\usepackage{algpseudocode}
\usepackage{authblk}

\usepackage{fullpage}
\usepackage{thmtools}
\usepackage[colorlinks]{hyperref}
\hypersetup{colorlinks={true},linkcolor={blue},citecolor=magenta}

\usepackage{amssymb,amsfonts,amsmath,amsthm,amscd,dsfont,mathrsfs}
\usepackage{complexity}
\usepackage{multirow}
\usepackage{mathpazo}
\usepackage{braket}
\usepackage{comment}
\usepackage{quantikz}
\usepackage[
backend=biber,
style=alphabetic,
sorting=ynt,
maxnames=99
]{biblatex}
\usepackage{url}
\usepackage{bbm}
\usepackage[capitalize]{cleveref}
\usepackage{hyperref} 
\usepackage[margin=1in]{geometry}
\hypersetup{breaklinks=true}

\theoremstyle{plain}
\newtheorem{theorem}{Theorem}[section]
\newtheorem{lemma}[theorem]{Lemma}
\newtheorem*{claim*}{Claim}
\newtheorem{corollary}[theorem]{Corollary}
\newtheorem{definition}[theorem]{Definition}
\newtheorem{remark}[theorem]{Remark}

\newtheorem{game}{Game}
\newtheorem*{theorem*}{Theorem}
\newtheorem*{definition*}{Definition}

\newmdtheoremenv[backgroundcolor=gray!10,
                 linewidth=0pt,
                 innerleftmargin=4pt,
                 innerrightmargin=4pt,
                 innertopmargin=1pt,
                 innerbottommargin=4pt,
            splitbottomskip=4pt]{problem}[prob]{Problem}

\newmdtheoremenv[backgroundcolor=gray!10,
                 linewidth=0pt,
                 innerleftmargin=4pt,
                 innerrightmargin=4pt,
                 innertopmargin=1pt,
                 innerbottommargin=5.5pt,
            splitbottomskip=4pt]{conjecture}[conj]{Conjecture}

\newcommand{\abs}[1]{\left| #1 \right|}

\mathchardef\mhyphen="2D

\newcommand{\Id}{\mathbb{1}}
\newcommand{\xinD}[1]{\Pi^{ #1 \in\mathsf{db}}}
\newcommand{\inD}{\Pi^{\in\mathsf{db}}}
\newcommand{\ninD}{\Pi^{\in\mathsf{db}, \perp}}
\newcommand{\xninD}[1]{\Pi^{#1 \in\mathsf{db}, \perp}}

\newcommand{\pilt}{\Pi^{\leq t}}
\newcommand{\pit}{\Pi^{t}}

\renewcommand{\proj}[1]{\ensuremath{|#1\rangle \langle #1|}}
\newcommand{\norm}[1]{\left\Vert {#1} \right\Vert}
\renewcommand{\vec}[1]{\mathbf{#1}} 
\newcommand{\Tr}{\mathrm{Tr}}

\renewcommand{\E}{\mathbb{E}}

\newcommand{\bit}{\{0,1\}}

\newcommand{\algo}{\mathcal}

\newcommand{\negl}{\ensuremath{\operatorname{negl}}\xspace}

\newcommand{\heart}{\ensuremath\heartsuit}
\newcommand{\outerprod}[2]{|#1\rangle\langle #2|}
\newcommand{\innerprod}[2]{\langle #1| #2\rangle}
\newcommand{\cp}{\mathsf{cpO}} 
\newcommand{\cf}{\mathsf{cfO}} 

\newcommand{\la}{\leftarrow}
\newcommand{\ra}{\rightarrow}

\newcommand{\ua}{\uparrow}
\newcommand{\da}{\downarrow}

\newcommand{\reg}[1]{{\color{darkgray}\mathsf{#1}}}
\newcommand{\ol}{\overline}

\newcommand{\spn}[1]{\mathsf{span}\{#1\}}
\newcommand{\dom}{\mathsf{Dom}}
\newcommand{\im}{\mathsf{Im}}
\newcommand{\fc}{\mathsf{fC}} 
\newcommand{\pc}{\mathsf{pC}}

\newcommand{\pu}{\mathsf{P}}
\newcommand{\mfpu}{\mathsf{mfP}}
\newcommand{\botd}{\pmb \bot}
\newcommand{\boti}{\pmb \bot}
\newcommand{\flip}{\mathsf{F}}
\newcommand{\iso}{\mathcal{I}}
\newcommand{\cmf}{\mathsf{cmfO}}
\newcommand{\mfc}{\mathsf{mfC}}
\newcommand{\mflip}{\mathsf{mfF}}
\newcommand{\supp}{\mathsf{Supp}}
\newcommand{\adv}[1]{{\mathsf{Dist}(#1)}}
\newcommand{\feist}{\mathsf{Feist}}
\newcommand{\piv}{\Pi^{\mathsf{val}}}
\newcommand{\hval}{\mathcal H^{\mathsf{val}}}
\newcommand{\piqv}{\Pi^{\mathsf{qval}}}
\newcommand{\hqval}{\mathcal H^{\mathsf{qval}}}

\newcommand{\pinqv}{\Pi^{\mathsf{qval}, \perp}}

\newcommand{\hsoph}{\mathcal H^{\mathsf{soph}}}
\newcommand{\pisoph}{\Pi^{\mathsf{soph}}}
\newcommand{\hnsoph}{\mathcal H^{\mathsf{soph}, \perp}}
\newcommand{\pinsoph}{\Pi^{\mathsf{soph}, \perp}}
\newcommand{\hele}{\mathcal H^{\mathsf{ele}}}
\newcommand{\hfele}{\mathcal H^{\mathsf{fele}}}
\newcommand{\piele}{\Pi^{\mathsf{ele}}}

\newcommand{\pinele}{\Pi^{\mathsf{ele}, \perp}}
\newcommand{\pifele}{\Pi^{\mathsf{fele}}}
\newcommand{\pinfele}{\Pi^{\mathsf{fele},\perp}}

\title{Compressed Permutation Oracles}
\author{Joseph Carolan}

\affil{University of Maryland}

\date{}

\begin{document}

\maketitle

\begin{abstract}
    The analysis of quantum algorithms which query random, invertible permutations has been a long-standing challenge in cryptography. Many techniques which apply to random oracles fail, or are not known to generalize to this setting. As a result, foundational cryptographic constructions involving permutations often lack quantum security proofs. With the aim of closing this gap, we develop and prove soundness of a compressed permutation oracle. Our construction shares many of the attractive features of Zhandry's original compressed function oracle: the purification is a small list of input-output pairs which meaningfully reflect an algorithm's knowledge of the oracle.
    
    We then apply this framework to show that the Feistel construction with seven rounds is a strong quantum PRP, resolving an open question of (Zhandry, 2012). We further re-prove essentially all known quantum query lower bounds in the random permutation model, notably the collision and preimage resistance of both Sponge and Davies-Meyer, hardness of double-sided zero search and sparse predicate search, and give new lower bounds for cycle finding and the one-more problem.
\end{abstract}

\pagebreak
{
  \hypersetup{linkcolor=black}
    \tableofcontents
}

\pagebreak

\section{Introduction}

Quantum algorithms which query random permutations and their inverse are ubiquitous in quantum and post-quantum cryptography~\cite{zhandry2016notequantumsecureprps,rosmanis2022tight,Zhandry21,cojocaru25lifting,alagic2023twosided,alagic2025sponge}. This model is commonly referred to as the quantum random permutation model, or QRPM. Like the related quantum random oracle model (the QROM), this model is an idealization. Where the QROM is an idealization of hash functions, the QRPM is an idealization of cryptographic permutations and ciphers.

Many tools exist for analyzing the QROM \cite{Zhandry12,10.1007/978-3-642-25385-0_3,unruh14revocable,YZ21}. One of the most powerful is the compressed oracle \cite{Zhandry2018}, which analyzes the purification of a random oracle in a specially chosen basis that has an intuitive interpretation. In particular, the purification can be written as a \emph{database} of input-output pairs which meaningfully correspond to points the adversary has knowledge of. Moreover, each query adds at most a single such pair to the database. The introduction of this tool led to a slew of new results in quantum and post-quantum cryptography \cite{don2021extract,chung2021compressed}, as well as query complexity \cite{liu19multicols,Hamoudi_2023}.

The QRPM has proven more challenging to analyze. For instance, the compressed oracle technique has resisted many attempts at generalizations to permutation oracles~\cite{Unruh2021,Unruh2023,rosmanis2022tight,czajkowski19lazy}---or any correlated oracle distribution.
Even basic search problems remained open until recently \cite{carolan24oneway,hosoyamanda18dm,Zhandry21,Unruh2021,MMW24}, and despite recent progress on permutation based hash functions~\cite{cpz24precomp,alagic2025sponge,cojocaru25lifting}, many important constructions still lack quantum security proofs.

An emblematic example is the Feistel cipher, also known as the Luby-Rackoff cipher or Feistel network, which is a standard construction of a psuedorandom permutation from a psuedorandom function~\cite{LR88}. A number of influential block ciphers, such as DES~\cite{fips46} and Blowfish~\cite{blowfish}, are based on this construction.
For purposes of illustration, suppose that $h, k, f$ are pseudorandom functions from $n$ bits to $n$ bits. Then, the Feistel network constructs a $2n$ bit permutation $\varphi$ as the following sequence of operations. To compute $\varphi(x)$: \begin{enumerate}
    \item Take the first $n$ bits of $x$, referred to as $x_L$, and compute $h(x_L)$. XOR this value into the last $n$ bits of $x$, referred to as $x_R$, obtaining $x'=x_L \Vert (x_R \oplus h(x_L))$.
    \item Take the last $n$ bits of $x'$, and compute $k(x_R')$. XOR the result into $x'_L$, obtaining $x'' = (x_L' \oplus k(x_R')) \Vert x_R'$.
    \item Take the first $n$ bits of $x''$, and compute $f(x_L'')$. XOR the result into $x''_R$, obtaining $y=x_L'' \Vert (x_R'' \oplus f(x_L''))$, where $y=\varphi(x)$.
\end{enumerate}

\begin{figure}
    \centering
    \includegraphics[width=0.5\linewidth]{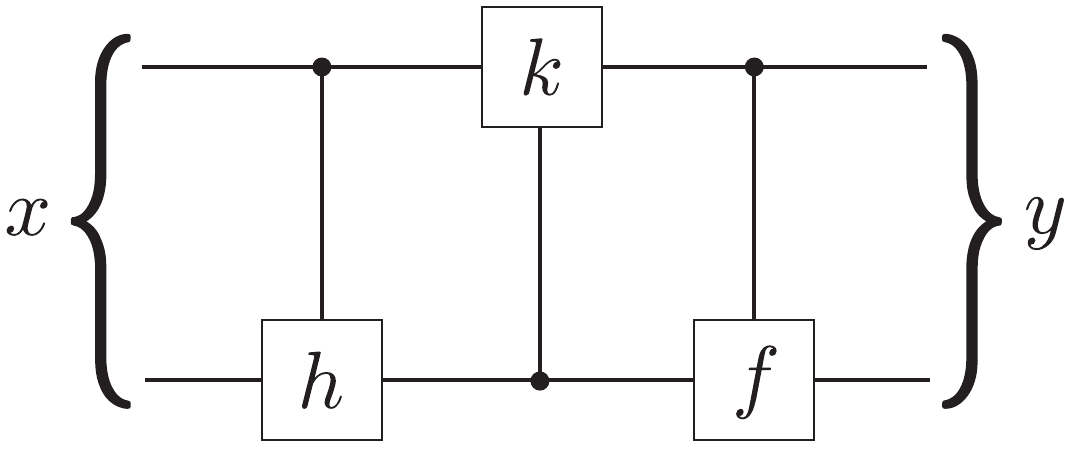}
    \caption{An illustration of three round Feistel.}
    \label{fig:feistel}
\end{figure}

This procedure is depicted in \Cref{fig:feistel}. The three functions corresponds to three rounds, and the construction can be extended with more functions in alternating direction for any number of rounds. It is known that the three round construction is indistinguishable from a random permutation with only forwards classical queries, and that four rounds are indistinguishable with forwards and inverse classical queries~\cite{LR88}. With quantum queries---the strongest plausible attack model---essentially nothing is known about security. This question was first raised by Zhandry in 2012~\cite{Zhandry12}, yet remains beyond existing techniques.\footnote{For instance, in a 2025 revision to \cite{zhandry2016notequantumsecureprps}, Zhandry says ``The quantum security of Feistel networks remains a largely open, very challenging problem.''}

Hosoyamada and Iwata~\cite{hosoyamada19lr} claimed a proof that four round Feistel is a quantum psuedorandom permutation or qPRP\footnote{Without the qualifier of strong, this means forwards-only secure.}, but their proof was found to contain a bug~\cite{bhaumik24badnorm}. The same paper which found the bug proved security against non-adaptive, forwards only queries.
On the algorithmic side, it is known that there is a quantum distinguishing attack for three rounds with forwards only queries~\cite{KM10}, and four rounds with both forwards and backwards queries~\cite{ito19feistcca}. These results indicate that classical proof techniques for Feistel security must break down.

\subsection{Our results}

Our main result is a compressed oracle technique for analyzing two-way accessible permutations. Our construction is simple to apply, and shares many features of Zhandry's original compressed function oracle~\cite{Zhandry2018}. In particular, our construction has a purifying register which is a list of input-output pairs, referred to as the compressed database. We establish that the construction is
\begin{enumerate}
    \item \emph{Sound}, in that it faithfully represents algorithms querying two-way accessible random permutations.
    \item \emph{Bounded}, in that the size of the compressed database grows by at most one with each query.
    \item \emph{Meaningful}, in that the points defined by the database correspond to points the algorithm ``knows''.
\end{enumerate}

We defer to \Cref{sec:comp-perms} for a formal description of both the construction and the above properties---note that the construction itself is only a few pages. 
At a high level, the construction is similar to the compressed function oracle but with two changes. First, the compression operator explicitly maintains injectivity of the database. Second, to answer inverse queries, we invert the database before and after answering.

The idea of inverting the truth table to answer inverse queries is due to Unruh~\cite{Unruh2023}; our construction is a small modification to Unruh's proposal.\footnote{In \cite{Unruh2023}, however, soundness is only conjectured.} The proof that this simple construction is sound constitutes the main technical component of this work.

\begin{theorem*}[Informal version of \Cref{thm:main-indist}]
    No algorithm can distinguish the compressed permutation oracle described above from a truly random size $N$ permutation unless it makes $\Omega(\sqrt[12]{N})$ bidirectional quantum queries.
\end{theorem*}

While our bound is exponential in the size of an input to the permutation, it does not seem tight.
Once the construction is in place, we apply it to a number of problems in the QRPM. The main application is a proof that the seven round Feistel construction is a strong quantum PRP. 

\begin{theorem*}[Informal version of \Cref{thm:feist-indist}]
    No algorithm can distinguish the seven round Feistel construction from a truly random size $N$ permutation unless it makes $\Omega(\sqrt[12]{N})$ bidirectional quantum queries.
\end{theorem*}

Additionally, we give a generic quantum lower bound for any predicate search problem in the QRPM in terms of the sparsity of the predicate. We defer to \Cref{sec:search} for a formal statement, but note that this result is analogous to search lower bounds proved in the QROM by standard compressed oracle methods~\cite{Zhandry2018,chung2021compressed}. Given our loose soundness bound, the result is not always tight.

However, essentially all known search lower bounds in the QRPM \cite{hosoyamanda18dm,cojocaru25lifting,carolan24oneway,cpz24precomp,alagic2025sponge,MMW24} can be seen as a special case of this theorem. We re-prove the preimage and collision resistance of the sponge hash function~\cite{BDPvA07,KeccakSub3,KeccakSponge3}, which underlies the current international standard SHA3, and the Davies-Meyer construction~\cite{winternetz84dm}, which underlies SHA1, SHA2, and MD5. Previous proofs are due to \cite{alagic2025sponge} and \cite{cojocaru25lifting}, both of which required the development of new techniques. For the sponge, our bound is tighter than the one in~\cite{alagic2025sponge}. We also prove the first lower bounds for the one-more problem and cycle finding in the QRPM, both of which appear---to the author's knowledge---beyond the scope of existing techniques.

Intriguingly, we observe that the looseness in our bounds is entirely due to the soundness error of our construction. If the soundness analysis were improved, then all of the previously mentioned lower bounds would be tight.\footnote{With the possible exception of the bound for the one-more and cycle finding problem.} This suggests that our compressed permutation oracle accurately captures the difficulty of these problems.

Given the range of applications and simplicity of our construction, we expect it to be useful for many other problems in the QRPM.

\subsection{Technical summary}
To explain the intuition underlying our security proof, let us begin by considering an approach which does not appear to work. A natural starting point for analyzing quantum algorithms which query invertible permutations is to purify the permutation oracle, mirroring the approach of Zhandry for the compressed function oracle. If we represent a permutation/function from $[N]$ to $[N]$ by its truth table, we can contrast the two purifications: \begin{align}
    \underbrace{\frac{1}{\sqrt{N!}} \sum_{\substack{y_1, \dots, y_N, \\
    y_i \neq y_j}} \ket{y_1} \otimes \dots \otimes \ket{y_N}}_{\text{permutation}}, && \text{ vs. } && \underbrace{\frac{1}{\sqrt{N^N}} \sum_{y_1, \dots, y_N} \ket{y_1} \otimes \dots \otimes \ket{y_N}}_{\text{function}}.
\end{align}
The two states look quite similar, but there is a subtle difference. The purification of random functions, in this representation, is a product state: each entry of the truth table is unentangled, reflecting the fact that the oracle outputs are uncorrelated. Then, a simple indicator that a query algorithm has learned about, say, $y_i$ is to check whether the $i$-th register above remains in uniform superposition (no knowledge), or is entangled with an external register (knowledge gained). This observation is what enables standard compressed oracles.

In the permutation case, there is weak entanglement between registers $i$ and $j$, reflecting the condition that $y_i \neq y_j$. While this might appear like a minor difference, it turns out that the intuition above completely breaks. In particular, the $i$-th register in the purification of a permutation is \emph{maximally} entangled with the remaining registers. For every $j \neq i$, $y_j$ eliminates one of the $[N]$ possibilities for $y_i$. Taken together, the complement of $y_i$ suffices to fix $y_i$. To check if $y_i$ is queried, we must do more than check if register $i$ is entangled, because it is maximally entangled before any queries are made.

To get around this, we must devise some other way to understand the entanglement between a query algorithm and the purification. One approach taken by Unruh~\cite{Unruh2021} and Rosmanis~\cite{rosmanis2022tight} is to identify an injective database (a.k.a. an injective partial function) $I$ with the sum of permutations agreeing with $I$ on all well-defined points. Unfortunately, this leads to compressed database states which are non-orthogonal. This significantly complicates the analysis, especially in the case of inverse queries, limiting the scope of applications to simple search problems.

We will take a different approach based on expanding the purifying register. The point of this expansion is that we will be able to define \emph{orthogonal} database states, which nonetheless characterize entanglement between the algorithm and oracle in an intuitive way. In more detail, we will write the full permutation $\varphi$ as a three round Feistel network defined by functions labeled $h, k, f$, sandwiched on either side by random permutations $\pi$ and $\omega$. We refer to this ensemble as \emph{masked Feistel}, and it is depicted in \Cref{fig:masked-feist}.

\begin{figure}
    \centering
    \includegraphics[width=0.72\linewidth]{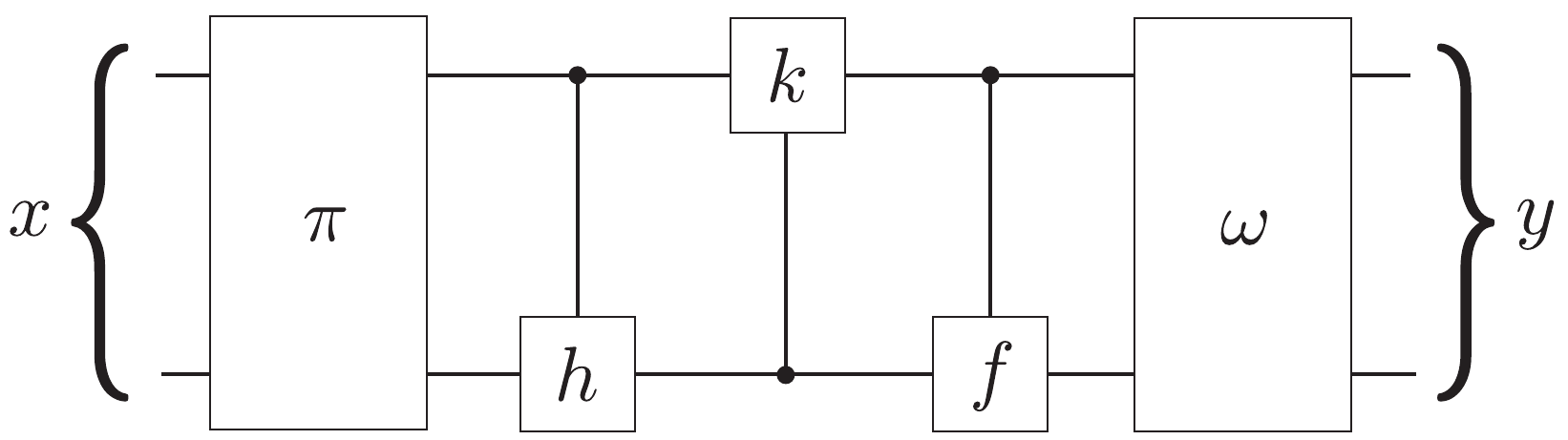}
    \caption{The masked Feistel decomposition.}
    \label{fig:masked-feist}
\end{figure}

We can write the purification of this ensemble, using $N=4^n$ as the permutation domain size, as
\begin{align}
    \frac{1}{\sqrt{(N!)^2 (2^n)^{3 \cdot 2^n}}}\sum_{\pi, \omega \in \mathbf S_N, h,k,f \in \mathbf F_{\sqrt N, \sqrt N}} \ket{\pi, \omega, h, k, f}
\end{align}
where $\mathbf F_{\sqrt N, \sqrt N}$ is the set of functions from $\bit^n$ to $\bit^n$ and $\mathbf S_N$ is the set of permutations on $\bit^{2n}$. We denote the permutation defined by this whole ensemble as $\varphi$. This $\varphi$ is a product of a uniform random permutation (say $\pi$) and some independent permutation, and is therefore uniform random. At this point, it may appear that we have made backwards progress, given the complexity of this purification. However, let us ignore for a moment the outer $\pi, \omega$ permutations, and consider what it means for $\varphi$ to satisfy $\varphi(x)=y$. This fixes three input-output pairs in the $h, k, f$ functions, corresponding to the three values reached along the evaluation of $\varphi(x)$---we call these three pairs a chain. In other words, input-output properties of $\varphi$ correspond in a direct way to input-output properties of the random functions $h, k, f$.

Now, we apply the standard compressed oracle theory to analyze $h, k, f$. In this picture, the purifying registers for these functions are databases which begin empty. Each new query to the ensemble adds at most one new input-output pair to these databases, which we will use to reconstruct the input-output pairs of $\varphi$. Then, so long as each chain has no collisions (i.e. the same function input-output pair appearing in multiple chains), we can uniquely recover an injective database.

Of course, an adversarial algorithm could easily thwart this plan, by exploiting the periodicity of the three-round Feistel to distinguish it from a random permutation; or even making two queries that collide on the first $n$ bits. This latter problem is the ``trivialization of norm'' faced by prior attempts to analyze quantum security of Feistel~\cite{hosoyamada19lr,bhaumik24badnorm}.
The purpose of the twirling permutations $\pi$ and $\omega$ is to prevent the algorithm from detecting this internal structure. We do not compress these permutation, instead keeping them in uniform superposition. Intuitively, with high probability over the twirl, every new query will lead to a distinct first $n$ bits on the internal Feistel. In turn, with high probability over the selection of the round function values, this will lead to a new chain (i.e. sequence of database values defining an input-output pair of $\varphi$) distinct from prior chains. This intuition can be formalized using the standard compressed oracle theory. At a high level, we define an isometry $\iso$ which maps a superposition over tuples $\pi, \omega, D_h, D_k, D_f$, where $D_h, D_k, D_f$ correspond to function databases, to an injective database $I$ . This isometry essentially acts as\footnote{The exact definition of an internal collision can be found in \Cref{sec:Feistel-twirl}, and of the isometry can be found in \Cref{sec:spaces}.} \begin{align}
    \iso \left(\frac{1}{\sqrt{K}}\sum_{\substack{\pi, \omega, D_h, D_k, D_f, \\ \text{supporting only $I$,} \\
    \text{no internal collisions}}} \ket{\pi, \omega, D_h, D_k, D_f}\right) &= \ket{I},
\end{align}
where supporting $I$ means the databases contain information to fix only the input-output pairs in $I$ (and none else), and $K$ is some normalization factor. The key point is that the states on the left hand side are mutually orthogonal, as they are supported on distinct sets of computational basis states.

Then, we show that this isometry \emph{intertwines} the masked Feistel oracle with our compressed permutation oracle. In other words, if $\cmf$ is the compressed masked Feistel oracle and $\cp$ is the compressed permutation oracle, we show $\iso \cdot \cmf \approx \cp \cdot \iso$. Applying this across each query, we show that the final purified permutation (represented here as the tuple $\pi, \omega, h, k, f$) is mapped by $\iso$ to the purification in the compressed permutation oracle experiment (a partial injective database). Then, by Uhlmann's theorem, the mixed state the algorithm sees after tracing out the purifying register is close in trace distance.

We stress that the masked Feistel ensemble is necessary only as a tool for proving soundness, and is not necessary to understand and apply the construction. Indeed, our applications to search lower bounds make no mention of masked Feistel, instead following essentially the same blueprint as compressed function oracle lower bounds. On the other hand, our proof that the seven round Feistel construction is a strong qPRP is quite similar. In particular, we can consider the seven round Feistel to be an instance of masked (three round) Feistel, where the twirling $\pi, \omega$ are drawn from two round Feistel and two round inverse Feistel respectively. We show that this distribution \emph{also} has the requisite statistical properties to hide the internal three rounds, allowing us to reuse many techniques from the soundness proof.

\subsection{Related work}

Prior work on constructing compressed permutation oracles has been met with limited success. Constructions were put forth by Unruh~\cite{Unruh2021} and Czajkowski et al~\cite{czajkowski19lazy}, though both were later found to be incorrect. A notable exception is the recent work of Majenz, Malavolta and Walter~\cite{MMW24}, referred to as the superposition permutation oracle. This construction was successfully employed to show search hardness for sparse relations, and is based on compressed oracle ideas. Unfortunately, this technique differs from standard compressed oracles in that the purification is not a small list of input-output pairs. This added complexity has limited the application of this technique to simple search problems involving only a single input-output pair---in contrast, our search lower bounds apply for predicates on any number of input-output pairs.

In the case where inverse queries are not allowed, a proposal by Rosmanis~\cite{rosmanis2022tight} based on representation theory has been used to re-prove lower bounds for permutation inversion. Furthermore, a number of results from the QROM can be lifted to this setting~\cite{Czajkowski2017,czajkowski19lazy}, as a random one-way permutation is quantum indistinguishable from a random function~\cite{Zhandry15,yuen2013quantumlowerbounddistinguishing}.

\subsection{Reader's guide}

This manuscript is written in order of mathematical dependence. \Cref{sec:prelim,sec:comp-perms} lay out the core ideas and notation, and should be read first. It may be helpful to read \Cref{sec:search}, especially \Cref{sec:search-simple}, afterwards to build intuition about our construction, as the search applications are relatively simple. The application to Feistel in \Cref{sec:feist-sec} should be read after the soundness proof.

For the soundness proof, \Cref{subsec:pieces-main} and the proof of the main theorem provides a high level overview of the pieces, and uses only notation from \Cref{sec:query-pperms,sec:spaces}.  \Cref{sec:Feistel-twirl} defines the purification of permutation oracles whose analysis forms the core of this work, though it is quite technical, and many of the definitions and lemmas may seem unmotivated if read immediately after \Cref{sec:comp-perms}. \Cref{sec:soundness} relates this purification to our compressed permutation oracle.

\subsection{Acknowledgements}
The author is grateful to Gorjan Alagic, Christian Majenz, and Saliha Tokat for insightful discussions and significant feedback on this manuscript.
JC acknowledges support from the U.S. Department of Energy grant DE-SC0020264.

\section{Preliminaries}
\label{sec:prelim}

Throughout this document, we will assume $N=4^n$ is a power of four, and identify $[N]$ with bit strings in the natural way. The choice of $4^n$ reflects that we are constructing a $2n$ bit permutation from $n$ bit to $n$ bit functions.

\subsection{Quantum}

Given a pure state $\ket{\psi} \in \algo H$, we denote the corresponding density matrix by $\rho(\ket{\psi}) = \outerprod{\psi}{\psi}$. We recall the trace norm of a density matrix, $\norm{\rho}_1 = \Tr[\sqrt{\rho\rho^\dagger}]$. The corresponding trace distance of two mixed states upper bounds the probability that they can be distinguished by any algorithm.

\begin{lemma}
    Let $\sigma, \tau$ be density matrices over $\algo H$. Then, given either $\sigma$ or $\tau$ with equal probability, the maximum distinguishing advantage $D(\sigma, \tau)$ is given by \begin{align}
        D(\sigma, \tau) &= \frac{1}{2} \norm{\sigma - \tau}_1.
    \end{align}
    \label{lem:close-states-dis}
\end{lemma}

For a tensor product Hilbert space, we use capital sans serif subscripts to label each factor, e.g.  $\ket{\psi}_{\reg{AB}} \in \algo H_{\reg{AB}} = \algo H_{\reg{A}} \otimes \algo H_{\reg{B}}$. We also refer to the state as having a register $\reg A$ and a register $\reg B$. When an operator $O_{\reg{A}}$ acts on $\algo H_{\reg A}$, we can naturally extend it's action to $\algo H_{\reg{AB}}$ as $O_{\reg{A}} \otimes \Id_{\reg B}$: we will drop this tensor product for convenience, writing $O_{\reg{A}} \ket{\psi}_{\reg{AB}}$ in place of $O_{\reg{A}} \otimes \Id_{\reg B} \ket{\psi}_{\reg{AB}}$.

We will show that two oracles are indistinguishable by showing that the possible views of an adversary making queries to one of the two oracles are described by density matrices which are close in trace distance. Uhlmann's theorem allows us to bound the trace distance of two density matrices by constructing an isometry on their purifications.

\begin{lemma}
    Let $\ket{\psi} \in \algo H_{\reg A} \otimes \algo H_{\reg B}$ and $\ket{\phi} \in \algo H_{\reg A} \otimes \algo H_{\reg C}$ be normalized quantum states, and let $V : \algo H_{\reg C} \ra \algo H_{\reg B}$ be an isometry. Let $\sigma_{\reg A}=\Tr_{\reg B}[\rho(\ket{\psi}_{\reg{AB}})]$ and $\tau_{\reg A}=\Tr_{\reg C}[\rho(\ket{\phi}_{\reg{AC}})]$. Then we have \begin{align}
        \frac{1}{2}\norm{\sigma_{\reg A} - \tau_{\reg A}}_1 &\leq \norm{\ket{\psi}_{\reg{AB}} - V_{\reg C} \ket{\phi}_{\reg{AC}}}.
    \end{align}
    \label{lem:approx-uhlman}
\end{lemma}

\begin{proof}
    We have \begin{align}
        \norm{\sigma_{\reg A} - \tau_{\reg A}}_1 &= \norm{\Tr_{\reg B}[\rho(\ket{\psi}_{\reg{AB}})] - \Tr_{\reg C}[\rho(\ket{\phi}_{\reg{AC}})]}_1\\
        &= \norm{\Tr_{\reg B}[\rho(\ket{\psi}_{\reg{AB}})] - \Tr_{\reg B}[\rho(V_{\reg C} \ket{\phi}_{\reg{AC}})]}_1 & \text{($V$ an isometry)} \\
        &\leq \norm{\rho(\ket{\psi}_{\reg{AB}}) - \rho(V_{\reg C} \ket{\phi}_{\reg{AC}})}_1 & \text{($L_1$ trace monotonicity)} \\
        &\leq 2\norm{\ket{\psi}_{\reg{AB}} - V_{\reg C} \ket{\phi}_{\reg{AC}}} & \text{($L_1$ and $L_2$ relation)}
    \end{align}
\end{proof}

Another simple and useful fact is that the norm of an operator given by a direct sum is the maximum of the norm of any of the terms. If an operator $O : \algo \algo H(A) \otimes \algo H(B) \ra \algo H(A) \otimes \algo H(B)$ can be written $O_{\reg{AB}} = \bigoplus_{a \in A} \proj{a}_{\reg A} \otimes O_{a, \reg B}$ for some family of operators $O_a : \algo H(B) \ra \algo H(B)$, we will sometimes say that it is \emph{controlled} on register $\reg A$.

\begin{lemma}
    Let $\algo H = \algo H_1 \oplus \dots \oplus \algo H_k$ be a Hillbert space, and $O : \algo H \ra \algo H$ be an operator of the form $O=O_1 \oplus \dots \oplus O_k$ where $O_i : \algo H_i \ra \algo H_i$. Then we have \begin{align}
        \norm{O} &= \max_{i \in [k]} \norm{O_i}.
    \end{align}
    \label{lem:dirsum-norm}
\end{lemma}

A projector $\Pi : \algo H \ra \algo H$ is an idempotent operator, where $\Pi^2 = \Pi$. Projectors are orthogonal if they are Hermitian. We will only work with orthogonal projectors. We sometimes say $\ket{\psi} \in \Pi$ to mean that $\Pi \ket{\psi} = \ket{\psi}$.

\subsection{Notation and definitions}
\paragraph{Databases.} We use $\mathbf F_{M, N}$ to denote the set of functions from $[M]$ to $[N]$. We use $\mathbf S_N$ to denote the set of permutations of $[N]$.
We use $\mathbf D_{M, N}$ to denote the set of partial functions from $[M]$ to $[N]$, or equivalently the set of functions from $[M]$ to $[N] \cup \{\bot\}$. We use $\mathbf I_{M, N}$ to denote set of partial \emph{injective} functions from $[M]$ to $[N]$, or equivalently the set of functions from $[M]$ to $[N] \cup \{\bot\}$ which are injective on all elements of the range except $\bot$.\footnote{As in, there may be multiple pre-images of $\bot$, but at most one preimage of any $y \in [N]$.} 
We will sometimes drop the subscript when it is clear from context. For our purposes, we will almost always work with $\mathbf I_{N, N}$ and $\mathbf D_{\sqrt{N}, \sqrt{N}}$, where we suppose $N=4^n$ is a power of four. We take total functions to be a special case of partial functions, meaning \begin{align}
    \mathbf F_{M, N} &\subset \mathbf D_{M, N} & \mathbf S_N &\subset \mathbf I_{N}.
\end{align}

We will interchange between the partial and total function representations, and refer to such objects as \emph{databases}, which may or may not be injective. For a $D \in \mathbf D_{M, N}$, we use $\dom(D)$ to denote the domain of $D$ as a partial function, or equivalently the preimages of $[N]$ as a total function. The image $\im(D)$ refers to the image of $D$ as a partial function, and similarly define $\dom(I)$ and $\im(I)$ for $I \in \mathbf I_{M, N}$. The size of a database can be defined as $|D|=|\dom(D)|$, and similarly for injective databases $|I|=|\dom(I)|$. When $M, N$ are understood from context, we write $\mathbf D_t$ to denote the set of $D \in \mathbf D$ of size $|D|=t$, and similarly $\mathbf D_{\leq t}$ to denote the set of $D \in \mathbf D$ of size $|D|\leq t$. We similarly define $\mathbf I_t$ and $\mathbf I_{\leq t}$.

We use $\botd$ to denote the trivial partial function, or equivalently the total function which maps all inputs to $\bot$.
Given a $D \in \mathbf D_{M, N}$, we use $D[x \ra y]$ to denote the the function \begin{align}
    D[x \ra y](x') &= \begin{cases}
        D(x') & \text{(If $x \neq x'$)} \\
        y & \text{(Otherwise)}
    \end{cases}
\end{align}
We similarly define $I[x \ra y]$, with the subtlety that we may need to discard a conflicting element from the domain of $I$ to ensure injectivity,  \begin{align}
    I[x \ra y](x') &= \begin{cases}
        I(x') & \text{(If $x \neq x'$, and either $y =\bot$ or $I(x') \neq y$)} \\
        \bot & \text{(If $x \neq x'$, $y \neq \bot$, and $I(x')=y$)} \\
        y & \text{(If $x = x'$)}
    \end{cases}
\end{align}
We will almost always use this notation when $x$ is not already in the domain, and in the case of injective functions where $y$ is not already in the range. We can identify databases with relations as well, and write $(x, y) \in D$ when $D(x)=y$, and similarly for $(x, y) \in I$.

Given a $D \in \mathbf D$, we define $D|^x$ as the set of databases which agree with $D$ on all inputs except $x$, in particular \begin{align}
    D|^x &\coloneqq \{D[x \ra y] \, : \, y \in [N] \cup \{\bot\}\}.
\end{align}

We let $\mathbf D |^x$ be the set of databases which are undefined on $x$, i.e., \begin{align}
    \mathbf D|^x &\coloneqq \{D \in \mathbf D \, : \, x\not\in \dom(D)\}\,.
\end{align}
Note that $\{ D|^x \, : \, D \in \mathbf D|^x\}$ is a partition of $\mathbf D$ for any fixed $x$. We define similar notions for injective databases $I \in \mathbf I$, once again taking care to maintain injectivity in the natural way. \begin{align}
    I|^x &\coloneqq \{I[x \ra y] \, : \, y \in ([N] \setminus \im(I)) \cup \{I(x)\} \cup \{\bot\}\} \\
    \mathbf I|^x &\coloneqq \{I \in \mathbf I \, : \, x\not\in \dom(I)\},
\end{align}
where once again $\{I|^x \, : \, I \in \mathbf I|^x\}$ is a partition of $\mathbf I$ for any fixed $x$. 
Given a set $S$, we denote a uniform random element of $S$ by $x \sim S$. Given a string $x \in \bit^{2n}$, we use $x_L$ to refer to the first $n$ bits and $x_R$ to refer to the last $n$ bits.

\paragraph{Spaces.} Given a set $S$, we denote by $\algo H(S)$ the Hilbert space spanned by an orthogonal basis labeled by elements of $S$, i.e., \begin{align}
    \algo H(S) &\coloneqq \spn{\ket{x} \, : \, x \in S}, & \innerprod{x}{x'} &= \delta_{x, x'}.
\end{align}
In such a space, for a set $A \subset S$ we denote the uniform superposition of elements in $A$ by \begin{align}
    \ket{+_{A}} &\coloneqq \frac{1}{\sqrt{|A|}} \sum_{x \in A} \ket{x}.
\end{align}
Given sets $A_1, \dots, A_l \subset S$ such that $\bigcup_{i=1}^{l} A_i = S$ and $A_i \cap A_j = \emptyset$ for all $i \neq j$, we say that $\{A_1, \dots, A_l\}$ is a partition of $S$. In that case, we may express $\algo H(S)$ as the direct sum $\algo H(A_1) \oplus \dots \oplus \algo H(A_l)$.

We will often deal with the Hilbert spaces $\algo H(\mathbf D)$ and $\algo H(\mathbf I) \leq \algo H(\mathbf D)$. We define the projector $\pit$ and $\pilt$ on the former space as \begin{align}
    \pit &\coloneqq \sum_{D \in \mathbf D_t} \outerprod{D}{D}, & \pilt &\coloneqq \sum_{k=0}^t \Pi_k.
\end{align}
We can define these operators on $\algo H(\mathbf I)$ by their restriction to this subspace. 

A norm without a subscript, $\norm{\cdot}$, refers to the standard $L_2$ norm. We use $\norm{\cdot}_1$ for the trace norm. For the Hilbert space $\algo H_{\reg A} \otimes \algo H(\mathbf D)$ or $\algo H_{\reg A} \otimes \algo H(\mathbf I)$, for any Hilbert space $\algo H_{\reg A}$, we use the notation $\norm{\cdot}_{\leq t}$ for the standard $L_2$ norm restricted to databases of size at most $t$. In other words, for a state $\ket{\psi}$ or operator $O$ on $\algo H_{\reg{AI}} = \algo H_{\reg A} \otimes \algo H(\mathbf I)$, we define \begin{align}
    \norm{\ket{\psi}_{\reg{AI}}}_{\leq t} &\coloneqq \norm{\pilt_{\reg I} \cdot \ket{\psi}_{\reg{AI}}}, & \norm{O_{\reg{AI}}}_{\leq t} &\coloneqq \norm{O_{\reg{AI}} \cdot \pilt_{\reg I}}.
\end{align}

We also define consistency projectors for databases, which assert that a given input-output pair appears in the database. For databases over $\algo H(\mathbf F_{M, N})$ and an input-output pair $(x, y) \in [M] \times [N]$, we define \begin{align}
    \xinD{(x, y)} &\coloneqq \sum_{D \in \mathbf D \, : \, D(x)=y} \proj{D},
\end{align}
and similarly for the space of injective databases.

\subsection{Compressed oracles}
\label{sec:comp-fs}
In this section, we briefly recall the standard compressed oracle due to \cite{Zhandry2018}. Our presentation diverges from the original to more closely match the presentation of compressed permutation oracles. Throughout, we assume $M$ and $N$ are powers of two, and identify elements of $[M]$ and $[N]$ with bitstrings by the natural representation.

Let $f \sim \mathbf F_{M, N}$ be a uniform random function, and consider the oracle \begin{align}
    \algo O_f \ket{x, y}_{\reg{XY}} &\coloneqq \ket{x, y \oplus f(x)}_{\reg{XY}}. \label{eqn:std-oracle}
\end{align}
Let $\algo A^{\algo O_f}$ be a quantum algorithm which makes $q$ queries to $\algo O_f$. This algorithm is described by a Hilbert space $\algo H_{\reg A}$ for its internal state and query registers, as well as a sequence $A_0, \dots, A_q$ of unitary operators on $\algo H_{\reg A}$ describing the action in between each query. We refer to registers $\reg X, \reg Y$ as the query registers, corresponding to the two registers in \Cref{eqn:std-oracle}. For simplicity, we consider these registers to be inside of register $\reg A$. We can write the final mixed state of the algorithm as \begin{align}
    \rho_{\algo A}^{(\algo O_f)} &= \E_{f \sim \mathbf F_{M, N}} \left[\rho(A_q \algo O_f \dots A_1 \algo O_f A_0 \ket{0}) \right].
\end{align}

Alternatively, we could purify this experiment, defining a register $\algo H_{\reg F} = \algo H(\mathbf F)$ to hold the purified function, and replacing the standard oracle with the purified oracle \begin{align}
    \pu_{\reg{XYF}} \ket{x, y}_{\reg{XY}} \ket{f}_{\reg F} &\coloneqq \ket{x, y \oplus f(x)}_{\reg{XY}} \ket{f}_{\reg F}.
\end{align}
Considering the purified experiment over $\algo H_{\reg{AF}}$ and tracing out $\reg F$ results in the same view, \begin{align}
    \rho_{\algo A}^{(\algo O_f)} &= \Tr_{\reg F} \left(\rho\left(A_{q, \reg A} \pu_{\reg{AF}} \dots A_{1, \reg A} \pu_{\reg{AF}} A_{0, \reg A} \ket{0}_{\reg A} \ket{+_{\mathbf F}}_{\reg F}\right)\right).
\end{align}

\paragraph{The compressed picture.}

The compressed oracle provides an alternative, powerful picture of quantum accessible random functions. In this picture, the purifying register is a database which begins empty. Each query modifies only the query point in the database, somewhat analogous to classical lazy sampling.

To define this formally, we will first consider the subspace $\algo H(D|^x) \leq \algo H(\mathbf D)$ for some $D \in \mathbf D = \mathbf D_{M,N}$ and $x \in [M] \setminus \dom(D)$. This subspace captures modifying the database on only the $x$-th position.
Let us write $\ket{+_{x, D}} = \frac{1}{\sqrt{N}} \sum_{y \in [N]} \ket{D[x\ra y]}$. We can now define the \emph{function compression} operator $\fc$ on this space, indexed by $x$ and $D$, as \begin{align}
    \fc_{x, D} &\coloneqq \Id - \outerprod{+_{x, D}}{+_{x, D}} - \outerprod{D}{D} + \outerprod{+_{x, D}}{D} + \outerprod{D}{+_{x, D}}.
\end{align}
In other words, this operator swaps the uniform superposition of outputs with the undefined output, for the input $x$. Recalling that $\{D|^x \, : \, D \in \mathbf D|^x\}$ is a partition of $\mathbf D$, we may then define the full compression operator and its controlled analog as \begin{align}
    \fc_x & \coloneqq \bigoplus_{D \in \mathbf D|^x} \fc_{x, D}, & 
    \fc_{\reg{XD}} \ket{x}_{\reg X} = \ket{x}_{\reg X} \otimes \fc_{x, \reg D}.
\end{align} 
The final piece we will need is the extended purified oracle $\pu$, which acts as \begin{align}
    \pu \ket{x, y} \ket{D} &\coloneqq \begin{cases}
        \ket{x, y \oplus D(x)} \ket{D} &\text{(If $x \in \dom(D)$)} \\
        \ket{x, y} \ket{D} &\text{(Otherwise)}
    \end{cases}
\end{align}

Now, we define the full compressed function oracle as the purified oracle conjugated by the compression operator, i.e.,
\begin{align}
    \cf_{\reg{XYP}} &\coloneqq \fc_{\reg{XP}} \cdot \pu_{\reg{XYP}}\cdot  \fc_{\reg{XP}}^\dagger,
\end{align}
noting that $\fc$ is an involution so the dagger is strictly unnecessary.
The compressed oracle experiment is somewhat similar to the purified experiment, except we allow the purifying register to include partial functions, $\algo H_{\reg D} = \algo H(\mathbf D)$. We refer to this register as the \emph{compressed oracle}, and it is initialized to $\ket{\pmb \bot}_{\reg D}$. We additionally replace purified oracle calls with the compressed oracle $\cf$. This corresponds to the final state \begin{align}
    \ket{\psi_{\algo A}^{(\cf)}}_{\reg{AD}} \coloneqq A_{q, \reg A} \cf_{\reg{AD}} \dots A_{1, \reg A} \cf_{\reg{AD}} A_{0, \reg A} \ket{0}_{\reg A} \ket{\botd}_{\reg D}.
\end{align}

The following facts were shown by \cite{Zhandry2018}, and are core to the theory. At a high level, compressed oracles are: \begin{enumerate}
    \item \emph{Sound}, in that they faithfully represent algorithms querying random functions.
    \item \emph{Bounded}, in that the size of the compressed database grows by at most one with each query.
    \item \emph{Meaningful}, in that the points defined by the database correspond to points the algorithm ``knows''.
\end{enumerate}

These key properties are formalized in the following lemmas.

\begin{lemma}[Soundness]
    The view of an algorithm querying the compressed oracle is the same as an algorithm querying a truly random function,\begin{align}
        \rho_{\algo A, \reg A}^{(\algo O_f)} &= \Tr_{\reg D} \left[\outerprod{\psi_{\algo A}^{(\cf)}}{\psi_{\algo A}^{(\cf)}}_{\reg{AD}}\right].
    \end{align}
    \label{lem:cf-sound}
\end{lemma}

\begin{lemma}[Bounded growth]
    The final state of the compressed oracle after $q$ queries is supported entirely on functions with domain size $q$, \begin{align}
        \ket{\psi_{\algo A}^{(\cf)}}_{\reg{AD}} \in \algo H_{\reg A} \otimes \algo H(\mathbf D_{\leq q}).
    \end{align}
    \label{lem:cf-bounded}
\end{lemma}

To establish the fundamental lemma (that formalizes the third property above), we will interpret the adversary as outputting a tuple of $l$ input-output pairs. We will compare the probabilities that these pairs appear in the database with the probability that they agree with the oracle.\footnote{Sometimes, this lemma is formulated in the context of search games with a target predicate. Our version implies this application as a special case, as shown in \Cref{lem:search-bound}.} To do this, we define the projector \begin{align}
    \Pi^{(\cf)}_{\reg{AD}} &\coloneqq \sum_{\{(x_1, y_1), \dots, (x_l, y_l)\} \in ([M] \times [N])^{\times l}} \left(\xinD{(x_1, y_1)}\dots \xinD{(x_l, y_l)}\right)_{\reg D} \otimes \\&\quad (\outerprod{x_1,y_1,\dots,x_l, y_l}{x_1,y_1,\dots,x_l, y_l} \otimes \Id)_{\reg A},
\end{align}

which asserts that the output of the adversary corresponds to input-output pairs which agree with the database.
We will also require the analogous operator that decompresses before each check, \begin{align}
    \Pi^{(\algo O_f)}_{\reg{AD}} &\coloneqq \sum_{\{(x_1, y_1), \dots, (x_l, y_l)\} \in ([M] \times [N])^{\times l}} \left(\fc_{x_1}\xinD{(x_1, y_1)}\fc_{x_1}^\dagger\dots \fc_{x_l}\xinD{(x_l, y_l)}\fc_{x_l}^\dagger\right)_{\reg D} \otimes \\&\quad (\outerprod{x_1,y_1,\dots,x_l, y_l}{x_1,y_1,\dots,x_l, y_l} \otimes \Id)_{\reg A}.
\end{align}
Because the compressed oracle is indistinguishable from a random oracle, this projector is equivalent to checking if the adversary output is consistent with the oracle.

\begin{lemma}[Fundamental lemma]
    The probability that an adversary finds $l$ input-output pairs of the oracle is exponentially close to the probability that the corresponding pairs appears in the compressed database, \begin{align}
        \norm{\Pi^{(\algo O_f)}_{\reg{AD}}\ket{\psi_{\algo A}^{(\cf)}}_{\reg{AD}}} &\leq \norm{\Pi^{(\cf)}_{\reg{AD}}\ket{\psi_{\algo A}^{(\cf)}}_{\reg{AD}}} + \sqrt{\frac{l}{N}}.
    \end{align}
    \label{lem:cf-meaningful}
\end{lemma}

\paragraph{Further properties.}
Another important property of compressed oracles is that they are \emph{valid}, in that decompressing on any given input will cause the truth table to have a well-defined output. A simple example of a valid database is $\ket{\botd}_{\reg D}$. If we consider applying decompression $\fc_x^\dagger$ to this state, we obtain \begin{align}
    \fc_{x, \reg D}^\dagger \ket{\botd}_{\reg D} &= \sum_{y \in [N]} \ket{(x, y)}_{\reg D},
\end{align}
where we use relation notation for the database. It is clear from the above expression that the output of $x$ is well defined. It is known that this property is preserved under compressed oracle queries.

To make this formal, we define the \emph{compressed} and \emph{decompressed} subspaces. The projector onto the former will be denoted by $\Xi$, and the latter by $\Gamma$. We define these operators for a fixed input $x$ as \begin{align}
    \Gamma_x &\coloneqq \sum_{D \in \mathbf D \text{ s.t. } x \in \dom(D)} \outerprod{D}{D}, & \Xi_x &\coloneqq \fc_x \Gamma_x \fc_x^\dagger.
\end{align}
Note that the compression operator interchanges compressed and decompressed subspaces. We define the general operators as the products of these over all values of $x$, \begin{align}
    \Gamma &\coloneqq \prod_{x \in [M]} \Gamma_x, & \Xi &\coloneqq \prod_{x \in [M]} \Xi_x.
\end{align}
Each $\Gamma_x$ and $\fc_x$ commutes with $\Gamma_{x'}$ and $\fc_{x'}$ for $x \neq x'$, so this is a well defined projector.
The full decompressed subspace is simply the projector onto all total functions, from which it is clear that the standard purified oracle maintains this subspace. Analogously, the compressed oracle maintains the compressed subspace. This guarantees that, after applying the $\fc_x$ operator, the compressed oracle will have a well-defined output for $x$.

\begin{lemma}
    For any $\ket{\psi} \in \Xi$, we have $\cf \ket{\psi} \in \Xi$.
    \label{lem:comp-oracle-valid}
\end{lemma}

The following further properties of compressed oracles are not, to our knowledge, shown in the literature. However, they have straightforward proofs, which we include.

\begin{lemma}
    Fix some input $x \in [M]$. For each $D \in \mathbf D|^x$, choose a set $S_{x, D} \subseteq [N]$ such that $|S_{x, D}| \geq N - t$. Let $\ket{S_{x, D}}=\frac{1}{\sqrt{|S_{x, D}|}} \sum_{y \in S_{x, D}} \ket{D[x\ra y]}$. 
    
    Define the operator \begin{align}
        G_{x, D} & \coloneqq \Id - \outerprod{D}{D} - \outerprod{S_{x, D}}{ S_{x, D}} + \outerprod{D}{S_{x, D}} + \outerprod{S_{x, D}}{D}, \\
        G_x & \coloneqq \bigoplus_{D \in \mathbf D|^x} G_{x, D}.
    \end{align}
    Then we have \begin{align}
        \norm{\fc_x - G_x} &= O\left(\sqrt{\frac{t}{N}}\right).
    \end{align}
    \label{lem:comp-no-bad}
\end{lemma}

\begin{proof}
    First, let us fix a database $D$ with $x \not\in D$ and use the shorthand $\ket{+_{x, D}}$ for the even superposition $\sum_{y \in [N]} \frac{1}{\sqrt{N}} D[x\ra y]$. Then we have \begin{align}
        \norm{\ket{+_{x, D}} - \ket{S_{x, D}}}^2 &= \sum_{y \not\in S_{x, D}} \frac{1}{N} + \sum_{y \in S_{x, D}} \left(\frac{1}{\sqrt{N}}-\frac{1}{\sqrt{|S_{x, D}|}}\right)^2 \\
        &\leq \frac{t}{N} + N \cdot \left(\frac{1}{\sqrt{N}}-\frac{1}{\sqrt{N-t}}\right)^2 \\
        &= O\left(\frac{t}{N}\right).
    \end{align}
    
    We may then write \begin{align}
        \fc_{x, D} - G_{x,D} =~&  (\outerprod{+_{x, D}}{+_{x, D}} - \outerprod{S_{x, D}}{ S_{x, D}}) +\\& \ket{D}(\bra{+_{x, D}}-\bra{S_{x, D}}) + (\ket{+_{x, D}}-\ket{S_{x, D}})\bra{D}.
    \end{align}
    By inspecting each term and applying the triangle inequality, one can verify that
    \begin{equation}
        \norm{\fc_{x,D}-G_{x,D}} = O(\sqrt{t/N})\,.    
    \end{equation}
    The claim now follows from \Cref{lem:dirsum-norm}.
\end{proof}

\begin{lemma}
    Fix an input $x\in [N]$ for a compressed database of a random function, and let $\reg R$ denote an ancillary register. Let $\ket{\phi_y}_{\reg R}$ be a family of states indexed by $y \in [N]$ which are pairwise orthogonal, and consider a state of the form \begin{align}
        \ket{\psi}_{\reg{DR}} &= \sum_{D \in \mathbf D, D(x) \neq \bot} \alpha_D \ket{D}_{\reg D} \ket{\phi_{D(x)}}_{\reg R}.
    \end{align}
    
    In other words, there is a record of the output of $x$ outside the database register. Then, for any operator $G_{x, \reg{D}}$ of the form described in \Cref{lem:comp-no-bad}, we have \begin{align}
        \norm{(\Id_{\reg{DR}}-G_{x, \reg{D}}) \ket{\psi}_{\reg{DR}}} &= O\left(\sqrt{\frac{1}{N-t}}\right).
    \end{align}
    \label{lem:recorded-no-comp}
\end{lemma}

\begin{proof}
    We can write \begin{align}
        \norm{(\Id_{\reg{DR}} - G_{x, \reg D}) \ket{\psi}_{\reg{DR}}} &= \norm{\sum_{D \in \mathbf D, D(x) \neq \bot} \frac{\alpha_D}{\sqrt{|S_{x, D[x\ra \bot]}|}} \left( \ket{D[x\ra \bot]} - \ket{S_{x, D[x\ra \bot]}}\right)_{\reg D} \ket{\phi_{D(x)}}_{\reg R}}\nonumber \\
        &= \sqrt{\sum_{D \in \mathbf D, D(x) \neq \bot} \frac{\abs{\alpha_D}^2}{|S_{x, D[x\ra \bot]}|}} \nonumber \\
        &= O\left(\sqrt{\frac{1}{N-t}}\right),
    \end{align}
    where the second line follows from the fact that states in the sum corresponding $D \neq D' \in \mathbf D$ will be orthogonal. To see this, observe that we could define an isometry which combined every input-output pair except $D(x)$ from register $\reg D$, and determined $D(x)$ from register $\reg R$, to construct $D$ in a new register.
\end{proof}
\section{Compressed permutation oracles}
\label{sec:comp-perms}

In this section, we present our compressed permutation oracle framework. The construction itself is similar to the function case, though establishing soundness turns out to be much more difficult. Our construction is similar to that of Unruh \cite{Unruh2023}, though we explicitly and unitarily maintain injectivity of the database. We begin with basic definitions of permutation oracles.

Let $\varphi \sim \mathbf S_{N}$ be a uniform random permutation, and consider the oracle \begin{align}
    \algo O_\varphi \ket{b, x}_{\reg X} \ket{y}_{\reg Y} &\coloneqq \ket{b, x}_{\reg X} \ket{y \oplus \varphi^{1-2b}(x)}_{\reg Y},
\end{align}
for a bit $b$ and $x, y \in [N]$. In other words, $\algo O_\varphi$ is a standard oracle for the permutation $\varphi$ with direction controlled by the bit $b$, on query point $x$. We do not explicitly consider controlled queries (i.e., the possibility of not querying the oracle), but we note that all our proofs apply to this case as well.\footnote{The easiest way to see this is to replace controlled queries with a controlled swap of the output register (holding $y$) with the uniform superposition $\ket{+_{[N]}}$, before and after the query. A standard oracle acts as identity when the output register is in uniform superposition.} We will throughout consider algorithms which may query both $\varphi$ and $\varphi^{-1}$, as $\algo O_{\varphi}$ exposes both directions.

Let $\algo A^{\algo O_\varphi}$ be a quantum algorithm that makes $q$ queries to $\algo O_\varphi$. This algorithm is described by a Hilbert space $\algo H_{\reg A}$ for its internal state, as well as a sequence $A_0, \dots, A_q$ of unitary operators on $\algo H_{\reg A}$, describing the action in between each query. We can write the final mixed state of the algorithm as \begin{align}
    \rho_{\algo A}^{(\algo O_{\varphi})} &= \E_{\varphi \sim \mathbf S_N} \left[\rho(A_q \algo O_\varphi \dots A_1 \algo O_\varphi A_0 \ket{0}) \right].
\end{align}

\paragraph{The compressed picture.}
Similar to the compressed function oracle, we will take a purifying database which begins empty. A query on input $x$ will affect only the database entries involving $x$. In contrast to the function case, however, this modification will be \emph{controlled} on the rest of the database to avoid collisions. This results in an operator which cannot be expressed as a direct product over database entries, which reflects the correlations between permutation outputs. Further, compression operators on different input points will no longer commute.

To formalize the action, fix for now some $I \in \mathbf I$ and $x \in [N] \setminus \dom(I)$, and consider the subspace $\algo H(I|^x) \leq \algo H(\mathbf I)$. This subspace reflects modifying only the $x$-th query point, though in this case we will perform this modification controlled on the rest of the database.
Let us define 
\begin{equation}
    \ket{+_{x, I}} = \frac{1}{\sqrt{N - |I|}} \sum_{y \in [N] \setminus \im(I)} \ket{I[x\ra y]}\,.
\end{equation} 
We can now define the \emph{permutation compression} operator $\pc$ on this subspace, indexed by $x$ and $I$, as \begin{align}
    \pc_{x, I} &\coloneqq \Id - \outerprod{+_{x, I}}{+_{x, I}} - \outerprod{I}{I} + \outerprod{+_{x, I}}{I} + \outerprod{I}{+_{x, I}}.
\end{align}
In other words, this operator swaps the uniform superposition of \emph{non-colliding} outputs with the undefined output, for the input $x$. Recalling that $\{I|^x \, : \, I \in \mathbf I|^x\}$ is a partition of $\mathbf I$, we may then define the full compression operator and its controlled variant as \begin{align}
    \pc_x & \coloneqq \bigoplus_{I \in \mathbf I|^x} \pc_{x, I}, & \pc_{\reg{XI}} \ket{x}_{\reg X}  = \ket{x}_{\reg X} \otimes \pc_{x, \reg I}.    
\end{align}
We recall the extended purified oracle $\pu$, which acts as \begin{align}
    \pu \ket{x, y} \ket{I} &\coloneqq \begin{cases}
        \ket{x, y \oplus I(x)} \ket{I} &\text{(If $x \in \dom(I)$)} \\
        \ket{x, y} \ket{I} &\text{(Otherwise)}
    \end{cases}
\end{align}
The final ingredient we need is the flip operator $\flip$, which inverts the database (viewed as a partial function). In other words, $I^{-1}(y)=x \neq \bot$ if and only if $I(x)=y$, else $I^{-1}(y)=\bot$. Here, the database is required to be injective, and so we can define the flip unitary \begin{align}
    \flip \ket{I} &= \ket{I^{-1}}.
\end{align}
The full compressed permutation oracle is similar to the compressed function oracle, except we maintain injectivity of the database. This allows us to invert before answering inverse queries.
In particular, we define the full compressed permutation oracle as
\begin{align}
    \cp \ket{b} &\coloneqq \ket{b} \otimes \begin{cases}
        \pc \cdot \pu\cdot  \pc^\dagger & \text{(If $b=0$)} \\
        \flip \cdot \pc \cdot \pu\cdot  \pc^\dagger \cdot \flip^\dagger & \text{(If $b=1$)}
    \end{cases}
\end{align}

where again the operators $\pc$ and $\flip$ are involutions, and so the dagger may be suppressed.
We take the purifying register to be $\algo H_{\reg I} = \algo H(\mathbf I)$. We refer to this register as the \emph{compressed permutation oracle}, and it is initialized to $\ket{\boti}_{\reg I}$. We additionally replace oracle calls with the compressed oracle $\cp$. This corresponds to the final state \begin{align}
    \ket{\psi_{\algo A}^{(\cp)}}_{\reg{AI}} = A_{q, \reg A} \cp_{\reg{AI}} \dots A_{1, \reg A} \cp_{\reg{AI}} A_{0, \reg A} \ket{0}_{\reg A} \ket{\boti}_{\reg I}.
\end{align}

Throughout the rest of this document, we will establish that this simple construction satisfies properties analogous to those of the compressed random oracle~\cite{Zhandry2018}. 

\subsection{Properties}
The key properties we will show are that our compressed oracle is:
\begin{enumerate}
    \item \emph{Sound}, in that it faithfully represent algorithms querying two-way accessible random permutations.
    \item \emph{Bounded}, in that the size of the compressed database grows at most linearly with the number of queries.
    \item \emph{Meaningful}, in that the points defined by the database correspond to points the algorithm ``knows''.
\end{enumerate}

It is fairly straightforward to establish bounded growth and the fundamental lemma, points (2) and (3) above, which we do in this section. For now, we only state soundness; proving it will be quite technically involved, and is the focus of the next two sections. Note also that soundness is only approximate.

\begin{lemma}[Soundness]
    The view of an efficient algorithm querying the compressed permutation oracle is negligibly close to the view of the same algorithm querying a truly random permutation, i.e., \begin{align}
        \rho_{\algo A, \reg{A}}^{(\algo O_{\varphi})} & \, \stackrel{\mathclap{\negl}}{\,\approx\,} \, \Tr_{\reg I} [\rho(\ket{\psi_{\algo A}^{(\cp)}}_{\reg{AI}})].
    \end{align}
    \label{lem:pf-sound}
\end{lemma}

\begin{proof}
    Follows from \Cref{thm:main-indist}.
\end{proof}

\begin{lemma}[Bounded growth]
    The final state of the compressed permutation oracle after $q$ queries is supported solely on databases of size $q$, \begin{align}
        \Pi_{\leq q, \reg I} \ket{\psi_{\algo A}^{(\cp)}}_{\reg{AI}} = \ket{\psi_{\algo A}^{(\cp)}}_{\reg{AI}}.
    \end{align}
    \label{lem:pf-bounded}
\end{lemma}

\begin{proof}
    Suppose that the query register $\reg X$ holds the state $\ket{0, x}_{\reg X}$. Then, the $\cp$ operator preserves the subspace $\algo H(I |^x)$ on the purifying $\reg I$ register. It follows that the operator can increase the size of the purifying register by at most $1$. A similar argument follows for the case of an inverse query.
\end{proof}

To establish the fundamental lemma, we will once again interpret the adversary as outputting a tuple of $l$ input-output pairs. We will compare the probabilities that these pairs appear in the database with the probability that they agree with the (compressed) oracle. To do this, we define the projector \begin{align}
    \Pi^{(\cp)}_{\reg{AI}} &\coloneqq \left(\sum_{\{(x_1, y_1), \dots, (x_l, y_l)\} \in [N]^{\times 2l}} \xinD{(x_1, y_1)}\dots \xinD{(x_l, y_l)}\right)_{\reg I} \otimes \\&\quad (\outerprod{x_1,y_1,\dots,x_l, y_l}{x_1,y_1,\dots,x_l, y_l} \otimes \Id)_{\reg A}.
\end{align}

We will also require the the analogous operator after decompressing, \begin{align}
    \Pi^{(\algo O_{\varphi})}_{\reg{AI}} &\coloneqq \sum_{\{(x_1, y_1), \dots, (x_l, y_l)\} \in [N]^{\times 2l}} \left(\pc_{x_1}\xinD{(x_1, y_1)}\pc_{x_1}^\dagger\dots \pc_{x_l}\xinD{(x_l, y_l)}\pc_{x_l}^\dagger\right)_{\reg I} \otimes \\&\quad (\outerprod{x_1,y_1,\dots,x_l, y_l}{x_1,y_1,\dots,x_l, y_l} \otimes \Id)_{\reg A}.
\end{align}

Once soundness of the compressed permutation oracle is established, this check corresponds to verifying that the adversary output is consistent with the oracle.
We can now state and prove the fundamental lemma. Due to the non-commutativity of compression operators, the dependence on $l$ is quadratically worse than in the function case, though this may not be tight.

\begin{lemma}[Fundamental lemma]
    The probability that an adversary finds some list of $l$ input-output pairs of the oracle after $t$ queries is exponentially close to the probability that the pair appears in the compressed database, 
    
    \begin{align}
        \norm{\Pi^{(\algo O_{\varphi})}_{\reg{AI}}\ket{\psi_{\algo A}^{(\cp)}}_{\reg{AI}}} &\leq \norm{\Pi^{(\cp)}_{\reg{AI}}\ket{\psi_{\algo A}^{(\cp)}}_{\reg{AI}}} + \frac{l}{\sqrt{N-t-l}}.
    \end{align}
    \label{lem:pf-meaningful}
\end{lemma}

\begin{proof}
    We have \begin{align}
        \norm{\Pi^{(\algo O_{\varphi})}_{\reg{AI}}\ket{\psi_{\algo A}^{(\cp)}}_{\reg{AI}}} - \norm{\Pi^{(\cp)}_{\reg{AI}}\ket{\psi_{\algo A}^{(\cp)}}_{\reg{AI}}} &\leq \norm{(\Pi^{(\algo O_{\varphi})}_{\reg{AI}} - \Pi^{(\cp)}_{\reg{AI}})\ket{\psi_{\algo A}^{(\cp)}}_{\reg{AI}}} \\
        &\leq \norm{\Pi^{(\algo O_{\varphi})}_{\reg{AI}} - \Pi^{(\cp)}_{\reg{AI}}}_{\leq t}.
    \end{align}
    
    Note that both these operators are controlled on $\reg A$, so by \Cref{lem:dirsum-norm} it suffices to analyze the difference for some fixed set $\mathbf x = \{(x_1, y_1), \dots, (x_l, y_l)\}$. We then obtain \begin{align}
        \norm{\Pi^{(\algo O_{\varphi})}_{\mathbf x, \reg{I}} - \Pi^{(\cp)}_{\mathbf x, \reg{I}}}_{\leq t} &\leq \sum_{i=1}^l \norm{\Pi_{(x_i, y_i)\in\mathsf{db}, \reg I} - \pc_{x_i, \reg I}\Pi_{(x_i, y_i)\in\mathsf{db}, \reg I}\pc_{x_i, \reg I}^\dagger}_{\leq t+l}
        \label{eqn:triangle-fund-perm}
    \end{align}
    by repeated applications of triangle inequality. We can analyze each term individually. Note that the operator $\Pi_{(x, y)\in\mathsf{db}, \reg I} - \pc_{x, \reg I}\Pi_{(x, y)\in\mathsf{db}, \reg I}\pc_{x, \reg I}^\dagger$ preserves the subspace $\algo H(I|^x)$ for any $I \in \mathbf I|^x$. Therefore, it suffices to consider the action on a state \begin{align}
        \ket{\phi}_{\reg I} &= \sum_{y' \in [N] \setminus \im(I) \cup \{\bot\}} \alpha_{y'} \ket{I[x\ra y']}_{\reg I},
    \end{align}
    for some $I \in \mathbf I|^x$ of size at most $t+l$. We write $s = \frac{1}{\sqrt{N-|\im(I)|}}\sum_{y' \in [N] \setminus \im(I)} \alpha_{y'}$. We can then compute \begin{align}
        \Pi_{(x, y)\in\mathsf{db}, \reg I} \ket{\phi}_{\reg I} &= \alpha_y \ket{I[x \ra y]} \\
        \pc_{x, \reg I}\Pi_{(x, y)\in\mathsf{db}, \reg I}\pc_{x, \reg I}^\dagger\ket{\phi}_{\reg I} &= \pc_{x, \reg I} \left(\alpha_y + \frac{\alpha_\bot - s}{\sqrt{N-|\im(I)|}}\right) \ket{I[x \ra y]} \\
        &= \left(\alpha_y + \frac{\alpha_\bot -s}{\sqrt{N-|\im(I)|}}\right) \cdot \nonumber\\&\quad\left(\ket{I[x \ra y]} - \frac{1}{\sqrt{N-|\im(I)|}}\ket{+_{x, I}} + \frac{1}{\sqrt{N-|\im(I)|}}\ket{I}\right) \\
        &= \alpha_y \ket{I[x \ra y]}_{\reg I} + O\left(\frac{1}{\sqrt{N-|\im(I)|}}\right).
    \end{align}
    where the last line follows from $\abs{s}, \abs{\alpha_\bot} \leq 1$. Recalling that $N-|\im(I)| \geq N-l-t$, \begin{align}
        \norm{\Pi_{(x, y)\in\mathsf{db}, \reg I} \ket{\phi}_{\reg I} - \pc_{x, \reg I}\Pi_{(x, y)\in\mathsf{db}, \reg I}\pc_{x, \reg I}^\dagger\ket{\phi}_{\reg I}} &= O\left(\frac{1}{\sqrt{N-l-t}}\right),
    \end{align}
    which combined with \Cref{eqn:triangle-fund-perm} proves the claim.
\end{proof}

\section{Purified permutations}

\label{sec:Feistel-twirl}

In this section, we introduce the ensemble of permutations which we will analyze to prove both soundness of our construction and of Feistel, and prove properties about the purified ensemble. We use $\varphi$ to denote the permutation which is being queried. We will write $\varphi$ as a product of constituent permutations. Some of the constituent permutations will be highly structured, which will enable us to apply the standard compressed oracle theory.

In particular, we consider an ensemble defined by two random permutations $\pi, \omega \in \mathbf S_{N}$, which we refer to as the outer or twirling permutations, and three random functions $h, k, f : \bit^{n}\rightarrow \bit^n$, which we refer to as the internal or inner Feistel. We will always take $h, k, f$ to be uniform random functions, but we do not yet commit to a particular distribution for $\pi$ and $\omega$. This tuple defines the permutation $\varphi_{(\pi, \omega, h, k, f)}$. The formula for this permutation is $\varphi_{(\pi, \omega, h, k, f)}(x) = y$, where \begin{align}
    u &= \pi(x),  &v &= \omega^{-1}(y), \label{eqn:uv} \\
    m \oplus u_R &= h(u_L), &v_L \oplus u_L &= k(m), & m \oplus v_R &= f(v_L). \label{eqn:muv}  
\end{align}
This is depicted in \Cref{fig:masked-feist-vars}, along with the variables above. For a generic value on the left wire, we will usually use $w$ or $l$; on the right wire, usually $z$ or $r$.

\begin{figure}
    \centering
    \includegraphics[width=0.72\linewidth]{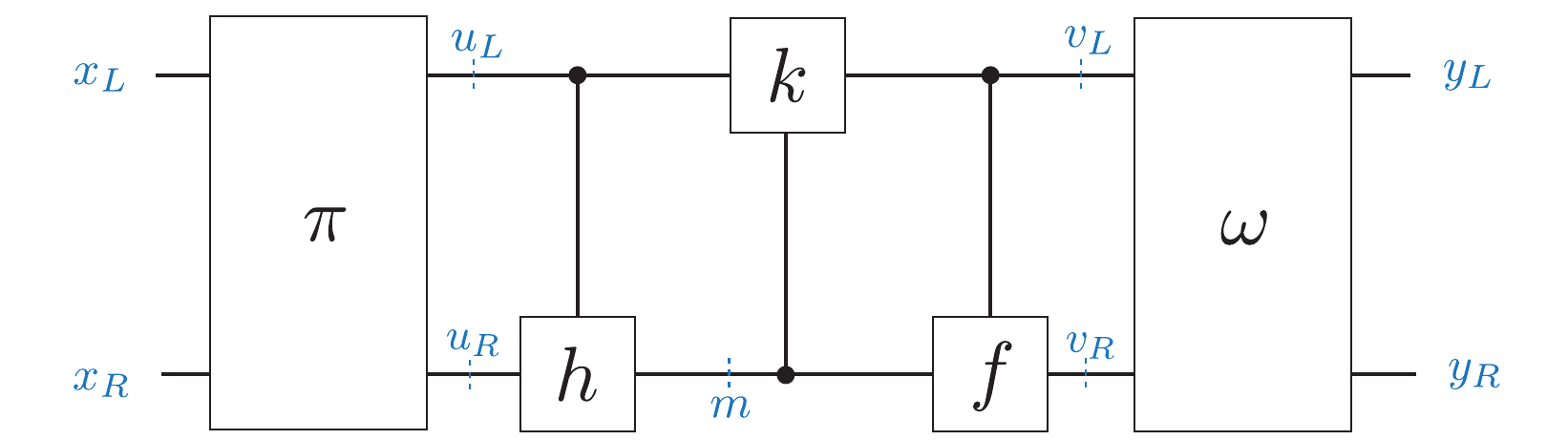}
    \caption{The masked feistel construction, with commonly chosen variable names in blue.}
    \label{fig:masked-feist-vars}
\end{figure}

We call this construction \emph{masked Feistel}, as the outer permutations will hide the structure of the internal Feistel construction. We will analyze this ensemble by purifying.
The purification register lives in the space $\algo H_{\reg P} = \algo H(\mathbf S_N)^{\otimes 2} \otimes \algo H(\mathbf F_{\sqrt{N}, \sqrt{N}})^{\otimes 3}$ for the permutations $\pi, \omega$ and functions $h, k, f$ respectively. We label registers holding the parts of a purification by $\reg{\Pi,\Omega,H,K,F}$ respectively. For concreteness, it may help to imagine $\pi, \omega$ being chosen at uniform random; this is the case we will analyze in our soundness proof. In this case, the permutation $\varphi$ is uniform random, as it is the product of a uniform random permutation (say $\pi$) and an independently chosen permutation.

\subsection{Query analysis}
\label{sec:query-pperms}

Let us define the purified query operator $\mfpu$ on this space to act as \begin{align}
    \mfpu \ket{x, y}_{\reg{XY}} \ket{\pi, \omega, h, k, f}_{\reg P} &\coloneqq \ket{x, y \oplus \varphi_{(\pi, \omega, h, k, f)}}_{\reg{XY}} \ket{\pi, \omega, h, k, f}_{\reg P}.
\end{align}
Looking ahead, we will answer inverse queries by conjugating this operator by a flip of the truth table; for now we focus on forwards queries for simplicity. Let $\algo D$ be the distribution on $\mathbf S_N^{\times 2}$ from which the twirling permutations are drawn. We can write the initial purified state \begin{align}
    \ket{\psi_0}_{\reg P} &= \sum_{\pi, \omega \in \mathbf S_N} \sqrt{p_{\algo D}(\pi, \omega)} \ket{\pi, \omega}_{\reg{\Pi\Omega}} \otimes \frac{1}{\sqrt{(2^n)^{3 \cdot 2^n}}} \sum_{h, k, f \in \mathbf F_{2^n, 2^n}} \ket{h, k, f}_{\reg{HKF}}.
\end{align}
Then, answering queries using the $\mfpu$ operator and tracing out the purifying register at the end of the experiment produces the same view as randomly selecting $\varphi$ from the appropriate distribution.

Now, we may use compressed oracle representations for the internal functions $\reg{HKF}$. Answering a query on input $x$ in this compressed picture can be thought of as the following out-of-place circuit: \begin{enumerate}
    \item Compute $u=\pi(x)$ into a fresh register
    \item Call the compressed oracle on $\reg H$ with input $u_L$, and write the output $z$ to a fresh register.
    \item Call the compressed oracle on $\reg K$ with input $u_R \oplus z$, and write the output $w$ to a fresh register.
    \item Call the compressed oracle on $\reg F$ with input $u_L \oplus w$, write the output $z'$ to a fresh register.
    \item Compute $v=u \oplus (w \Vert (z \oplus z'))$, and output $\omega(v)$.
    \item Uncompute all intermediate values.
\end{enumerate} 
However, this circuit has many redundant steps. For instance, at the end of step (2) the database in register $\reg H$ is compressed on input $u_L$. The next operation which acts on this register is a decompression on $u_L$, in step (6). These operations are inverses of one another, and so they cancel out. Similarly, step (2) writes $z$ to a fresh wire, only for the next step to be a decompression conditioned on $z$. It will instead be easier to imagine performing the decompression in steps (2-4) without writing the result to a new register, nor performing the final compression in these steps. Then, the query can be answered in step (5) using information in the databases contain in registers $\reg{HKF}$. Finally, we undo by compressing steps (2-4), in the reverse order.

We will call this sequence of decompressions the \emph{masked Feistel decompression}, writing $\mfc^\dagger$, and $\mfc$ for the corresponding masked Feistel compression. At an intuitive level, this decompression operator performs the following, where we use $D_h, D_k, D_f$ to denote the databases stored in $\reg{H,K, F}$ respectively. We introduce some auxiliary variables only for notational convenience; one should not think of these as being stored in any register.

\begin{enumerate}
    \item Decompress $\reg H$ on input $u_L$, where $u=\pi(x)$.
    \item Decompress $\reg K$ on input $u_R \oplus D_h(u_L)$.
    \item Decompress $\reg F$ on input $u_L \oplus D_k(u_R \oplus D_h(u_L))$.
\end{enumerate}

Note that the database entries referenced are guaranteed to exist here, as they were just decompressed, and the compressed oracle maintains validity. After this operator is performed, the value $\varphi(x)$ is guaranteed to be determined, and can be XORed into the target query register $\reg{Y}$. Following this, undoing the above steps in reverse order will re-compress the database. This method of answering queries is equivalent to the out-of-place circuit described above.

We will now define the compression operator $\mfc$ more formally. Let \begin{align}
    U^{(H)}_{u, \reg H} &\coloneqq \fc_{u_L, \reg H} \\
    U^{(K)}_{u, \reg{HK}} \ket{D_h}_{\reg H} &\coloneqq \ket{D_h}_{\reg H} \otimes \begin{cases}
        \Id_{\reg K} & \text{(If $u_L \not\in \dom(D_h)$)} \\
        \fc_{u_R \oplus D_h(u_L), \reg K} & \text{(Otherwise)}
    \end{cases} \\
    U^{(F)}_{u, \reg{HKF}} \ket{D_h, D_k}_{\reg{HK}} &\coloneqq \ket{D_h, D_k}_{\reg{HK}} \otimes \begin{cases}
        \Id_{\reg F} & \text{(If $u_L \not\in \dom(D_h)$)} \\
        \Id_{\reg F} & \text{(If $u_R \oplus D_h(u_L) \not\in \dom(D_k)$)} \\
        \fc_{u_L \oplus D_k(u_R \oplus D_h(u_L)), \reg F} & \text{(Otherwise)}
    \end{cases} \\
    U_{u, \reg{HKF}} &\coloneqq U^{(H)}_{u, \reg H} U^{(K)}_{u, \reg{HK}} U^{(F)}_{u, \reg{HKF}},
\end{align}
where $U_u$ compresses the internal databases on input $u$. Note that the choice of $\Id$ in the above cases is arbitrary; as described above, when these operators are called in the context of the masked Feistel construction the values are guaranteed to exist, and so these cases do not occur.

We can now define $\mfc$ by applying the twirl, to obtain \begin{align}
    \mfc_{x, \reg{\Pi HKF}} \ket{\pi}_{\reg \Pi} &\coloneqq \ket{\pi}_{\reg \Pi} \otimes  U_{\pi(x), \reg{HKF}}, \\
    \mfc_{\reg{X\Pi HKF}} \ket{x}_{\reg X}  &\coloneqq \ket{x}_{\reg X}\otimes \mfc_{x, \reg{\Pi HKF}}.
\end{align}

We also define an operator analogous to the flip operator, which we denote by $\mflip$, which acts as \begin{align}
    \mflip_{\reg{P}} \ket{\pi}_{\reg \Pi}\ket{\omega}_{\reg \Omega} \ket{D_h}_{\reg H} \ket{D_k}_{\reg K} \ket{D_f}_{\reg F} &\coloneqq \ket{\omega^{-1}}_{\reg \Pi}\ket{\pi^{-1}}_{\reg \Omega} \ket{D_f}_{\reg H} \ket{D_k}_{\reg K} \ket{D_k}_{\reg F},
\end{align}
in other words this operator inverts and exchanges both twirl permutations, and interchanges the databases in $\reg F$ and $\reg H$. This operator is an involution, and if a purification $P$ represents permutation $\varphi_P$, then the flipped purification $P'$ where $\mflip \ket{P} = \ket{P'}$ represents the permutation $\varphi^{-1}_P=\varphi_{P'}$. Now, we can define the \emph{compressed masked Feistel oracle} $\cmf$ as \begin{align}
    \cmf \ket{b} &\coloneqq \ket{b} \otimes \begin{cases}
        \mfc \cdot \mfpu \cdot \mfc^\dagger & \text{(If $b=0$)} \\
        \mflip \cdot \mfc \cdot \mfpu \cdot \mfc^\dagger \cdot \mflip^\dagger & \text{(If $b=1$)} 
    \end{cases}
\end{align}
Observe that in this case we have $\mfc \neq \mfc^\dagger$, as the order of compression/decompressions on the $\reg{HKF}$ databases are flipped. This is in contrast with the compressed function and compressed permutation oracles, where compression is an involution.

We use notation $D=(D_h, D_k, D_f) \in \mathbf D_{2^n, 2^n}^{\times 3}$ for the compressed databases of functions $h, k, f$, and think of $P = (\pi, \omega, D)$ as a possible state of the purification registers. We use $\mathbf P$ to denote set of possible purifications (i.e. $P \in \mathbf P$). The size of a purification is the size of its largest database, $|P|=\max \{|D_h|,|D_k|,|D_f|\}$. We will define a number of desirable properties of both the twirling permutations $(\pi, \omega)$ and the internal Feistel databases $(D_h, D_k, D_f)$, beginning with the latter.





\subsection{Internal Feistel}
\label{sec:combin-Feistel}

In this section, we will analyze the internal three rounds of Feistel. In particular, our main object of study will be triples $D=(D_h, D_k, D_f)$ of databases. A core concept is a chain, which is a set of database points which suffice to define an input-output pair of the Feistel permutation.

\begin{definition}
    Fix a $D=(D_h, D_k, D_f)$. We say there is a \emph{chain} from $u$ to $v$ in $D$ if $D_h(u_L)\oplus u_R = D_f(v_L) \oplus v_R = m$, and $u_L \oplus v_L = D_k(m)$. The values $(u, m, v)$ define the chain, and $(u_L, m, v_L)$ are the input points along the chain.
    \label{def:chain}
\end{definition}
We also say that the \emph{rightward semi-chain} from $u$ is length $0$ if $D_h(u_L)=\bot$, is length $1$ if $D_k(D_h(u_L) \oplus u_R) = \bot$, is of length $2$ if $D_f(u_L\oplus D_k(D_h(u_L) \oplus u_R))=\bot$, and is of length $3$ (and therefore a full chain to some $v$) otherwise.

We say that $D$ \emph{supports} the pair $(u, v)$ if there is a chain from $u$ to $v$.
We use $\supp(D)$ to denote the injective database supported by the $D$ tuple, and $\dom(D)$ and $\im(D)$ to denote the domain and range of this database (i.e. we write $\dom(D)$ in place of $\dom(\supp(D))$). We write $D(u)$ in place of $\supp(D)(u)$, and similar for $D^{-1}(v)$. Finally, for $u \in \dom(D)$ we write $D[u\ra \bot]$ to denote $D$ with the chain beginning at $u$ removed.

Two chains \emph{collide} if any of the function inputs along the chain collide.

\begin{definition}
    If $(u, m, v) \neq (u', m', v')$ are distinct chains in $D$, then they \emph{collide} if $m=m'$, $u_L = u_L'$, or $v_L = v_L'$.
    \label{def:chain-col}
\end{definition}

A database triple $D=(D_h, D_k, D_f)$ is \emph{canonical} if no two chains in $D$ collide, and every input output pair in each database participates in exactly one chain. The set of all canonical database triples is denoted $\mathbf C \subset \mathbf D_{2^n, 2^n}^{\times 3}$, and $\mathbf C_t$ denotes the set of canonical database where each $D_h$, $D_k$, and $D_f$ are of size $t$.

\begin{lemma}
    Let $D=(D_h, D_k, D_f) \in \mathbf C_t$ be a canonical triple of databases of size $t$ each. Then there are $t$ full chains $u_1 \ra v_1, \dots, u_t \ra v_t$ with distinct $u_i$ and distinct $v_i$. There are $t(t-1)$ length $2$ rightward semi-chains, $t\cdot (2^n-t)$ of length $1$, and $(2^n-t) \cdot 2^n$ of length $0$. There are the same number of leftward semi-chains of the appropriate length.
\end{lemma}

\begin{proof}
    The set of chains being size $t$ is the definition of a canonical database. To compute the number of length $2$ semi-chains, we choose any $w \in D_h$, and $z$ such that $z \oplus D_h(w) \in D_k$, but which does not lead to an already queried value in $D_f$. We know that $w = (u_i)_L, z=(u_i)_R$ satisfy the first two conditions but not the third, but if a different $z' \neq z$ lead to a queried value in $D_f$ then we would have a colliding chain. This would violate canonicity, so we may choose any $z' \neq z$, of which there are $t-1$ satisfying the remaining constraints.

    For the number of semi-chains of length $1$, we choose $w \in D_h$ and any $z$ s.t. $D_h(w) \oplus z \not\in D_k$. There are $t$ choices for $w$, and $2^n-t$ choices for $z$ given $w$. Similarly, for semi-chains of length $0$ we may pick any $w \not\in D_h$ of which there are $2^n-t$, and any $z$ of which there are $2^n$. We can verify that
    \begin{align}
        t + t(t-1) + t \cdot (2^n-t) + (2^n-t)\cdot 2^n &= t^2 + (t + 2^n)\cdot (2^n-t) \\
        &= 2^{2n}.
    \end{align}
\end{proof}

Given some injective database $I \in \mathbf I_{t}$, let $\mathbf D(I)$ denote the set of canonical databases which support $I$. Note that this set is empty if any one of the following are true.

\begin{definition} An $I \in \mathbf I_t$ is \emph{non-allowable} if any of the following predicates are true. 
    \begin{enumerate}
        \item There is an \emph{input left collision} in $I$ if there are distinct $x, x' \in \dom(I)$ such that $x_L =x_L'$.
        \item There is an \emph{output left collision}  in $I$ if there are distinct $y, y' \in \im(I)$ such that $y_L = y_L'$.
        \item There is an \emph{internal left collision} in $I$ if there is an $(x, y)\in I, x' \in \dom(I), y' \in \im(I)$ with $(x, y) \neq (x', y')$ and $x_L \oplus y_L = x'_L \oplus y'_L$.
    \end{enumerate}
    \label{def:non-allowable-dbs}
\end{definition}

We say $I$ is allowable if none of these events occur, i.e. $\mathbf D(I) \neq \emptyset$, and denote by $\mathbf A_t \subset \mathbf I_t$ the set of allowable databases.\footnote{In fact, condition (3) entails conditions (1) and (2), though we explicitly state all three for clarity.}

\begin{lemma}
    \Cref{def:non-allowable-dbs} characterizes all non-allowable databases.
    \label{lem:allowable-dbs-characterization}
\end{lemma}

\begin{proof}
    Suppose that $I$ does not satisfy \Cref{def:non-allowable-dbs}. If there is an input or output left collision, then there cannot be disjoint chains for the left colliding inputs or outputs. Suppose then that there are not input or output left collisions, yet there is an internal collision. Let $(u, v) \in I, u' \in \dom(I), v' \in \im(I)$ be such a collision. Suppose that $D$ were a canonical database supporting $I$, such that there is a chain from $u$ to $v$. We write this chain as \begin{align}
        \underbrace{(u_L, z_h)}_{\in D_h} \mapsto \underbrace{(u_R \oplus z_h, w_k)}_{\in D_k} \mapsto \underbrace{(u_L \oplus w_k, z_f)}_{\in D_f},
    \end{align}
    where we observe that $w_k = u_L \oplus v_L$. Now observe that $u_L' \in \dom(D_h)$ and $v_L' \in \dom(D_f)$, so let us define \begin{align}
        z_h' &= D_h(u_L'), & z_f' &= D_f(v_L').
    \end{align}
    Consider the inputs $u'' = u_L' \Vert (z_h' \oplus u_R \oplus z_h)$ and $v'' = v_L' \Vert (z_f' \oplus u_R \oplus z_h)$. There is then a chain from $u''$ to $v''$ of the form \begin{align}
        \underbrace{(u_L'', z_h')}_{\in D_h} \mapsto \underbrace{(u_R \oplus z_h, w_k)}_{\in D_k} \mapsto \underbrace{(v_L'', z_f')}_{\in D_f},
    \end{align}
    where we use the fact that $u_L' \oplus w_k = v_L'=v_L''$. This chain collides with the chain found from $u$ to $v$. We know either $u'' \neq u$ (if $u' \neq u$, and assumption of no input left collisions) or $v'' \neq v$ (if $v' \neq v$, and assumption of no output left collisions), so there is a colliding chain in $D$, a contradiction.

    Now let us consider the other direction. Suppose that $I$ does satisfy \Cref{def:non-allowable-dbs}, we will construct a database $D$ which supports exactly $I$ and is canonical. Fix some ordering on the input output pairs as $I=(u_1, v_1), (u_2, v_2), \dots, (u_t, v_t)$. Then, beginning with an empty $D = (D_h, D_k, D_f)$, iteratively construct the database as in \Cref{alg:db-extension}.
    \begin{algorithm}[H]
        \caption{Database extension}
        \label{alg:db-extension}
        For every $j \in [t]$:
        
        \qquad Choose some $z_j$ such that $(u_j)_R \oplus z_j \not\in \dom(D_k)$, and add $((u_j)_L, z_j)$ to $D_h$.
        
        \qquad Add $((u_j)_R \oplus z_j, (u_j)_L \oplus (v_j)_L)$ to $D_k$.
        
        \qquad Add $((v_j)_L, (v_j)_R \oplus z_j \oplus (u_j)_R)$ to $D_f$.
    \end{algorithm}
    After $j-1$ iterations, step (2) has $2^n-j+1$ choices, and hence will succeed (unless $t \geq 2^n$, in which case database states are anyways not defined). If step (2) succeeds, then it will create new length $1$ right facing semi-chains beginning at $(u_j)_L \Vert z_R$ for every $z_R \in \bit^{n}$. Of these, $j-1$ will be extended immediately to length $2$, corresponding to the values already in $D_k$. On the left wire, these length $2$ right facing semi-chains lead to the values $(u_j)_L \oplus (u_i)_L \oplus (v_i)_L$ for $i \in [j-1]$. From the allowability of $I$, we know that these values will never be extended to full length $3$ chains, as $(u_j)_L \oplus (u_i)_L \oplus (v_i)_L \neq (v_{p})_L$ for any $p \in [t]$.
    
    Step (3) will add a new pair to $D_k$ if step (2) succeeds, and of the previously described length $1$ right facing semi-chains, will extend only the one beginning at $u_j = (u_j)_L \Vert (u_j)_R$, as desired.
    However, it will also extend length $1$ right facing semi-chains of the form $(u_i)_L \Vert (z_i \oplus (u_j)_R)$, for every $i \in [j-1]$. These now-length-$2$ right facing semi-chains will lead to a value $(u_i)_L \oplus (u_j)_L \oplus (v_j)_L$ on the left wire after the application of $h$ and $k$. However, we know that this value is distinct from all $(v_p)_L$ for any $p \in [t]$, by the fact that $I$ has no internal left collisions. Hence, these chains are never extended to length $3$ chains. Therefore, each iterations adds a chain from $u_j$ to $v_j$, disjoint from all prior chains, and adds no other chains. The databases are clearly minimal at the end of this procedure, so $D$ is a canonical database tuple supporting $I$.
\end{proof}

\begin{lemma}
     For any $I \in \mathbf A_t$, we have the equation $|\mathbf D(I)| = \binom{2^n}{t}$.
    \label{lem:counting-allowables}
\end{lemma}

\begin{proof}
    Consider the algorithm described in the proof of \Cref{lem:allowable-dbs-characterization}. In the first iteration of step (1), there are $2^n$ choices, in the second, $2^n-1$, and etc. up to $2^n-t+1$ in the final. Any choice which is not of the form described leads immediately to a chain collision, and hence is not allowable. Steps (2) and (3) do not allow any choices, as the values in $D_k$ and $D_f$ are constrained by the value of $v$. However, the choice of ordering of input-output pairs is arbitrary, leading to a factor $t!$ redundancy, and therefore a total of $\binom{2^n}{t}$ databases consistent with $I$.
\end{proof}

\begin{definition}
    Let $A \in \mathbf A_t$, and $u \not\in \dom(I)$. We say that $A$ \emph{allows} $u$ if there exists a $v \in \bit^{2n}$ such that $A[u\ra v] \in \mathbf A_{t+1}$. We denote this as $A \heart u$, and if $A[u \ra v]$ remains allowable then we write $A \heart (u, v)$ and say $(u, v)$ is allowed.
    \label{def:allows}
\end{definition}

\begin{remark}
    If $A$ allows $u$, then $u'$ with $u'_L=u_L$ is allowed as well. There are at least $2^n-t^2-t$ allowable left substrings, by a direct counting argument.
    \label{rem:allowable-xs}
\end{remark}

\begin{definition}
    Fix some $A \in \mathbf A_t$ and $u \in \bit^n$ which $A$ allows. Let $\mathbf V_{u, A}$ denote the set of values $v \in \bit^{2n}$ such that $A[u\ra v]$ is allowable, and let $\mathbf L_{u, A}$ denote the set of left substrings in $\mathbf V_{u, A}$. Formally: \begin{align}
        \mathbf V_{u, A} &\coloneqq \{v \, : \, A[u\ra v] \in \mathbf A_{t+1}\}, \\
        \mathbf L_{u, A} &\coloneqq \{v_L \, : \, v \in V_{u, A}\}.
    \end{align}
    \label{def:allowable-lefts}
\end{definition}

We will sometimes use a length $n$ bitstring $l \in \bit^n$ as the index $\mathbf L_{l, A}$ or $\mathbf V_{l, A}$. This is well defined, as these sets only depend on the left $n$ bits of $u$.

\begin{lemma}
    For any $A \in \mathbf A_t$ and $u \in [N] \setminus \dom(A)$ which $A$ allows, we have \begin{align}
        |\mathbf L_{u, A}| \geq 2^n - 2t^2-2t,
    \end{align}
    and further for any $l \in \mathbf L_{u, A}$ and $r \in \bit^n$, we have $l \Vert r \in \mathbf V_{x, A}$.
    \label{lem:allow-prefixes}
\end{lemma}

\begin{proof}
    The second statement is straightforward, as allowability does not constrain the right substring in any way. For the first, observe that an assignment $A[u\ra v]$ can cause an internal left collision in one of two ways: \begin{enumerate}
        \item There exists $u' \in \dom(A) \cup \{u\}$, $v' \in \im(A) \cup \{v\}$ with $(u', v') \neq (u, v)$ such that \begin{align}
            u_L' \oplus v_L' &= u_L \oplus v_L
        \end{align}
        \item There exists $u' \in \dom(A) \cup \{u\}$, $(u'', v'') \in A$ such that \begin{align}
            u_L' \oplus v_L &= u_L'' \oplus v_L''
        \end{align}
    \end{enumerate}
    Note that in case (1), when $v=v'$ the equation reduces to the equality $u_L' = u_L$, which is not true by the fact that $A$ allows $u$. Hence, there are $t^2+t$ equations for the first ($t+1$ for $u'$, $t$ for $v'$), which therefore eliminate at most $t^2+t$ possible values for $v_L$. Similarly, there are $t+1$ choices for $u'$ and $t$ choices for $(u'', v'')$ in case (2), giving at most $2t^2+2t$ eliminated choices of $v_L$. Recall that the lack of internal left collisions implies no output left collisions (in particular, case (1) with $u=u'$ gives $v_L \neq v_L'$ for all $v \in \im(A)$), so any choice not excluded leads to an allowable database.
\end{proof}

\paragraph{Database extensions.}

Now we turn to studying subsets of three round databases which can occur as intermediate states between canonical databases. We will refer to these as extensions. At a high level, these sets represent databases which would be canonical, except for one chain which is only partially defined.

\begin{definition}
    We say that a tuple $D$ is a \emph{one-extension} if there exists: \begin{enumerate}
        \item some canonical $D' = (D_h', D_k', D_f')$,
        \item some $l \in \bit^n$ such that $D'$ allows $l \Vert r$ for some $r$, yet $l \Vert r \not\in \dom(D')$ for any $r$,
        \item some $z \in \bit^n$,
    \end{enumerate} such that $D=(D_h'[l \ra z], D_k', D_f')$.
    \label{def:one-extension}
\end{definition}
We would say that $D$ above is a one-extension of $D'$ with $(l, z)$. We use $\mathbf E (\mathbf C_t)$ to denote the set of one-extensions of size $t$ canonical databases, and $\mathbf E(D)$ to denote the set of one-extensions of $D$. We define two-extensions similarly.
\begin{definition}
    We say that a tuple $D$ is a \emph{two-extension} if there exists: \begin{enumerate}
        \item some $D'$ which is a one-extension of $D''$ with $(l, z)$,
        \item some $m \in \bit^n$ such that $m \not\in \dom(D_k'$),
        \item some $w \in \bit^n$ such that $l \oplus w \in \mathbf L_{l, A(D'')}$,
    \end{enumerate} such that $D=(D_h', D_k'[m \ra w], D_f')$.
    \label{def:two-extension}
\end{definition}
We would say that $D$ above is a two-extension of $D''$ with $(l, z), (m, w)$, and a one-extension of $D'$ with $(m, w)$. We use $\mathbf E^2 (\mathbf C_t)$ to denote the set of two-extensions of size $t$ canonical databases, $\mathbf E^2(D)$ to denote the set of two-extensions of a $D \in \mathbf C$, and $\mathbf E(D')$ to denote the set of two-extensions obtained as one-extensions of $D' \in \mathbf E(\mathbf C)$. Finally, we define a three extension,
\begin{definition}
    We say that a tuple $D$ is a \emph{three-extension} if there exists: \begin{enumerate}
        \item some $D'$ which is a two-extension of $D''$ with $(l, z), (m, w)$,
        \item some $z' \in \bit^n$,
    \end{enumerate} such that $D=(D_h', D_k', D_f'[l \oplus w \ra z'])$.
    \label{def:three-extension}
\end{definition}
Similarly, we use $\mathbf E^3 (\mathbf C_t)$ to denote three-extensions of canonical databases, $\mathbf E^3 (D)$ for $D \in \mathbf \mathbf C_t$ to denote three-extensions of $D$, $\mathbf E^2 (D)$ for $D \in \mathbf E(\mathbf C_t)$ to denote three-extensions obtained as two-extensions of a fixed one-extension, and $\mathbf E(D)$ for $D \in \mathbf E^2(\mathbf C_t)$ to denote three-extensions obtained as one-extensions of a fixed two-extension.

\begin{lemma}
    Fix some $u \in [N]$. For any $t \geq 0$, canonical $D \in \mathbf C_{t+1}$ with $u \in \dom(D)$, there exists a unique $D' \in \mathbf C_t$ with $u \not\in \dom(D)$ such that $D \in \mathbf E^3(D')$. Additionally, all three-extensions of canonical databases are themselves canonical.
    \label{lem:ext-unique}
\end{lemma}

\begin{proof}
    For the first statement, consider the database $D'$ obtained by removing from $D$ the chain beginning at $u$. Canonicity is preserved under removing a chain (this cannot add a chain collision, or make the database non-minimal), so $D \in \mathbf E^3(D')$. Further, every $D'' \neq D'$ will have either a chain beginning at $u' \in \dom(D)$ and ending at a $v' \neq D(u')$, or a chain beginning at some $u'' \neq u$ and $u'' \not\in \dom(D)$. In either case, no extension can remove this chain, and such a chain does not appear in $D$, so $D'' \not\in \mathbf C^3(D')$.

    For the second statement, recall the correctness of \Cref{alg:db-extension}, proved in \Cref{lem:allowable-dbs-characterization}. In particular, an extension can be seen as one iteration of said algorithm, over all possible allowable $(u, v)$ pairs.
\end{proof}

Now we prove some disjointness properties for extension sets. These properties will be useful later on to argue that the span of various extension sets will be orthogonal, and hence certain operators can be written in a direct sum manner. This in turn will be key to analyzing the complicated effect of compression operators on the databases when they are entangled with the other registers. Towards this, we introduce the pipe notation for database triples. We use an upwards arrow $\ua$ for subsets that will occur while evaluating forwards and/or adding output points, and a downwards arrow $\da$ for subsets that occur while evaluating backwards and/or removing output points.

\begin{definition}
    Fix some $D \in \mathbf C_t$ and $u \in [N] \setminus \dom(D)$ which $D$ allows. Then we define \begin{align}
        D\ua^u_h &\coloneqq \{D\} \cup \{(D_h[u_L \ra z], D_k, D_f) \, : \, z \in \bit^n \setminus \{u_R \oplus \dom(D_k)\}\}.
    \end{align}
    For some $D' \in D\ua^u_h \subset \mathbf E(D)$ with extension $(u_L, z)$ and $z \in \bit^n \setminus \{u_R \oplus \dom(D_k)\}$, define \begin{align}
        D'\ua^u_k &\coloneqq \{D'\} \cup \{(D_h', D_k'[u_R\oplus z \ra w], D_f') \, : \, w \in \{u_L \oplus \mathbf L_{u, \supp(D)}\}\}.
    \end{align}
    For some $D'' \in D'\ua^u_k \subset \mathbf E^2(D)$ with extension $(u_L, z), (u_R \oplus z, w)$ as above, we define
    \begin{align}
        D''\ua^u_f &\coloneqq \{D''\} \cup \{(D_h'', D_k'', D_f''[u_L \oplus w \ra z']) \, : \, z \in \bit^n\}.
    \end{align}
    \label{def:up-pipe}
\end{definition}
When this notation is applied to sets of databases, it is the union of the application to each individual one. We will collapse the notation $(D\ua^u_h)\ua^u_k = D\ua^u_{h,k}$, and similarly $((D\ua^u_h)\ua^u_k)\ua^u_f = D\ua^u_{h,k,f}$.

\begin{lemma}
    For any $D \neq D' \in \mathbf C$ and $u \in [N]$ which is allowed by $D$ and $D'$ but supported by neither, meaning $u \not\in \dom(D) \cup \dom(D')$, we have \begin{align}
        D\ua^u_h \cap D'\ua^u_h &= \emptyset, & D\ua^u_{h, k} \cap D'\ua^u_{h, k} &= \emptyset, & D\ua^u_{h, k, f} \cap D'\ua^u_{h, k, f} &= \emptyset.
    \end{align}
    \label{lem:disjoint-uparrows}
\end{lemma}

\begin{proof}
    There is WLOG a chain in $D$ but not $D'$ that does not involve $u$, and this chain cannot be removed by any extension, nor can it be added by any extension that does not originate from $u$.
\end{proof}

Now we turn to defining and analyzing the downwards arrow notation. In contrast to the upwards arrow, the downwards arrow notation should be thought of as changing or removing an existing input-output point.

\begin{definition}
    Let $D \in \mathbf C$ and $u \in \dom(D)$. Then we define \begin{align}
        D\da^u_h &\coloneqq \{(D_h[u_L \ra z], D_k, D_f) \, : \, z \in \bit^n \cup \{\bot\} \setminus \{u_R \oplus \dom(D_k)\}\} \\
        D\da^u_k &\coloneqq \{(D_h, D_k[u_R \oplus D_h(u_L) \ra w], D_f) \, : \, w \in \{u_L \oplus \mathbf L_{u, D[u\ra\bot]}\} \cup \{\bot\} \} \\
        D\da^u_f &\coloneqq \{(D_h, D_k, D_f[D(u)_L \ra z]) \, : \, z \in \bit^n \cup \{\bot\}\}
    \end{align}
    \label{def:down-pipe}
\end{definition}

\begin{lemma}
    For any $D \neq D' \in \mathbf C$ and $u \in \dom(D) \cap \dom(D')$, we have the disjointness relations \begin{align}
        D\da^u_h \cap D'\da^u_h &= \emptyset, & D\da^u_k \cap D'\da^u_k &= \emptyset.
    \end{align}
    The identity $D\da^u_f \cap D'\da^u_f = \emptyset$ holds if $D[u \ra \bot] \neq D'[u \ra \bot]$.
    \label{lem:disjoint-downarrows}
\end{lemma}

\begin{proof}
    Let us begin with the first equation. Observe that such a pair $D \neq D'$ has either disjoint chains not involving $u$, or $D(u) \neq D'(u)$. Let us begin with the first case. In this case, changing the image of $u_L$ in $D_h$ as described (or making it undefined) will modify the chain from $u$ to $D(u)$ and/or $u$ to $D'(u)$, changing the last $n$ bits of the preimage. However, it will not create any further chains, as $D_h(u_L)$ is chosen such that it does not connect to any other pair in $D_k$. Hence, the disjoint chain not involving $u$ will remain, and the databases will be distinct.

    In the other case, consider reassignments of the form $D_h[u_L \ra z]$ and $D_h'[u_L \ra z']$. In such a case, there may be no chain from $u$ to either $D(u)$ or $D'(u)$ in the redefined $D$ or $D'$ respectively. If one of $z$ or $z'$ is $\bot$, observe that there are exactly two input-output pairs, $(r, w) \in D_k$ and $(u_L \oplus w, z_f) \in D_f$ which do not participate in a chain. We have $u_L \oplus w \Vert r \oplus z_f = D(u)$. Similarly, under the redefinition $D_h'[u_L \ra z]$ there are exactly two input-output pairs, $(r', w') \in D_k'$ and $(u_L \oplus w', z_f') \in D_f'$ which do not participate in a chain. We have $u_L \oplus w' \Vert r' \oplus z_f' = D'(u) \neq D(u)$. Therefore, these databases must be distinct. Similarly, if neither $z$ nor $z'$ is $\bot$, there will remain some chain to $D(u)$ in the first set and to $D'(u) \neq D(u)$ in the second.

    Now let us turn to the middle statement. Similar to before, the claim is easily handled if there is disjoint chains in $D, D'$ not involving $u$, or if the reassignment is to $\bot$ in either case. In the other case, we perform the non-trivial reassignment $D_k[u_R \oplus D_h(u_L) \ra w'']$ and $D_k'[u_R \oplus D_h'(u_L) \ra w''']$, and have $D(u) \neq D'(u)$. In this case, there are exactly two input-output pairs, $(u_L, z) \in D_h$ and $(u_L \oplus w, z_f) \in D_f$ which do not participate in a chain for the former. Observe that these pairs would uniquely be completed to a chain under undoing the reassignment, i.e. $D_k[u_R\oplus z \ra w] = D_k$. Under this unique completion, we would naturally obtain endpoint $D(u)$ for the chain beginning at $u$. Similarly, there is a unique completion of the modified $D'$, whose endpoint would be $D'(u)$. These unique completions lead to distinct chains, and so the databases cannot be equal.

    The final case follows similarly to the previous two: by the assumption that $D[u \ra \bot] \neq D'[u \ra \bot]$ there are distinct chains not involving $u$ in $D$ and $D'$, which will not be modified by the described reassignment in $D_f$. It is worth observing that if $D(u)_L = D'(u)_L$, on the other hand, then the claim fails to hold.
\end{proof}

With these in place, orthogonality of the spaces spanned by these sets follows. In the following lemma, for some $S \subset \mathbf D$ we take $\algo H(S) \leq \algo H(\mathbf D)$ as a subspace.

\begin{corollary}[of \Cref{lem:disjoint-uparrows,lem:disjoint-downarrows}]
    For any $D \neq D' \in \mathbf C$ and $u \in [N]$ which is allowed by $D$ and $D'$ but supported by neither, we have \begin{align}
        \algo H(D\ua^u_h) &\perp \algo H(D'\ua^u_h), & \algo H(D\ua^u_{h, k}) &\perp \algo H(D'\ua^u_{h, k}), & \algo H(D\ua^u_{h, k, f}) &\perp \algo H(D'\ua^u_{h, k, f}).
    \end{align}
    Further, for any $D \neq D' \in \mathbf C$ and $u \in \dom(D) \cap \dom(D')$, we have \begin{align}
        \algo H(D\da^u_h) &\perp \algo H(D'\da^u_h), & \algo H(D\da^u_k) &\perp \algo H(D'\da^u_k),
    \end{align}
    and $\algo H(D\da^u_f) \perp \algo H(D'\da^u_f)$ when $D[u \ra \bot] \neq D'[u \ra \bot]$.
    \label{cor:ortho-arrows}
\end{corollary}

We can now introduce the notation of assignment for database triples. For $u \not\in \dom(D)$ and $v \not\in \im(D)$ with $(u, v)$ allowable, we define $D[u \ra v]$ as the subset of $\mathbf E^3(D)$ which contains a chain from $u$ to $v$. In other words, \begin{align}
    D[u \ra v] &\coloneqq \{D' \, : \, D'(u)=v, D'[u \ra \bot] = D\}.
\end{align}

\begin{corollary}[of \Cref{lem:ext-unique}]
    If $D \neq D'$ and both allow $(u, v)$, then $D[u\ra v] \cap D'[u \ra v] = \emptyset$.
    \label{cor:specific-ext-unique}
\end{corollary}

These extension sets will allow us to define a modified version of masked Feistel compression, which is close in operator norm to the original but easier to analyze.

\subsection{Canonical compression}
\label{sec:canon-comp}
In this section, we describe an alternative to the $U_u$ compression operator, which we will label by $G_u$ and refer to as \emph{canonical compression}. This operator is constructed such that it maintains canonicity of the database, as long as the input $u$ is allowed. In order to define it, we will require notation for the set of ouput values which we can insert without resulting in an internal collision. Let us fix some databases $D= (D_h,D_k, D_f)$ with $D(u)=\bot$, and write the sets \begin{align}
    H(u, D) &\coloneqq \{z \, : u_R \oplus z \not\in \dom(D_k)\} \\
    K(u, D) &\coloneqq \{w \, : u_L \oplus w \in \mathbf L_{u, \supp(D)}\},
\end{align}
where $\mathbf L_{u, \supp(D)}$ is as defined in \Cref{sec:combin-Feistel}. When we use $\mathbf L_{u, \supp(D')}$ for $D' \in \mathbf E(D)$, this is equivalent to $\mathbf L_{u, \supp(D)}$.
We can define compression operators which swap the uniform superposition over values in this set with the undefined label, rather than the full superposition. To formalize this, fix some $D_h \in \mathbf D|^{u_L}$ and some $D_k \in \mathbf D$, to define \begin{align}
    \ket{+_{u_L, H(u, D_k)}} &\coloneqq \frac{1}{\sqrt{|H(u, D_k)|}} \sum_{z \in H(u, D_k)} \ket{D_h[u_L \ra z]} \\
    G^{(H)}_{u, D_h, D_k, \reg{H}} &\coloneqq \big(\Id - \outerprod{+_{u_L, H(u, D_k)}}{+_{u_L, H(u, D_k)}} - \outerprod{D_h}{D_h} + \\&\qquad  \outerprod{+_{u_L, H(u, D_k)}}{D_h} + \outerprod{D_h}{+_{u_L, H(u, D_k)}}\big)_{\reg H} \nonumber\\
    G^{(H)}_{u, \reg{HK}} \ket{D_k}_{\reg K} &\coloneqq \bigoplus_{D_h \in \mathbf D|^{u_L}, D_k \in \mathbf D} \outerprod{D_k}{D_k}_{\reg K} \otimes G^{(H)}_{u, D_h, D_k, \reg{H}}.
\end{align}
We similarly define a canonical compression operator for the $K$ register, first by fixing some $D=(D_h, D_k, D_f) \in \mathbf E(\mathbf C)$ with $D_k(u_R \oplus D_h(u_L)) = \bot$; we use $m=u_R \oplus D_h(u_L)$ going forward, and defining \begin{align}
    \ket{+_{z, K(u, D)}} &\coloneqq \frac{1}{\sqrt{|K(u, D)|}} \sum_{w \in K(u, D)} \ket{D_h[z \ra w]} \\
    G^{(K)}_{u, D, \reg{K}} &\coloneqq \begin{cases}
        \big(\Id - \proj{+_{m, K(u, D)}} - \outerprod{D_k}{D_k} + \\  \outerprod{+_{m, K(u, D)}}{D_k} + \outerprod{D_k}{+_{m, K(u, D)}}\big)_{\reg K} & \text{(If $D_h(u_L) \neq \bot$)} \\
        \Id_{\reg K} & \text{(Otherwise)}
    \end{cases}\\
    G^{(K)}_{u, \reg{HKF}} \ket{D_k}_{\reg K} &\coloneqq \Id \oplus \bigoplus_{\substack{D \in \mathbf E(\mathbf C) \\ \text{ s.t. } D_k(u_R \oplus D_h(u_L)) = \bot}} \outerprod{D_h, D_f}{D_h, D_f}_{\reg{HF}} \otimes G^{(K)}_{u, D_h, D_f, \reg{K}}.
\end{align}
Note that this is a unitary operator, as $K(u, D)$ does not depend on $D_k(m)$, which cannot participate in a chain as it is $\bot$.
We take $G^{(F)}_{u, \reg{HKF}} = U^{(F)}_{u, \reg{HKF}}$, and define the full canonical compression operator \begin{align}
    G_{u, \reg{HKF}} &= G^{(H)}_{u, \reg{HK}} G^{(K)}_{u, \reg{HKF}} G^{(F)}_{u, \reg{HKF}}.
\end{align}

\begin{remark}
    The canonical compression operators are exponentially close to the standard compression operators after a bounded number of queries. In particular, we have \begin{align}
        \norm{G^{(H)}_{u, \reg{HK}} - U^{(H)}_{u, \reg{H}}}_{\leq t} &= O\left(\sqrt{\frac{t}{2^n}}\right), & \norm{G^{(K)}_{u, \reg{HKF}} - U^{(K)}_{u, \reg{HK}}}_{\leq t} &= O\left(\sqrt{\frac{t^2}{2^n}}\right), \\
        \norm{G^{(F)}_{u, \reg{HKF}} - U^{(F)}_{u, \reg{HKF}}}_{\leq t} &= 0, & \norm{G_{u, \reg{HKF}} - U_{u, \reg{HKF}}}_{\leq t} &= O\left(\sqrt{\frac{t^2}{2^n}}\right).
    \end{align}
    \label{rem:canonical-comp-close}
\end{remark}

\begin{proof}
    The two equations on the top line follow from \Cref{lem:comp-no-bad}, and the bottom left by definition. Then, the bottom right follows from the other three equations by a simple sequence of triangle inequalities.
\end{proof}

\paragraph{Properties.} We can work out the action of $G_u$ and its components in some examples. These facts will be useful in proving soundness of our compressed permutation oracle.

\begin{lemma}
    Let $A \in \mathbf A_t$ be an allowable, size $t$ injective database and let $u \in [N] \setminus \dom(A)$ be an input which $A$ allows. Then we have \begin{align}
        G_{u, \reg{HKF}} \ket{+_{\mathbf D(A)}}_{\reg{HKF}} &= \frac{1}{\sqrt{|\mathbf V_{u, A}|}} \sum_{v \in \mathbf V_{u, A}} \ket{+_{\mathbf D(A[u\ra v])}}_{\reg{HKF}}.
    \end{align}
    \label{lem:canonical-decomp}
\end{lemma}

\begin{proof}
    Let us fix some $D=(D_h, D_k, D_f)$ which is canonical with $\supp(D)=A$, and compute the action on such a state. We have \begin{align}
        G_{u, \reg{HKF}} \ket{D}_{\reg{HKF}} &= \frac{1}{\sqrt{(2^n-t) \cdot |\mathbf L_{u, \supp(D)}| \cdot 2^n}} \sum_{\substack{
        z_1 \in \bit^n \setminus \{\dom(D_k) \oplus u_R\}, \\
        w \in \{\mathbf L_{u, \supp(D)} \oplus u_L\}, \\
        z_2 \in \bit^n
        }} \ket{D_h[u_L \ra z_1]}_{\reg H} \otimes \\&\quad \ket{D_k[u_R\oplus z_1 \ra w]}_{\reg K} \ket{D_f[u_L\oplus w \ra z_2]}_{\reg F}.
    \end{align}
    This state can be seen to be the uniform superposition of canonical databases which support $A[u \ra v]$ for some $v$, and agree with $D$\footnote{Here we say database triples $D, D'$ agree if for every $w \in \dom(D_h) \cap \dom(D_h')$, we have $D_h(w)=D_h'(w)$, and similarly for $D_f$ and $D_k$.}, i.e. the set $D\ua_{h,k,f}^u$. Now observe that every canonical $D'$ supporting $A[u \ra v]$ for some $v$ will agree with one, and only one, canonical $D$ supporting $I$, from \Cref{lem:ext-unique}. Hence, when we extend the above calculation to the uniform superposition of all canonical $D$ supporting $I$, we find \begin{align}
        G_{u, \reg{HKF}} \ket{+_{\mathbf D(A)}}_{\reg{HKF}} &= \frac{1}{\sqrt{|\mathbf V_{u, A}|}} \sum_{v \in \mathbf V_{u, A}} \ket{+_{\mathbf D(A[u\ra v])}}_{\reg{HKF}}.
    \end{align}
\end{proof}

\begin{lemma}
    Fix some $A \in \mathbf A_t$ and $u \in [N] \setminus \dom(A)$ which $A$ allows, to write the state \begin{align}
        \ket{\psi}_{\reg{HKF}} &\coloneqq \sum_{v \in \mathbf V_{u, A}} \alpha_v \ket{+_{\mathbf D(A[u\ra v])}}_{\reg{HKF}}.
    \end{align}
    Then we have \begin{align}
        \norm{\ket{\psi}_{\reg{HKF}} - G_{u, \reg{HK}}^{(H)}\ket{\psi}_{\reg{HKF}}} &= O\left(\frac{1}{\sqrt{2^n-t}}\right).
    \end{align}
    \label{lem:canonical-H-close}
\end{lemma}

\begin{proof}
    Let us consider some $D \in \mathbf D(A)$, and take $|A|=t$. We can write the component of $\ket{\psi}$ which is an extension of $D$ as \begin{align}
            \ket{\psi_D}_{\reg{HKF}} &=\Pi^{\algo H(\mathbf E^3(D))}_{\reg{HKF}} \ket{\psi}_{\reg{HKF}} \\
            &= \frac{1}{\sqrt{|\mathbf D(A)| \cdot (2^n-t)}} \sum_{z \in \bit^n \setminus \{u_R \oplus \dom(D_k)\}} \ket{D_h[u_L \ra z]}_{\reg H} \otimes \\
            &\quad \sum_{v \in \mathbf V_{u, A}} \alpha_v \ket{D_k[u_R \oplus z \ra u_L \oplus v_L], D_f[v_L \ra u_R \oplus z \oplus v_R]}_{\reg{KF}}. \label{eqn:psid-second-line}
        \end{align}
    Further, recall from \Cref{lem:ext-unique} we can write \begin{align}
        \ket{\psi}_{\reg{HKF}} &= \sum_{D \in \mathbf D(A)} \ket{\psi_D}_{\reg{HKF}},
    \end{align}
    so by linearity it suffices to analyze each component separately. Let us define \begin{align}
        \algo H_D = \bigoplus_{v \in \mathbf V_{u, A}} \algo H(D[u \ra v] \da^u_h),
    \end{align}
    where we have $\algo H_D \perp \algo H_{D'}$ for $D \neq D' \in \mathbf D(A)$ from \Cref{cor:specific-ext-unique,lem:disjoint-downarrows}. Observe that $\ket{\psi_D} \in \algo H_D$, and furthermore the image under $G^{(H)}_{u, \reg{HK}}$ will remain in this space. Therefore, we can write \begin{align}
        \norm{\ket{\psi}_{\reg{HKF}} - G_{u, \reg{HK}}^{(H)} \ket{\psi}_{\reg{HKF}}} &= \sqrt{\sum_{D \in \mathbf D(A)} \norm{\ket{\psi_D}_{\reg{HKF}} - G^{(H)}_{u, \reg{HK}} \ket{\psi_D}_{\reg{HKF}}}^2}.
    \end{align}
    Now, we have $\innerprod{\psi_D}{\psi_D} = \frac{1}{|\mathbf D(A)|}$. Furthermore, observe that the $\reg{KF}$ registers, i.e. \Cref{eqn:psid-second-line}, will have a single set of input-output registers which participate in the chain from $u$ to $v$. If $D_k[u_L \ra z']$ is re-assigned to $z'\neq z$, then this chain will not begin at $u$, rather at $u_L \Vert (u_R \oplus z \oplus z')$. Hence, there is a record of $z$ external to $D_k(u_L)$. We can therefore apply \Cref{lem:recorded-no-comp}, obtaining \begin{align}
        \norm{\ket{\psi_D}_{\reg{HKF}} - G^{(H)}_{u, \reg{HK}} \ket{\psi_D}_{\reg{HKF}}}^2 &= O\left(\frac{1}{|\mathbf D(A)| \cdot (2^n-t)}\right), \\
        \norm{\ket{\psi}_{\reg{HKF}} - G^{(H)}_{u, \reg{HK}} \ket{\psi}_{\reg{HKF}}} &= O\left(\frac{1}{\sqrt{2^n-t}}\right).
    \end{align}
\end{proof}

\begin{lemma}
    Fix some $A \in \mathbf A_t$ and $u \in [N] \setminus \dom(A)$ which $A$ allows, to write the state \begin{align}
        \ket{\psi}_{\reg{HKF}} &\coloneqq \sum_{v \in \mathbf V_{u, A}} \alpha_v \ket{+_{\mathbf D(A[u\ra v])}}_{\reg{HKF}}.
    \end{align}
    Then we have \begin{align}
        \norm{\ket{\psi}_{\reg{HKF}} - G_{u, \reg{HKF}}^{(K)} \ket{\psi}_{\reg{HKF}}} &= O\left(\frac{1}{\sqrt{2^n-t}}\right).
    \end{align}
    \label{lem:canonical-K-close}
\end{lemma}
\begin{proof}
    This proof proceeds similarly to that of \Cref{lem:canonical-H-close}.
\end{proof}

\begin{lemma}
    Fix some $A \in \mathbf A_t$ and $u \in [N] \setminus \dom(A)$ which $A$ allows, to write the state \begin{align}
        \ket{\psi}_{\reg{HKF}} &\coloneqq \sum_{v \in \mathbf V_{u, A}} \alpha_v \ket{+_{\mathbf D(A[u\ra v])}}_{\reg{HKF}}.
    \end{align}
    We take $\alpha_v=0$ for $v\not\in \mathbf V_{u, A}$.
    Then we have \begin{align}
        \norm{\ket{\psi}_{\reg{HKF}} - G_{u, \reg{F}}^{(F)} \ket{\psi}_{\reg{HKF}}} &= O\left(\sqrt{\frac{T}{2^n-t}}\right),
    \end{align}
    where $T=\sum_{l \in \bit^n} T_l$, and $T_l=\abs{\sum_{r \in \bit^n} \alpha_{l \Vert r}}^2$.
    \label{lem:canonical-F-close}
\end{lemma}

\begin{proof}
    We will use the same setup as the proof for \Cref{lem:canonical-H-close}. Let us consider some $D \in \mathbf D(A)$, and take $|A|=t$. We can write the component of $\ket{\psi}$ which is an extension of this state as \begin{align}
            \ket{\psi_D}_{\reg{HKF}} &=\Pi^{\algo H(\mathbf E^3(D))}_{\reg{HKF}} \ket{\psi}_{\reg{HKF}} \\
            &= \frac{1}{\sqrt{|\mathbf D(A)| \cdot (2^n-t)}} \sum_{z \in \bit^n \setminus \{u_R \oplus \dom(D_k)\}} \ket{D_h[u_L \ra z]}_{\reg H} \otimes \\
            &\quad \sum_{v \in \mathbf V_{u, A}} \alpha_v \ket{D_k[u_R \oplus z \ra u_L \oplus v_L], D_f[v_L \ra u_R \oplus z \oplus v_R]}_{\reg{KF}}. \label{eqn:psid-second-line-v2}
        \end{align}
    Further, recall from \Cref{lem:ext-unique} that these states are orthogonal and span the relevant space, so we can write \begin{align}
        \ket{\psi}_{\reg{HKF}} &= \sum_{D \in \mathbf D(A)} \ket{\psi_D}_{\reg{HKF}},
    \end{align}
    so by linearity it suffices to analyze each component separately. Let us define \begin{align}
        \algo H_D = \bigcup_{v \in \mathbf V_{u, A}} \algo H(D[u \ra v]),
    \end{align}
    where we have $\algo H_D \perp \algo H_{D'}$ for $D \neq D' \in \mathbf D(A)$ from \Cref{cor:specific-ext-unique,lem:disjoint-downarrows}. Observe that $\ket{\psi_D} \in \algo H_D$, and furthermore the image under $G^{(F)}_{u, \reg{HKF}}$ will remain in this space. Therefore, we can write \begin{align}
        \norm{\ket{\psi}_{\reg{HKF}} - G_{u, \reg{HKF}}^{(F)} \ket{\psi}_{\reg{HKF}}} &= \sqrt{\sum_{D \in \mathbf D(A)} \norm{\ket{\psi_D}_{\reg{HKF}} - G^{(F)}_{u, \reg{HKF}} \ket{\psi_D}_{\reg{HKF}}}^2}. \label{eqn:split-G-sum}
    \end{align}
    We can analyze the action of $G_{u, \reg{HKF}}^{(F)}$ on this state as \begin{align}
        G_{u, \reg{HKF}}^{(F)} \ket{\psi_D}_{\reg{HKF}} &= \frac{1}{\sqrt{|\mathbf D(A)| \cdot (2^n-t)}} \sum_{\substack{z \in \bit^n \setminus \{u_R \oplus \dom(D_k)\} \nonumber\\
        l \in \mathbf L_{u, A}}} \ket{D_h[u_L \ra z], D_k[u_R \oplus z \ra u_L \oplus l]}_{\reg{HK}} \otimes \\
        &\quad \fc_{l, \reg{F}} \sum_{r \in \bit^n} \alpha_{l \Vert r} \ket{D_f[l \ra u_R \oplus z \oplus r]}_{\reg{F}},
    \end{align}
    from which it is clear to see that \begin{align}
        \norm{\ket{\psi_D} - G^{(F)}_{u} \ket{\psi_D}} &= \frac{1}{\sqrt{|\mathbf D(A)| \cdot (2^n-t)}} \cdot \\&\quad  \sqrt{\sum_{\substack{z \in \bit^n \setminus \{u_R \oplus \dom(D_k)\} \nonumber\\
        l \in \mathbf L_{u, A}}}\norm{\sum_{
        r \in \bit^n} \left(\Id_{\reg F} - \fc_{l, \reg{F}}\right) \alpha_{l\Vert r} \ket{D_f[l \ra u_R \oplus z \oplus r]}_{\reg{F}}}^2} \\
        &= \frac{1}{\sqrt{|\mathbf D(A)| \cdot (2^n-t)}} \cdot O\left(\sqrt{\sum_{l \in \mathbf L_{u, A}} \abs{\sum_{r \in \bit^n} \alpha_{l \Vert r}}^2}\right) \\
        &= O\left(\sqrt{\frac{T}{|\mathbf D(A)| \cdot (2^n-t)}}\right),
    \end{align}
    which combined with \Cref{eqn:split-G-sum} shows the claim.
\end{proof}

\subsection{Cromulent twirling}
\label{sec:twirling}

We now turn to working out properties of the twirling permutations. A number of statistical properties are required of the distribution for our analysis to hold. First, define the left and right action of permutations on databases as
\begin{align}
    \pi * I &\coloneqq \{(\pi(x), y) \, : \, (x, y) \in I\} \\
    I * \omega &\coloneqq \{(x, \omega^{-1}(y)) \, : (x, y) \in I\}.
\end{align}
This action corresponds to the internal database which is supported; if the full masked Feistel ensemble supports $I$, then the internal three rounds of Feistel will support $\pi * I * \omega$. In this way, we will require certain properties of $\pi$ and $\omega$ to ensure that this internal database is allowable, as well as to understand the action of a query.

\begin{definition}
    We say $\pi, \omega$ resolve $I \in \mathbf I_t$ if $\pi * I * \omega \in \mathbf A_t$. We let $\mathbf R_I$ be the set of permutation pairs $(\pi, \omega)$ which resolve $I$. For some fixed $A \in \mathbf A_t$, we let $\mathbf R_{I\ra A}$ be the set of permutation pairs $(\pi, \omega)$ such that $\pi * I * \omega = A$. For some fixed $x \in [N]$, we let $\mathbf R_{x, I}$ be the set of $\pi, \omega$ which resolve $I$ and for which $\pi * I * \omega \heart \pi(x)$, meaning the internal database allows $\pi(x)$.
\end{definition}

Let $\algo D$ be a distribution on $\mathbf S_N^{\times 2}$, and $p_{\algo D}$ be the corresponding probability function. We use the notation $\algo D_I$ to denote the induced distribution of $\algo D$ conditioned on resolving $I$. In other words, \begin{align}
    p_{\algo D_I}(\pi, \omega) &= p_{\algo D}(\pi, \omega \, | \, (\pi, \omega) \in \mathbf R_I).
\end{align}
Similarly, $\algo D_{I \ra A}$ for $A \in \mathbf A$ is the distribution conditioned on satisfying $\pi * I * \omega = A$, and $\algo D_{x, I}$ is the distribution conditioned on the inner distribution allowing $\pi(x)$. Formally, \begin{align}
    p_{\algo D_{I \ra A}}(\pi, \omega) &= p_{\algo D}(\pi, \omega \, | \, (\pi, \omega) \in \mathbf R_{I \ra A}), \\
    p_{\algo D_{x, I}}(\pi, \omega) &= p_{\algo D}(\pi, \omega \, | \, (\pi, \omega) \in \mathbf R_{x, I}).
\end{align}

We say that such a distribution is \emph{cromulent} if it satisfies the following three properties. These are statistical properties about the distribution which will suffice to show that the twirl hides the internal three Feistel rounds.

\begin{definition}
    Distribution $\algo D$ is \emph{cromulent} if, for all $I \in \mathbf I_t$, $x \neq x' \in [N] \setminus \dom(I)$, $l \in \bit^n$, and $y \neq y' \in [N] \setminus \im(I)$, we have \begin{align}
        \Pr_{(\pi, \omega) \sim \algo D_I}[(\pi, \omega) \in \mathbf R_{x, I}] &= 1-O(t^2/2^n), \\
        \Pr_{(\pi, \omega) \sim \algo D_{x, I}}[\omega^{-1}(y)_L = l] &= \frac{1}{2^n},\\
        \Pr_{(\pi, \omega) \sim \algo D_{x, I}}[\omega^{-1}(y)_L =\omega^{-1}(y')_L = l] &= \frac{1}{2^{2n}} \pm O\left(\frac{t^2}{2^{3n}}\right),
    \end{align}
    and for all $(\pi, \omega) \in \mathbf R_{I[x \ra y]}$ and $A=\pi * I * \omega$ we have \begin{align}
        \sqrt{\frac{p_{\algo D_{I[x\ra y]}}(\pi, \omega)}{N-t}} - \sqrt{\frac{p_{\algo D_{x, I}}(\pi, \omega)}{|\mathbf V_{\pi(x), A}|}} &= \pm O\left(\sqrt{\frac{t^2}{2^n}} \cdot \sqrt{\frac{p_{\algo D_{x, I}}(\pi, \omega)}{|\mathbf V_{\pi(x), A}|}}\right)
    \end{align}
    Further, all equations hold in the inverse direction, if we replace $(\pi, \omega)$ with $(\omega^{-1}, \pi^{-1})$, $I$ with $I^{-1}$, $x$ with $y$ and $x'$ with $y'$.
    \label{def:cromulent}
\end{definition}

The first condition captures the intuition that each new query should, with high probability, lead to a non-colliding left wire at the start of the internal Feistel. The second and third point capture properties necessary for compression to work as desired: compression should only remove a chain when the algorithm uncomputes an input-output pair. For plain three-round Feistel, uncomputing only the last $n$ bits of an input would suffice to remove the last element of its chain. These conditions capture that the last $n$ bits are sufficiently scrambled by $\omega$ to render this attack, and anything like it, impossible. The last equation is a technical requirement related to normalization, and is due to small deviations in the number of allowable databases.

Note that the final remark concerning inverses is implied by the previous equations if $\algo D$ is invariant under the flip operation, $(\pi, \omega) \mapsto (\omega^{-1}, \pi^{-1})$. In all of the settings we consider, this will be the case. It therefore suffices to show the properties for the listed cases only.
The following lemmas establish that the uniform distribution $p_{\algo D}(\pi, \omega) = 1/(N!)^2$ is cromulent. We will in turn show that, so long as the twirling permutations are cromulent, the masked Feistel ensemble is indistinguishable from our compressed permutation oracle.

\begin{lemma}
    Fix some $I \in \mathbf I_t$ and $x \in [N] \setminus \dom(I)$. Then we have \begin{align}
        \Pr_{(\pi, \omega) \sim \mathbf R_{I}} [\pi * I * \omega \text{ allows } \pi(x)] &= 1 - O(t^2 / 2^n).
    \end{align}
    \label{lem:twirl-perm-prob}
\end{lemma}

\begin{proof}
    Consider sampling all values $\pi(x)$ for $x \in \dom(I)$, and all values of $\omega$ such that $\pi * I * \omega \in \mathbf A_t$. The remaining values of $\pi$ may be sampled uniformly at random, and so $\pi(x)$ comes from a set of size $2^{2n}-t$. There are $O(t^2 2^n)$ values which are not allowed (from the $O(t^2)$ left substrings times $2^n$ right substrings), so we have \begin{align}
        \Pr_{(\pi, \omega) \sim \mathbf R_{I}} [\pi * I * \omega \text{ allows } \pi(x)] &= \frac{2^{2n}-O(t^22^n)}{2^{2n}-t} \\
        &= 1 - O\left(\frac{t^2}{2^n}\right).
    \end{align}
\end{proof}

\begin{lemma}
    Fix some $I \in \mathbf I_t$, some $x \in [N] \setminus \dom(I)$, and $l \in \bit^n$. Then, for all $y \neq y' \in [N] \setminus \im(I)$, we have \begin{align}
        \Pr_{\substack{\pi, \omega \sim \mathbf R_{x, I}}}[\omega^{-1}(y)_L = l] &= \frac{1}{2^n}, \\
        \Pr_{\substack{\pi, \omega \sim \mathbf R_{x, I}}}[\omega^{-1}(y)_L = \omega^{-1}(y')_L = l] &= \frac{1}{2^{2n}} \pm O\left(\frac{t^2}{2^{3n}}\right).
    \end{align}
    \label{lem:twirl-perm-second-mom}
\end{lemma}

\begin{proof}
    For the first equation, observe that allowability in this sense is translation invariant on the internal left wire: if $\pi, \omega \in \mathbf R_I$, then we can construct $\omega'(v) = \omega(v \oplus (s \Vert 0^n))$ and $\pi'(x) = \pi(x) \oplus (s \Vert 0^n)$. Then, for any $s \in \bit^n$ the newly constructed permutations will also be in $\mathbf R_I$. Therefore, the distribution induced by $\omega^{-1}(y)_L$ is uniform.

    For the second equation, let us analyze the probability that that $\omega^{-1}(y)_L = \omega^{-1}(y')_L$. Consider fixing all of the other input-output points of $\pi$ and $\omega$ such that the internal database is allowable. Then, both $\omega^{-1}(y)_L$ and $\omega^{-1}(y')_L$ are chosen uniformly at random from a universe of size $2^{2n}-t$. Therefore, there will be a collision on the left wire with probability $2^{-n} \pm O(t2^{-2n})$. Combining this with the previous translation invariance, observe that we can randomly shift the constructed $\pi, \omega$ as described above after this procedure. This gives a $2^{-n}$ probability that $\omega(y)=l$, independent of the prior event, leaving us with the result  \begin{align}
        \Pr_{\substack{\pi, \omega \sim \mathbf R_{x, I}}}[\omega^{-1}(y)_L = \omega^{-1}(y')_L = l] &= \frac{1}{2^{2n}} \pm O\left(\frac{t}{2^{3n}}\right).
    \end{align}
    This is tighter than the bound stated, showing the claim.
\end{proof}

\begin{lemma}
    For any $I \in \mathbf I_t$, $x \in [N] \setminus \dom(I)$, $y \in [N] \setminus \im(I)$, and $\pi, \omega \in \mathbf R_{I[x\ra y]}$ with $A=\pi * I * \omega$, we have \begin{align}
        \frac{1}{\sqrt{|\mathbf R_{I[x \ra y]}| \cdot (N-t)}} - \frac{1}{\sqrt{|\mathbf R_{x, I}| \cdot |\mathbf V_{\pi(x), A}|}} &= \pm O\left(\frac{t^2}{2^n \sqrt{|\mathbf R_{I[x\ra y]}| \cdot (N-t)}}\right).
    \end{align}
    \label{lem:twirl-perm-uniform}
\end{lemma}

\begin{proof}
    Observe that we have \begin{align}
        \frac{|\mathbf R_{I[x \ra y]}|}{|\mathbf R_{x, I}|} &= 1 \pm O\left(\frac{t^2}{2^n}\right), & \frac{N-t}{|\mathbf V_{\pi(x), \mathbf A}|} &= 1 \pm O\left(\frac{t^2}{2^n}\right),
    \end{align}
    the first by an argument analogous to \Cref{lem:twirl-perm-prob} and the second by \Cref{lem:allow-prefixes}. If we let $a = |\mathbf R_{I[x \ra y]}| \cdot (N-t)$, then we can write \begin{align}
        \frac{1}{\sqrt{|\mathbf R_{x, I}| \cdot |\mathbf V_{\pi(x), A}|}} &= \frac{1}{\sqrt{a(1 \pm O(t^2/2^n)}} \\
        &= \frac{1}{\sqrt{a}} \pm O\left(\frac{t^2}{2^n \cdot \sqrt{a}}\right),
    \end{align}
    where the last line follows from the taylor expansion of $\frac{1}{\sqrt{a+x}}$ about $x=0$.
\end{proof}

Going forward, we will say that $P = (\pi, \omega, D)$ supports $I$ if $D$ supports $\pi * I * \omega$.

\section{Soundness}
\label{sec:soundness}

In this section, we will establish that our compressed permutation oracle is indistinguishable from a random permutation oracle. We will do this by showing that an oracle for the masked Feistel construction with uniform twirl (which is therefore a uniform permutation oracle) is indistinguishable from our compressed permutation oracle.

In fact, all of the machinery we develop prior to the main theorem will assume only that the masked Feistel has a cromulent twirling distributions. We take $\algo D$ to be the distribution on $\mathbf S_N^{\times 2}$ from which the twirling permutations are drawn. The only part of this section where we assume that $\algo D$ is uniform is in the main \Cref{thm:main-indist}. In \Cref{sec:feist-sec}, we will use a different cromulent distribution to establish the quantum security of the Feistel construction, re-using much of the machinery developed here.

\subsection{Valid, elegant, and sophisticated subspaces}
\label{sec:spaces}

We will begin with the \emph{sophisticated} subspace of $\algo H_{\reg P}$, which will capture purifications corresponding to compressed permutation databases. Towards this, we first define the notation $\ket{\mathbf P(I)}$ for $I \in \mathbf I$, where dependence on $\algo D$ is suppressed, as

\begin{align}
    \ket{\mathbf P(I)} &\coloneqq \sum_{(\pi, \omega) \in \mathbf R_I} \sqrt{p_{\algo D_I}(\pi, \omega)}\ket{\pi, \omega}_{\reg{\Pi\Omega}} \otimes \ket{+_{\mathbf D(\pi * I * \omega)}}_{\reg{HKF}}.
\end{align}

Observe that these states satisfy $\innerprod{\mathbf P(I)}{ \mathbf P(I')} = \delta_{I, I'}$, as for any $I \neq I'$ the two states have distinct supports. The span of these states form a subspace, which we call the \emph{sophisticated} subspace, and refer to by \begin{align}
    \hsoph &\coloneqq \spn{\ket{\mathbf P(I)} \, : \, I \in \mathbf I_{\leq 2^n}}, & \pisoph &\coloneqq \sum_{I \in \mathbf I_{\leq 2^n}} \outerprod{\mathbf P(I)}{\mathbf P(I)}.
\end{align}
Then, we can takes these states to define our intertwining isometry.

\begin{definition}
    We define the \emph{intertwiner} $\iso : \algo H(\mathbf P) \ra \algo H(\mathbf I) \oplus \algo H(\mathbf P)$ as the isometry \begin{align}
        \iso = \Id + \sum_{I \in \mathbf I_{\leq 2^n}} \outerprod{I}{\mathbf P(I)} - \outerprod{\mathbf P(I)}{\mathbf P(I)}.
    \end{align}
\end{definition}

In other words, $\iso$ is an isometry which maps $\ket{\mathbf P(I)}$ to $\ket{I}$, and is the identity on all states in $\hnsoph$. 

As a technical note, we will take the compressed permutation oracle purifying register to be $\algo H(\mathbf I) \oplus \algo H(\mathbf P)$ in this section. We will consider the permutation compression operator $\pc$, the purified query operation $\pu$, and flip operation $\flip$, to act as the identity on the $\algo H(\mathbf P)$ component. We will still label this entire Hilbert space as register $\reg I$, with the expectation that only a small portion of the state will lie in the $\algo H(\mathbf P)$ component. Observe that all of the aforementioned operators (and therefore $\cp$) preserve the two subspaces, in that they can be written as a direct sum $O_1 \oplus O_2$ for operator $O_1$ on $\algo H(\mathbf I)$ and $O_2$ on $\algo H(\mathbf P)$. Therefore, the compressed oracle experiment with this extended purification is perfectly indistinguishable from the standard compressed oracle experiment.\footnote{The reason we do this instead of constructing the isometry from $\algo H(\mathbf I)$ to $\algo H(\mathbf P)$ is that it is unclear how to define validity for the compressed permutation oracle. In the isometric subspace, we will see that it is easy to define, which in turn will be used to conclude validity of our construction.}

Let us observe the simple fact that $\iso$ is an intertwiner between the flip operators $\flip$ and $\mflip$ on the sophisticated subspace, as well as between the purified oracles $\pu$ and $\mfpu$. Observe that the the sophisticated subspace is maintained by the masked Feistel flip operator $\mflip$ and masked Feistel purified query operator $\mfpu$.

\begin{remark}
    We have \begin{align}
        \flip_{\reg I} \cdot \iso_{\reg P} \cdot \pisoph_{\reg P} &= \iso_{\reg P} \cdot \mflip_{\reg{P}} \cdot \pisoph_{\reg P}, & \pu_{\reg{XYI}} \cdot \iso_{\reg P} \cdot \pisoph_{\reg P} &= \iso_{\reg P}\cdot \mfpu_{\reg{XYP}} \cdot \pisoph_{\reg P}.
    \end{align}
    \label{rem:intertwiner-oracle-flip}
\end{remark}

\begin{proof}
    For the first equation, observe that if $P=(\pi, \omega, D_h, D_k, D_f) \in \mathbf P$ supports injective database $I$, then the flipped $(\omega^{-1}, \pi^{-1}, D_f, D_k, D_h)$ supports $I^{-1}$. This holds even outside the sophisticated subspace, so it holds on it as well.

    For the second equation, recall that $\iso_{\reg P}$ sends $\ket{\mathbf P(I)}$ to $\ket{I}$, which determines it's action entirely on the sophisticated subspace. The support of every standard-basis state in the sum defining $\ket{\mathbf P(I)}$ is $I$, and so $\mfpu$ answers queries identically to $\pu$.
\end{proof}

We will define the \emph{valid} subspace of $\algo H_{\reg P} = \algo H(\mathbf P)$ as the span of states where the $\reg{HKF}$ registers lie purely in the valid subspace as defined in \Cref{sec:comp-fs}. We will use the notation $\hval_{\reg P}$ to denote this subspace, and $\piv$ to denote the projector onto it, which can be defined as \begin{align}
    \piv_{\reg P} &\coloneqq \Xi_{\reg H} \Xi_{\reg K} \Xi_{\reg F}.
\end{align}
From \Cref{lem:comp-oracle-valid}, queries to the internal functions preserves this subspace. It follows that $\cmf$ preserves it as well, meaning $\cmf$ can be written of the form $\cmf=O_1 \oplus O_2$ where $O_1 : \hval_{\reg P} \ra \hval_{\reg P}$ and $O_2 : \algo H^{\mathsf{val}, \perp}_{\reg P} \ra \algo H^{\mathsf{val}, \perp}_{\reg P}$ are unitary operators. 

A related subspace is the \emph{query valid} subspace $\algo H^{\mathsf{qval}}_{\reg{XP}}$, this time a subspace of $\algo H_{\reg X} \otimes \algo H_{\reg P}$. The projector onto this subspace is denoted $\piqv$, and the subspace captures the property that, after decompression, the query point will have a well-defined output. To define this, we will first extend the definition of the decompressed projector $\Gamma_x$ to $\algo H(\mathbf P)$, writing \begin{align}
    \Gamma_{x, \reg P} \ket{P}_{\reg P} &\coloneqq \begin{cases}
        0 & \text{(If $\supp(P)(x) = \bot$)} \\
        \ket{P}_{\reg P} & \text{(Otherwise)}
    \end{cases} \\
    \inD_{\reg{XP}} \ket{x}_{\reg X} &\coloneqq \ket{x}_{\reg X} \otimes \Gamma_{x, \reg P}.
\end{align}
We also need the notation $\mathsf{ctrl\mhyphen}\mflip_{\reg{XP}}$ denote the $\mflip_{\reg{P}}$ operator controlled on the $b$ bit of the $\reg X$ register, which can be written \begin{align}
    \mathsf{ctrl\mhyphen}\mflip_{\reg{XP}} \ket{b, x}_{\reg X} &= \ket{b, x}_{\reg X} \otimes \begin{cases}
        \Id_{\reg P} & \text{(If $b=0$)} \\
        \mflip_{\reg P}. & \text{(If $b=1$)}
    \end{cases}
\end{align}
With this in hand, we can define the query valid projector as
\begin{align}
    \piqv_{\reg{XP}} &\coloneqq  (\mathsf{ctrl\mhyphen}\mflip \cdot \mfc \cdot \inD \cdot \mfc^\dagger \cdot \mathsf{ctrl\mhyphen}\mflip^\dagger )_{\reg {XP}}
\end{align}

Observe that $\hqval \leq \algo H_{\reg X} \otimes \hval$, or in words, if the purification register is valid then decompressing on any query point will result in a well defined answer. This has the simple corollary that every state reached by an algorithm lies in the query valid subspace.

\begin{corollary}
    Let $\algo A$ be a quantum algorithm making $q$ queries with intermediate operations $A_0, \dots, A_q$, and define \begin{align}
        \ket{\psi^{(\cmf)}}_{\reg{AXYP}} &\coloneqq A_{q, \reg{AXY}} \cmf_{\reg{XYP}} \dots A_{1, \reg{AXY}} \cmf_{\reg{XYP}} A_{0, \reg{AXY}} \ket{0}_{\reg{AXY}} \ket{\mathbf P(\boti)}_{\reg P}.
    \end{align}
    Then we have $\ket{\psi^{(\cmf)}}_{\reg{AXYP}} \in \piv_{\reg P}$, and thus $\ket{\psi^{(\cmf)}}_{\reg{AXYP}} \in \piqv_{\reg{XP}}$.
    \label{cor:mfeist-qval}
\end{corollary}

Note that the projectors onto valid or query valid do not commute with projectors onto the sophisticated subspace. We would like to focus our attention on states which are both query valid and sophisticated, as these are the ones which we will have tools to analyze. Unfortunately, given the preceding fact, it is difficult to understand this intersection space.

Instead, we will work with a subspace of the sophisticated space which we call the \emph{elegant} subspace, denoted $\hele$. This space will act as a proxy for the intersection of sophisticated and query valid. It is a subspace of $\algo H_{\reg X} \otimes \hsoph$. We write $\piele$ for the projector onto $\hele$, and can define for any $b \in \bit, x \in [N]$

\begin{align}
    \ket{+_{x, \mathbf P(I)}} &\coloneqq \frac{1}{\sqrt{N-|I|}} \sum_{y \in [N] \setminus \im(I)} \ket{{\mathbf P(I[x \ra y])}}, \\
    \piele_{\reg{XP}} \ket{b, x}_{\reg{X}} &\coloneqq \ket{b, x}_{\reg{X}} \otimes \sum_{I \in \mathbf I|^x \,: \, |I| < 2^m} \left(\pisoph - \outerprod{+_{x, \mathbf P(I)}}{+_{x, \mathbf P(I)}}\right)_{\reg P}
\end{align}
To understand this definition, consider a state of the form \begin{align}
    \ket{\psi}_{\reg{XP}} = \ket{0, x}_{\reg X} \left(\alpha_{\bot} \ket{\mathbf P(I)}  +\sum_{y \in [N] \setminus \im(I)} \alpha_y \ket{{\mathbf P(I[x \ra y])}}\right)_{\reg P},
\end{align}
for some $x \in [N]$ and $I \in \mathbf I|^x$ of size less then $2^n$. This state will lie in the elegant subspace if and only if $\sum_{y \in [N] \setminus \im(I)} \alpha_y = 0$. If we replace the $\mathbf P(I)$ and ${\mathbf P(I[x\ra y])}$ with $I$ and $I[x\ra y]$ above (considering a purification register $\algo H(\mathbf I)$), then this is precisely the condition that we will obtain a valid output after applying decompression $\pc^\dagger_{\reg{XP}}$. In this way, the elegant subspace captures states which, after applying the intertwiner $\iso$, will have a query valid answer on decompression $\pc^\dagger_{\reg{XP}}$. 

Note also that we do not consider the value of $b$ in the $\reg X$ register; this subspace will be most relevant when the decompression operator is applied during the compressed masked Feistel. In this scenario, the database will have already been flipped in the case $b=1$. As a result, this subspace is defined assuming a forwards query.
It will also be worth defining a version of the elegant subspace conjugated by the controlled flip, called \emph{flip elegant}. We will define it by its projector. Intuitively, this projector capture the same structure as the elegant projector; the state before applying $\mathsf{ctrl\mhyphen}\mflip$ will be flip elegant (i.e. the input state to $\cmf$), and the state after applying the $\mathsf{ctrl\mhyphen}\mflip$ will be elegant (i.e. the input state to decompression $\mfc^\dagger$). We write \begin{align}
    \pifele_{\reg{XP}} &\coloneqq (\mathsf{ctrl\mhyphen}\mflip \cdot \piele \cdot \mathsf{ctrl\mhyphen}\mflip^\dagger)_{\reg{XP}}.
\end{align}

While this subspace is not exactly the same as the intersection of sophisticated and query valid, it can act as a proxy. We can formalize this using the following lemma.

\begin{restatable}{lemma}{validele}
    Let $\ket{\psi}_{\reg{XP}} \in \algo H_{\reg X} \otimes \hsoph$ be a sophisticated quantum state satisfying $\ket{\psi}_{\reg{XP}} \in  \pisoph_{\reg P}$, and supported on purifications of size at most $t$. If we define \begin{align}
        \ket{\phi}_{\reg{XP}} &\coloneqq \piqv_{\reg{XP}} \ket{\psi}_{\reg{XP}},& 
        \ket{\mu}_{\reg{XP}} &\coloneqq \pifele_{\reg{XP}} \ket{\psi}_{\reg{XP}},
    \end{align}
    then we have \begin{align}
        \norm{\ket{\phi}_{\reg{XP}}-\ket{\mu}_{\reg{XP}}} &= O\left(\sqrt{\frac{t^2}{2^n}}\right).
    \end{align}
    \label{lem:ele-valid}
\end{restatable}

We defer the proof of this lemma until we have more technical machinery.

\subsection{A twirling lemma}

The following lemma will be key in analyzing masked Feistel. It allows us to formalize the intuition that the twirling permutations hide the internal Feistel structure. For intuition, consider the lemma statement without the condition $\sum_{y \in [N]} \alpha_y=0$. Then, we could set $\alpha_y = 2^{-n}$, and we would obtain $\E_{\omega \sim D}[T] = 2^n$. What we show is that, for any orthogonal vector, this quantity is comparatively small. This, in turn, will allow us to show that the masked Feistel compression operator is indeed close to the identity on the subspace where the corresponding compressed permutation oracle acts as the identity.

\begin{lemma}
    Fix some vector $\ket{\alpha} = \sum_{y \in [N]} \alpha_y \ket{y}$ of unit norm, and such that $\sum_{y \in [N]} \alpha_y = 0$. Let $\algo D$ be a distribution such that 
    for all $y \neq y' \in [N]$ with $\alpha_y \neq 0$ and $\alpha_{y'} \neq 0$, and $l \in \bit^n$, \begin{align}
        \Pr_{(\pi, \omega) \sim \algo D}[\omega^{-1}(y)_L = l] &= 2^{-n}, & \Pr_{(\pi, \omega) \sim \algo D}[\omega^{-1}(y)_L = \omega^{-1}(y')_L = l] &= 2^{-2n} \pm O(t^22^{-3n}).
    \end{align}
    Then we have \begin{align}
        \E_{(\pi, \omega) \sim \algo D} [T] &=
        \E_{(\pi, \omega) \sim \algo D} \left[\sum_{l \in \bit^n} \abs{\sum_{r \in \bit^n} \alpha_{\omega(l \Vert r)}}^2 \right] = O(t^2).
    \end{align}
    \label{lem:twirl-perms-JL}
\end{lemma}

\begin{proof}
    Let us compute the expectation of $T_l$, defined as \begin{align}
        T_l &= \abs{\sum_{r \in \bit^n} \alpha_{\omega(l \Vert r)}}^2.
    \end{align}
    We can write $I_y$ to be the indicator variable that is $1$ if $\omega^{-1}(y)_L = l$. If $\alpha_y=0$, then we take $I_y$ to be $1$ with probability $2^{-n}$ independent of the other $I_{y'}$, which we can do without changing the following equation. Then we can write \begin{align}
        \E[T_l] &= \E\left[\abs{\sum_{y \in [N]} I_y\alpha_{y}}^2\right] \\
        &= \E\left[\sum_{y \in [N]} I_y^2 \abs{\alpha_y}^2 + \sum_{y \neq y' \in [N]} I_yI_{y'} \abs{\alpha_y \alpha_{y'}}\right] \\
        &= \sum_{y \in [N]} \E[I_y] \alpha_y\alpha_y^* + \sum_{y \neq y' \in [N]} \E[I_yI_{y'}] \alpha_y \alpha_{y'}^* \\
        &= 2^{-n} \underbrace{\sum_{y \in [N]} \alpha_y \alpha_y^*}_{S_1} + 2^{-2n} \underbrace{\sum_{y \neq y' \in [N]} \alpha_y \alpha_{y'}^*}_{S_2} + \underbrace{\sum_{y \neq y' \in [N]} \pm O(t^22^{-3n})\alpha_y \alpha_{y'}^*}_{S_3},
    \end{align}
    where the last line follows from cromulence of $\algo D$.
    It is clear that $S_1=1$. We can compute $S_2$ as \begin{align}
        S_2 &= \sum_{y \neq y' \in [N]} \alpha_y \alpha_{y'}^* \\
        &= \sum_{y \in [N]} \alpha_y \left(\sum_{y' \in [N] \setminus \{y\}} \alpha_{y'}^*\right) \\
        &= \sum_{y \in [N]} - \alpha_y \alpha_y^* & \text{(By $\sum_{y\in[N]} \alpha_y=0$)} \\
        &= -1.
    \end{align}
    It remains to compute $S_3$. Note that we cannot factor out the $O(t^22^{-3n})$, as the constant pre-factor and sign may be different for different values of $y$ and $y'$. Instead we apply triangle inequality, obtaining \begin{align}
        S_3 &= \sum_{y \neq y' \in [N]} \pm O(t^22^{-3n})\alpha_y \alpha_{y'}^* \\
        &\leq O(t^22^{-3n}) \sum_{y \neq y' \in [N]} |\alpha_y| |\alpha_{y'}^*| &\text{(Triangle inequality)} \\
        &\leq O(t^22^{-3n}) 
        \left(\sum_{y\in [N]} |\alpha_y|\right)^2 \\
        &\leq O(t^22^{-n}).
    \end{align}
    We therefore have $0 \leq \E[T_l] = O(t^22^{-n})$ for all $l \in \bit^n$, and so the claim follows by linearity of expectation, and the identity $T = \sum_{l \in \bit^n} T_l$.
\end{proof}

\subsection{Sanitized and ideal masked Feistel}

We will here define two modification to the masked Feistel oracle, called the \emph{ideal} masked Feistel denoted $\overline \mfc$, and the \emph{sanitized} masked Feistel denoted $\widetilde \mfc$, both of which will be helpful in our analysis. For this section, we will imagine the $\reg X$ register to contain only the query point $\reg X$, i.e. the forward query case. Recall that we answer inverse queries by flipping the truth table, answering as a forward query, and flipping back. In this way, due to the symmetry of both masked Feistel and our compressed permutation oracle, it will suffice to analyze only forwards queries here. 

Let us begin by defining the ideal masked Feistel compression operator $\overline\mfc$. This operator is the identity on the non-elegant subspace. On the elegant subspace, we will first define it for some fixed $I \in \mathbf I|^x$ of size $|I| < 2^n$, defining $\mfc_{x, I}$ as an operator on $\algo H(\{\mathbf P(I') \, : \, I' \in I|^x\})$.
\begin{align}
    \overline \mfc_{x, I} &\coloneqq \Id-\proj{\mathbf P(I)}-\proj{+_{x, \mathbf P(I)}} +\outerprod{\mathbf P(I)}{+_{x, \mathbf P(I)}} + \outerprod{+_{x, \mathbf P(I)}}{\mathbf P(I)}\end{align}
The masked Feistel decompression and its coherently controlled counterpart are then naturally defined as
\begin{align}
    \overline\mfc_{x, \reg{P}} &\coloneqq  \bigoplus_{I \in \mathbf I|^x, |I|<2^n} \overline \mfc_{x, I, \reg P}, &
    \overline \mfc_{\reg{XP}} \ket{x}_{\reg X} &\coloneqq \ket{x}_{\reg X} \otimes \overline \mfc_{x, \reg P}.
\end{align}

Note the clear similarity to the definition of the permutation compression operator $\pc$; the only difference being the use of $\mathbf P(I)$ in place of $I$, reflecting the distinct---but isomorphic---Hilbert spaces on which these operators act. This can be formalized via the following two remarks.

\begin{remark}
    The operator $\iso$ intertwines ideal masked Feistel decompression $\overline \mfc^\dagger$ on the elegant subspace with decompression $\pc^\dagger$ on the corresponding subspace. \begin{align}
        \iso_{\reg P} \cdot\overline \mfc^\dagger_{\reg P}\cdot\piele_{\reg P} = \pc^\dagger_{\reg P} \cdot \iso_{\reg P}\cdot \piele_{\reg P}
    \end{align}
    \label{rem:intertwiner-comp}
\end{remark}

\begin{remark}
    The action of $\overline \mfc_{\reg{XP}}$ interchanges the elegant space and the intersection of $\inD$ and the sophisticated space, \begin{align}
        \piele_{\reg{XP}} \overline \mfc_{\reg{XP}}  = \overline \mfc_{\reg{XP}}  (\inD\pisoph)_{\reg{XP}}.
    \end{align}
    \label{rem:ideal-action}
\end{remark}

We use this to define the ideal masked feistel oracle $\overline \cmf$, which acts as \begin{align}
    \overline \cmf_{\reg{XYP}} &\coloneqq \mathsf{ctrl\mhyphen}\mflip_{\reg{XP}} \cdot \overline \mfc_{\reg{XP}} \cdot \mfpu_{\reg{XYP}} \cdot \overline \mfc_{\reg{XP}}^\dagger \cdot \mathsf{ctrl\mhyphen}\mflip_{\reg{XP}}^\dagger.
\end{align}

The following remark is based on similar reasoning to the prior.

\begin{remark}
    The ideal compressed mask Feistel operator is exactly intertwined by $\iso$ with the compressed permutation oracle, on the flip elegant subspace.
    \begin{align}
        \iso_{\reg P} \cdot \overline \cmf_{\reg{XYP}}\cdot \pifele_{\reg P} &= \cp_{\reg{XYP}} \cdot \iso_{\reg P} \cdot \pifele_{\reg P}.
    \end{align}
    \label{rem:ele-ideal-int}
\end{remark}

The ideal operator is what we are aiming for, in that its action is precisely analogous to the compressed permutation oracle up to a change of basis. To show that the action of the ideal $\overline \cmf$ is close the actual $\cmf$, we will define the sanitized compression as an intermediate operation. This operator is similar to standard compression, except it will maintain canonicity of the internal $D_h, D_k, D_f$ tuple of databases. We define \begin{align}
    \widetilde\mfc_{x, \reg{P}} \ket{\pi}_{\reg \Pi} &\coloneqq \ket{\pi}_{\reg \Pi} \otimes  G_{\pi(x), \reg{HKF}}.
\end{align}

Observe that this operator is close the masked Feistel compression operator.

\begin{corollary}[of \Cref{lem:canonical-decomp}]
    The sanitized and unsanitized masked Feistel compression operators are close: $\norm{\widetilde\mfc - \mfc}_{\leq t} = O\left(\sqrt{\frac{t^2}{2^n}}\right)$.
    \label{cor:close-mfc}
\end{corollary}



We will therefore be able to replace the unsanitized $\mfc$ with its sanitized variant $\widetilde \mfc$.
It will further be helpful to define the \emph{allowability} projector $\Pi^{\heart}_{x, \reg{P}}$. This projects onto the subspace in which the query point $x$ is allowed by the purification, i.e. \begin{align}
    \Pi^{\heart}_{x, \reg{P}} &= \sum_{\pi, \omega, D \, : \, D \heart \pi(x)} \outerprod{\pi, \omega, D}{\pi, \omega, D}_{\reg P}.
\end{align}
Naturally, the coherently controlled version is defined as $\Pi^{\heart}_{\reg{XP}} \ket{x}_{\reg X} = \ket{x}_{\reg X} \otimes \Pi^{\heart}_{x, \reg{P}}$. A useful fact about this projector is that it acts gently on the sophisticated subspace when the input point is undefined, which follows from the cromulence of $\algo D$.

\begin{lemma}
    We have \begin{align}
        \norm{(\Id - \Pi^\heart_{x}) \xninD{x} \pisoph }_{\leq t} &= O\left(\sqrt{\frac{t^2}{2^n}}\right).
    \end{align}
    \label{lem:heart-norm}
\end{lemma}

\begin{proof}
    First, let us consider some fixed $I \in \mathbf I_{\leq 2^n}$ where $x \not\in \dom(I)$ (any state orthogonal to this family is annihilated by $\xninD{x}\pisoph$). We can write \begin{align}
        \ket{\mathbf P(I)}_{\reg P} &=  \sum_{(\pi, \omega) \in \mathbf R_I} \sqrt{p_{\algo D_I}(\pi, \omega)}\ket{\pi, \omega}_{\reg{\Pi\Omega}} \otimes \ket{\mathbf D(\pi * I * \omega)}_{\reg{HKF}}, \\
        \norm{\Pi^\heart_{x, \reg P}\ket{\mathbf P(I)}_{\reg P} - \ket{\mathbf P(I)}_{\reg P}} &= \norm{ \sum_{\substack{(\pi, \omega) \in \mathbf R_I, \\
        (\pi, \omega) \not\in \mathbf R_{x, I}
        }} \sqrt{p_{\algo D_I}(\pi, \omega)} \ket{\pi, \omega}_{\reg{\Pi\Omega}} \otimes \ket{\mathbf D(\pi * I * \omega)}_{\reg{HKF}}} \label{eqn:heart-norm}\\
        &= \sqrt{\Pr_{(\pi, \omega) \sim \algo D_{I}} [(\pi, \omega) \not\in \mathbf R_{x, I}]} \\
        &= O\left(\sqrt{\frac{t^2}{2^n}}\right), \qquad\qquad\qquad \text{(By cromulence)}
    \end{align}
    in particular the property established in \Cref{lem:twirl-perm-prob}.
    Further, observe that the state in \Cref{eqn:heart-norm} is supported entirely on purifications $P$ with the property $\supp(P)=I$. Therefore, if we had considered a different database $I$, we would have an orthogonal state. The claim then follows from \Cref{lem:dirsum-norm}.
\end{proof}

Now, we will show closeness of ideal and real in the case where the purification does not support the input point.

\begin{lemma}
    We have \begin{align}
        \norm{(\overline \mfc^\dagger - \mfc^\dagger)_{\reg{XP}} \cdot(\ninD \piele)_{\reg{XP}}}_{\leq t} &= O\left(\sqrt{\frac{t^2}{2^n}}\right).
    \end{align}
    \label{lem:cmf-undef-act}
\end{lemma}

\begin{proof}
We begin with a simpler claim.
\begin{claim*}
    Let $x \in [N]$ and $I \in \mathbf I|^x$ be of size $t=|I| < 2^n$. Then we have \begin{align}
        \norm{\widetilde \mfc_{x}^\dagger \cdot \Pi^{\heart}_{x} \ket{\mathbf P(I)}- \ket{+_{x, \mathbf P(I)}}} &= O\left(\sqrt{\frac{t^2}{2^n}}\right).
    \end{align}
    Further, $\widetilde \mfc_{x} \cdot \Pi^{\heart}_{x} \ket{\mathbf P(I)}$ is supported entirely on the span of $\ket{\mathbf P(I[x \ra y])}$ for $y \in [N] \setminus \im(I)$.
\end{claim*}

    To prove this claim, observe that the state after the projection is
    \begin{align}
            \Pi^\heart_{x} \ket{\mathbf P(I)} &=  \sum_{(\pi, \omega) \in \mathbf R_{x, I}} \sqrt{p_{\algo D_I}(\pi, \omega)} \ket{\mathbf D(\pi * I * \omega)}_{\reg{HKF}} \otimes \ket{\pi, \omega}_{\reg{\Pi\Omega}} \\
            &= \left(1-O\left(\frac{t^2}{2^n}\right)\right) \cdot \sum_{(\pi, \omega) \in \mathbf R_{x, I}} \sqrt{p_{\algo D_{x, I}}(\pi, \omega)} \ket{\mathbf D(\pi * I * \omega)}_{\reg{HKF}} \otimes \ket{\pi, \omega}_{\reg{\Pi\Omega}},
        \end{align}
        by cromulence.
        Applying the $\widetilde\mfc_x$, we obtain \begin{align}
            \widetilde\mfc_x & \cdot \Pi_{\heart x} \ket{\mathbf P(I)} = \left(1-O\left(\frac{t^2}{2^n}\right)\right) \cdot\sum_{(\pi, \omega) \in \mathbf R_{x, I}} \sqrt{p_{\algo D_{x, I}}(\pi, \omega)} \left(G_{\pi(x)}^\dagger\ket{\mathbf D(\pi * I * \omega)}\right)_{\reg{HKF}} \otimes \ket{\pi, \omega}_{\reg{\Pi\Omega}} \\
            &= \left(1-O\left(\frac{t^2}{2^n}\right)\right) \cdot \sum_{\substack{(\pi, \omega) \in \mathbf R_{x, I}, \\ A \coloneqq \pi * I * \omega
            }} \left(\sum_{v \in \mathbf V_{\pi(x), A}} \sqrt{\frac{{p_{\algo D_{x, I}}(\pi, \omega)}}{{|\mathbf V_{\pi(x), A}|}}} \ket{\mathbf D(A[\pi(x) \ra v])}\right)_{\reg{HKF}} \otimes \ket{\pi, \omega}_{\reg{\Pi\Omega}} \\
            &=\sum_{\substack{y \in [N]\setminus \dom(I), \\ (\pi, \omega) \in \mathbf R_{I[x \ra y]}, \\ A \coloneqq \pi * I * \omega
            }} \left(1\pm O\left(\sqrt{\frac{t^2}{2^n}}\right)\right) \cdot \sqrt{\frac{{p_{\algo D_{I[x \ra y]}}(\pi, \omega)}}{{N-t}}} \cdot \ket{\mathbf D(A[\pi(x) \ra \omega^{-1}(y)])}_{\reg{HKF}} \otimes \ket{\pi, \omega}_{\reg{\Pi\Omega}},
        \end{align}
        where the second line follows from \Cref{lem:canonical-decomp}, and the last line from cromulence, in particular the property established in \Cref{lem:twirl-perm-uniform}.
        Observe that the above state would, without the $1\pm O\left(\frac{t^2}{2^n}\right)$ term, be precisely $\ket{+_{x, \mathbf P(I)}}$. The two claims follow.

    With this in hand, recall that any state orthogonal to $\ket{\mathbf P(I)}$ for $I \in \mathbf I|^x$ will be annihilated by $\ninD\piele$. Further, by definition we have \begin{align}
        \overline \mfc_{x, \reg{P}} \ket{\mathbf P(I)} &\coloneqq \ket{+_{x, \mathbf P(I)}}.
    \end{align}
    Thus, for any state in the family $\ket{x}_{\reg X} \ket{\mathbf P(I)}_{\reg P}$, we have that the images of $\widetilde\mfc^\dagger_x$ and of $\ol\mfc^\dagger$ lie in the span of $\ket{x}_{\reg X} \ket{\mathbf P(I[x \ra y])}_{\reg P}$. We would obtain an orthogonal subspace for any different choice if $x \in [N]$ and $I \in \mathbf I|^x$, so these operators are close for any state in the linear combination. Putting everything together, we have \begin{align}
        &\norm{(\overline \mfc^\dagger - \mfc^\dagger)_{\reg XP} \cdot(\ninD \piele)_{\reg{XP}}}_{\leq t} \leq \underbrace{\norm{(\widetilde \mfc^\dagger - \mfc^\dagger)_{\reg XP} \cdot(\ninD \piele)_{\reg{XP}}}_{\leq t}}_{O(\sqrt{t^2/2^n})} + \nonumber\\
        &\qquad\qquad \underbrace{\norm{(\widetilde \mfc^\dagger \Pi^\heart - \widetilde \mfc^\dagger)_{\reg XP} \cdot(\ninD \piele)_{\reg{XP}}}_{\leq t}}_{O(\sqrt{t^2/2^n})} + \underbrace{\norm{(\overline \mfc^\dagger - \widetilde \mfc^\dagger\Pi^\heart)_{\reg XP} \cdot(\ninD \piele)_{\reg{XP}}}_{\leq t}}_{O(\sqrt{t^2/2^n})} \\
        &\qquad\qquad=O\left(\sqrt{\frac{t^2}{2^n}}\right).
    \end{align}
\end{proof}

Another important case is the portion of the elegant subspace where the query point is already defined prior to decompression.

\begin{lemma}
    We have \begin{align}
        \norm{(\overline \mfc^\dagger - \mfc^\dagger)_{\reg{XP}} \cdot (\inD\piele)_{\reg{XP}}}_{\leq t} &= O\left(\sqrt{\frac{t^2}{2^n}}\right).
    \end{align}
    \label{lem:cmf-def-act}
\end{lemma}


\begin{proof}
    Let $\ket{\psi}_{\reg{XP}} \in \algo H_{\reg X} \otimes \hele$ be an elegant quantum state, where the $\reg P$ register is supported on databases of size at most $t$, and the query point is defined: $\ket{\psi}_{\reg{XP}} \in \inD_{\reg{XP}}$.
    Note that the entire operation is controlled on the $\reg X$ register, so by \Cref{lem:dirsum-norm} it suffices to do the analysis for some fixed $x \in [N]$. Further, we can replace $\mfc^\dagger$ by $\widetilde \mfc^\dagger$ by \Cref{cor:close-mfc}.  The action of ideal decompression $\ol \mfc^\dagger$ on this subspace is the identity, which we would like to show is close the action of $\widetilde \mfc^\dagger$. We will drop the dagger ${}^\dagger$ from $\widetilde \mfc$, observing that if $U^\dagger$ is close to $\Id$ on some subspace then by linearity so is $U$. In other words, we will show for any fixed $x\in[N]$ that $\norm{\widetilde \mfc_{x} \ket{\psi}- \ket{\psi}}_{\leq t} = O\left(\sqrt{\frac{t^2}{2^n}}\right)$ for any elegant $\ket{\psi} \in \hele$ with $\ket{\psi} \in \xinD{x}$.

    Let us expand $\widetilde \mfc$ into it's constituent parts. We will begin with the $G^{(F)}$ application, considering \begin{align}
        \norm{\left(\sum_{\pi \in \mathbf S_N} \outerprod{\pi}{\pi}_{\reg \Pi} \otimes G^{(F)}_{\pi(x), \reg{HKF}} - \Id_{\reg{P}}\right) \ket{\psi}_{\reg{P}}}.
    \end{align}
    Observe that both states (the one after applying the left or the right operator in the above parenthesis) will be supported entirely on purifications $P$ such that $\supp(P)|^x = I|^x$. This is because a re-assignment to the final element of the chain beginning at $\pi(x)$ (in $D_f$) can only modify the chain beginning at $\pi(x)$. 
    Let us make then the assumption that the state $\ket{\psi}$ is supported purely on states of the form $\ket{\mathbf P(I[x \ra y])}$ for $I \in \mathbf I |^x$ of size $|I|=t-1$. The other components of $\ket{\psi}$ will remain orthogonal, so by \Cref{lem:dirsum-norm} it suffices to do the analysis on this component only. We will write \begin{align}
        \ket{\psi} = \sum_{y \in [N] \setminus \im(I)} \alpha_y \ket{\mathbf P(I[x \ra y]}_{\reg P},
    \end{align}
    where $\sum_{y \in [N] \setminus \im(I)} \alpha_y = 0$. We can equivalently write \begin{align}
        \ket{\psi}_{\reg P} &= \sum_{\substack{\pi, \omega \in \mathbf R_{x, I}, \\
            A \coloneqq \pi * I * \omega, \\
            u \coloneqq \pi(x)}}\sqrt{p_{\algo D_{x, I}}(\pi, \omega)} \ket{\pi, \omega}_{\reg{\Pi\Omega}}\otimes \left(\sum_{\substack{v \in \mathbf V_{u,A}}} \alpha_{\omega(v)} \ket{+_{\mathbf D(A[u\ra v])}}_{\reg{HKF}}\right). 
    \end{align}
    Applying the $G^{(F)}$ operator, we find
    \begin{align}
        \left(\sum_{\pi \in \mathbf S_N} \outerprod{\pi}{\pi}_{\reg \Pi} \otimes G^{(F)}_{\pi(x), \reg{HKF}}\right)\ket{\psi}_{\reg P} &= \sum_{\substack{\pi, \omega \in \mathbf R_{x, I}, \\
            A \coloneqq \pi * I * \omega, \\
            u \coloneqq \pi(x)}}\sqrt{p_{\algo D_{x, I}}(\pi, \omega)} \ket{\pi, \omega}_{\reg{\Pi\Omega}}\otimes\nonumber\\&\quad \left( G_{u, \reg{HKF}}^{(F)}\sum_{\substack{v \in \mathbf V_{u,A}}} \alpha_{\omega(v)} \ket{+_{\mathbf D(A[u\ra v])}}_{\reg{HKF}}\right),
    \end{align}
    giving us difference \begin{align}
        \norm{\left(\sum_{\pi \in \mathbf S_N} \outerprod{\pi}{\pi}_{\reg \Pi} \otimes G^{(F)}_{\pi(x), \reg{HKF}} - \Id_{\reg{P}}\right) \ket{\psi}_{\reg{P}}} &= \Bigg \Vert \sum_{\substack{\pi, \omega \in \mathbf R_{x, I}, \\
            A \coloneqq \pi * I * \omega, \\
            u \coloneqq \pi(x)}}\sqrt{p_{\algo D_{x, I}}(\pi, \omega)} \ket{\pi, \omega}_{\reg{\Pi\Omega}}\otimes\nonumber\\&\quad \left( \left(G_{u, \reg{HKF}}^{(F)} - \Id_{\reg{HKF}}\right)\sum_{\substack{v \in \mathbf V_{u,A}}} \alpha_{\omega(v)} \ket{+_{\mathbf D(A[u\ra v])}}_{\reg{HKF}}\right) \Bigg\Vert \\
            &= \sqrt{\sum_{\pi, \omega \in \mathbf R_{x, I}}{p_{\algo D_{x, I}}(\pi, \omega)} \cdot T_{\omega, \alpha}}
    \end{align}
    where $T_{\omega, \alpha} = \sum_{l \in \bit^n} T_{l, \omega, \alpha}$ and $T_{l, \omega, \alpha} = \left(\sum_{r \in \bit^n} \abs{\alpha_{\omega(l \Vert r)}}\right)^2$. This last line follows from \Cref{lem:canonical-F-close}. This can be rewritten as \begin{align}
        \norm{\left(\sum_{\pi \in \mathbf S_N} \outerprod{\pi}{\pi}_{\reg \Pi} \otimes G^{(F)}_{\pi(x), \reg{HKF}} - \Id_{\reg{P}}\right) \ket{\psi}_{\reg{P}}}
        &= \sqrt{\E_{(\pi, \omega) \sim \algo D_{x, I}}\left[\frac{T_{\omega, \alpha}}{2^n-t}\right]} \label{eqn:last-F-step1} \\
        &= O\left(\sqrt{\frac{t^2}{2^n}}\right), \label{eqn:last-F-step2}
    \end{align}
    where the last line follows from \Cref{lem:twirl-perms-JL} combined with cromulence, and in particular \Cref{lem:twirl-perm-second-mom}. Therefore, we may replace the application of $G^{(F)}$ in $\widetilde \mfc$ with the identity $\Id$ at the above cost. A similar argument shows that the application of $G^{(K)}$ and $G^{(H)}$ can be replaced with the identity $\Id$ at a cost $O(\sqrt{t^2/2^n})$, using \Cref{lem:canonical-K-close} and \cref{lem:canonical-H-close} repsectively in place of \Cref{lem:canonical-F-close}. Indeed, for these cases the final step corresponding to \Cref{eqn:last-F-step1,eqn:last-F-step2} can be ommitted, due to the simpler form of  \Cref{lem:canonical-H-close,lem:canonical-K-close}.
\end{proof}

\begin{corollary}[of \Cref{lem:cmf-undef-act,lem:cmf-def-act}]
    We have \begin{align}
        \norm{(\overline\mfc^\dagger - \mfc^\dagger)\piele}_{\leq t} = O\left(\sqrt{\frac{t^2}{2^n}}\right), \\
        \norm{(\overline\mfc - \mfc)\inD\pisoph}_{\leq t} = O\left(\sqrt{\frac{t^2}{2^n}}\right).
    \end{align}
    \label{cor:ideal-cmf-act}
\end{corollary}

\begin{proof}
    The first equation is immediate from the referenced lemmas. 
    To conclude the second, by \Cref{cor:close-mfc} we may substitute $\widetilde \mfc$ for $\mfc$. Then, observe that \begin{align}
        \norm{(\overline\mfc - \widetilde\mfc)\inD\pisoph}_{\leq t} &= \norm{(\overline\mfc - \widetilde\mfc)\overline\mfc^\dagger \cdot \piele \cdot \overline\mfc}_{\leq t} \\
        &= \norm{\piele \cdot \overline\mfc - \widetilde\mfc \widetilde \mfc^\dagger \cdot \piele \cdot \overline\mfc}_{\leq t} + O\left(\sqrt{\frac{t^2}{2^n}}\right) \\
        &= O\left(\sqrt{\frac{t^2}{2^n}}\right).
    \end{align}
\end{proof}

We can now prove the main result for this section, establishing the soundness of our intertwiner on the elegant subspace.

\begin{lemma}
    On the flip elegant subspace, $\iso$ is an approximate intertwiner between the compressed permutation oracle and the compressed masked Feistel oracle,
    \begin{align}
        \norm{(\iso_{\reg P} \cdot \cmf_{\reg{XYP}} - \cp_{\reg{XYI}} \cdot \iso_{\reg P}) \pifele_{\reg P}}_{\leq t} &= O\left(\sqrt{\frac{t^2}{2^n}}\right).
    \end{align}
    \label{cor:ideal-cmf-int}
\end{lemma}

\begin{proof}
    We can write \begin{align}
        &\hspace{-0.5in}\cmf_{\reg{XYP}} \pifele_{\reg{XP}} = \mathsf{ctrl\mhyphen}\mflip_{\reg{XP}} \cdot \mfc_{\reg{XP}} \cdot \mfpu_{\reg{XYP}} \cdot \mfc^\dagger_{\reg{XP}} \cdot \mathsf{ctrl\mhyphen}\mflip_{\reg{XP}}^\dagger \cdot \pifele_{\reg{XP}} \\
        &\hspace{-0.5in}= \mathsf{ctrl\mhyphen}\mflip_{\reg{XP}} \cdot \mfc_{\reg{XP}} \cdot \mfpu_{\reg{XYP}} \cdot \mfc^\dagger_{\reg{XP}}\cdot \piele_{\reg{XP}} \cdot \mathsf{ctrl\mhyphen} \mflip_{\reg{XP}}^\dagger\cdot \pifele_{\reg{XP}} \\
        &\hspace{-0.5in}= \mathsf{ctrl\mhyphen}\mflip_{\reg{XP}} \cdot \mfc_{\reg{XP}} \cdot \mfpu_{\reg{XYP}} \cdot \overline \mfc^\dagger_{\reg{XP}}\cdot \piele_{\reg{XP}} \cdot \mathsf{ctrl\mhyphen} \mflip_{\reg{XP}}^\dagger\cdot \pifele_{\reg{XP}} + O\left(\sqrt{\frac{t^2}{2^n}}\right) & \text{(\Cref{cor:ideal-cmf-act})} \nonumber  \\
        &\hspace{-0.5in}= \mathsf{ctrl\mhyphen}\mflip_{\reg{XP}} \cdot \mfc_{\reg{XP}} \cdot  (\inD\pisoph)_{\reg{XP}} \cdot\mfpu_{\reg{XYP}} \cdot \overline \mfc^\dagger_{\reg{XP}}\cdot \mathsf{ctrl\mhyphen} \mflip_{\reg{XP}}^\dagger\cdot \pifele_{\reg{XP}} + O\left(\sqrt{\frac{t^2}{2^n}}\right) & \text{(\Cref{rem:ideal-action})} \nonumber \\
        &\hspace{-0.5in}= \mathsf{ctrl\mhyphen}\mflip_{\reg{XP}} \cdot \overline\mfc_{\reg{XP}} \cdot  (\inD\pisoph)_{\reg{XP}} \cdot\mfpu_{\reg{XYP}} \cdot \overline \mfc^\dagger_{\reg{XP}}\cdot \mathsf{ctrl\mhyphen} \mflip_{\reg{XP}}^\dagger\cdot \pifele_{\reg{XP}} + O\left(\sqrt{\frac{t^2}{2^n}}\right) & \text{(\Cref{cor:ideal-cmf-act})} \nonumber \\
        &\hspace{-0.5in}= \mathsf{ctrl\mhyphen}\mflip_{\reg{XP}} \cdot \overline\mfc_{\reg{XP}}  \cdot\mfpu_{\reg{XYP}} \cdot \overline \mfc^\dagger_{\reg{XP}}\cdot \mathsf{ctrl\mhyphen} \mflip_{\reg{XP}}^\dagger\cdot \pifele_{\reg{XP}} + O\left(\sqrt{\frac{t^2}{2^n}}\right) \nonumber \\
        &\hspace{-0.5in}= \overline \cmf_{\reg{XYP}} \pifele_{\reg{XP}} + O\left(\sqrt{\frac{t^2}{2^n}}\right).
    \end{align}
    The claim now follows from \Cref{rem:ele-ideal-int}.
\end{proof}

The final case is decompression on states which are sophisticated but not elegant. For a technical reason, we will show that the decompression operator results in an invalid output and non-sophisticated purification. This will be useful in bounding the component of the state on this subspace. We use $\algo H_A / \algo H_B$ for spaces $\algo H_B \leq \algo H_A$ to denote the orthogonal complement of $\algo H_B$ in $\algo H_A$.

\begin{lemma}
    Let $\ket{\psi}_{\reg{XP}} \in \algo H_{\reg X} \otimes (\hsoph / \hele)$ be a sophisticated, but not elegant, quantum state supported on purifications of size at most $t$. Then we have \begin{align}
        \norm{\pisoph_{\reg{XP}} \cdot \widetilde \mfc_{\reg{XP}}^\dagger \cdot \ket{\psi}_{\reg{XP}}} &= O\left(\sqrt{\frac{1}{2^n}}\right), & \norm{\inD_{\reg{XP}} \cdot \widetilde \mfc_{\reg{XP}}^\dagger \cdot \ket{\psi}_{\reg{XP}}} &= O\left(\sqrt{\frac{1}{2^n}}\right).
    \end{align}
    \label{lem:cmf-nele-act}
\end{lemma}

\begin{proof}
    Observe that that $\ket{\psi}_{\reg{XP}} \in \inD_{\reg{XP}}$, the orthogonal complement of $\inD_{\reg{XP}}$ in the sophisticated space $\hsoph$ lies in the elegant subspace $\hele$. By definition, we can write the sanitized masked Feistel decompression as \begin{align}
        \widetilde \mfc_{\reg{XP}}^\dagger &= \sum_{\pi \in \mathbf S_N} \outerprod{\pi}{\pi}_{\reg \Pi} \otimes \left(G^{(F)}_{\pi(x), \reg{HKF}} G^{(K)}_{\pi(x), \reg{HKF}} G^{(H)}_{\pi(x), \reg{HK}}\right),
    \end{align}
    where we recall that the $G^{(H)},G^{(K)}, G^{(F)}$ unitaries are involutions. By an argument analogous to that given in the proof of \Cref{lem:cmf-def-act}, the first two operations applied to the state, $G^{(H)},G^{(K)}$, perturb the state a distance at most $O(\sqrt{1/2^n})$. By triangle inequality, we may therefore consider only the action of of the final $G^{(F)}$ operator. For now, let us fix some $x \in [N]$ to be the value in the $\reg{X}$ register, and $I \in \mathbf I|^x$ of size $|I|=t-1$, and work out the action on the $\ket{+_{x, \mathbf P(I)}}_{\reg P}$ state. We find 
    {
    \allowdisplaybreaks
    \begin{align}
        \left(\sum_{\pi \in \mathbf S_N} \outerprod{\pi}{\pi}_{\reg \Pi} \otimes G^{(F)}_{\pi(x), \reg{HKF}}\right) \ket{+_{x, \mathbf P(I)}}_{\reg P} &= \sum_{\substack{\pi, \omega \in \mathbf R_{x, I}, \\
            A \coloneqq \pi * I * \omega, \\
            u \coloneqq \pi(x)}}\sqrt{p_{\algo D_{x, I}}(\pi, \omega)} \ket{\pi, \omega}_{\reg{\Pi\Omega}}\otimes\nonumber\\&\quad \left(G_{u, \reg{HKF}}^{(F)}\sum_{\substack{v \in \mathbf V_{u,A}}} \frac{1}{\sqrt{|\mathbf V_{u, A}|}} \ket{+_{\mathbf D(A[u\ra v])}}_{\reg{HKF}}\right)
    \end{align} \vspace{-0.3in}\begin{align}
        &= \sum_{\substack{\pi, \omega \in \mathbf R_{x, I}, \\
            A \coloneqq \pi * I * \omega, \\
            u \coloneqq \pi(x)}}\sqrt{p_{\algo D_{x, I}}(\pi, \omega)} \ket{\pi, \omega}_{\reg{\Pi\Omega}}\otimes\nonumber\\&\quad G_{u, \reg{HKF}}^{(F)}\Bigg(\sum_{\substack{D \in \mathbf D(A), \\
            z \in H(u, D), \\
            w \in K(u, D), \\
            r \in \bit^n}} \frac{1}{\sqrt{|\mathbf D(A)| \cdot |H(u, D)| \cdot |K(u, D)| \cdot 2^n}}\cdot \nonumber\\&\quad \ket{D_h[u_L \ra z]}_{\reg{H}} \ket{D_k[u_R\oplus z \ra w]}_{\reg K} \ket{D_f[u_L \oplus w \ra r]}_{\reg F}\Bigg) \\
        &= \sum_{\substack{\pi, \omega \in \mathbf R_{x, I}, \\
            A \coloneqq \pi * I * \omega, \\
            u \coloneqq \pi(x)}}\sqrt{p_{\algo D_{x, I}}(\pi, \omega)} \ket{\pi, \omega}_{\reg{\Pi\Omega}}\otimes\nonumber\\&\quad \Bigg(\sum_{\substack{D \in \mathbf D(A), \\
            z \in H(u, D), \\
            w \in K(u, D)}} \frac{1}{\sqrt{|\mathbf D(A)| \cdot |H(u, D)| \cdot |K(u, D)|}}\cdot \nonumber\\&\quad \ket{D_h[u_L \ra z]}_{\reg{H}} \ket{D_k[u_R\oplus z \ra w]}_{\reg K} \ket{D_f}_{\reg F}\Bigg).
    \end{align}
    }
    The key properties of this state is that it is fully undefined on input $x$, as can be seen by the lack of assignment of $u_L \oplus w$ in $D_f$; it is thus annihilated by $\inD_{\reg{XP}}$. Further, there is a length two rightward semi-chain beginning at $u$ which is not completed, in every term in the superposition. The databases are non-minimal, and the state is therefore annihilated by $\pisoph_{\reg{XP}}$. Further, the actual state $\ket{\psi}$ is a linear combination of states of this form, so we have \begin{align}
        \norm{\inD_{\reg{XP}} \cdot \left(\sum_{\pi \in \mathbf S_N} \outerprod{\pi}{\pi}_{\reg \Pi} \otimes G^{(F)}_{\pi(x), \reg{HKF}}\right) \ket{\psi}_{\reg{XP}}} &= 0,
    \end{align}
    and similarly for the sophisticated projector.
    The claim now follows from triangle inequality.
\end{proof}

\subsection{Preserved subspaces}

We are now ready to show that the elegant subspace of sophisticated provides a faithful proxy of the query-valid subspace, proving \Cref{lem:ele-valid}. We restate it below.

\validele*

\begin{proof}
    We can write \begin{align}
        \norm{\ket{\phi} - \ket{\mu}} &= \norm{\piqv\ket{\mu} + \piqv (\ket{\psi} - \ket{\mu}) -\ket{\mu}} \\
        &\leq \norm{\piqv\ket{\mu} -\ket{\mu}} + \norm{\piqv (\ket{\psi} - \ket{\mu})} \\
        &= \norm{\pinqv\ket{\mu}} + \norm{\piqv (\ket{\psi} - \ket{\mu})}.
    \end{align}
    Therefore, it suffices to prove that \begin{align}
        \norm{\pinqv \ket{\mu}} &= O\left(\sqrt{\frac{t^2}{2^n}}\right), & \norm{\piqv (\ket{\psi} - \ket{\mu})} &= O\left(\sqrt{\frac{t^2}{2^n}}\right).
    \end{align}
    We perform the analysis for the case where the flip bit is $b=0$, the other case follows similarly.
    Let us begin with the right hand equation. By definition, we have that $\ket{\psi} - \ket{\mu} \in \pinfele$ is non-flip-elegant. Then we have \begin{align}
        \norm{\piqv (\ket{\psi} - \ket{\mu})} &= \norm{\mfc_{\reg{XP}} \inD_{\reg{XP}}\mfc_{\reg{XP}}^\dagger (\ket{\psi}_{\reg{XP}} - \ket{\mu}_{\reg{XP}})} \\
        &=\norm{\inD_{\reg{XP}} \widetilde \mfc_{\reg{XP}} ^\dagger(\ket{\psi}_{\reg{XP}} - \ket{\mu}_{\reg{XP}})} + O\left(\sqrt{\frac{t^2}{2^n}}\right) & \text{(By \Cref{cor:close-mfc})}\nonumber \\
        &= O\left(\sqrt{\frac{t^2}{2^n}}\right). & \text{(By \Cref{lem:cmf-nele-act})} \nonumber
    \end{align}
    In the other case, where $\ket{\mu} \in \pifele$, we have \begin{align}
        \norm{\pinqv \ket{\mu}} &= \norm{\mfc_{\reg{XP}} \ninD_{\reg{XP}}\mfc_{\reg{XP}}^\dagger \ket{\mu}_{\reg{XP}}} \\
        &= \norm{ \ninD_{\reg{XP}} \overline \mfc_{\reg{XP}}^\dagger \ket{\mu}_{\reg{XP}}} + O\left(\sqrt{\frac{t^2}{2^n}}\right) & \text{(By \Cref{cor:close-mfc,cor:ideal-cmf-act})} \nonumber \\
        &= O\left(\sqrt{\frac{t^2}{2^n}}\right). & \text{(By \Cref{rem:ideal-action})} \nonumber
    \end{align}
\end{proof}

Another important property which we can now prove is that the $\cmf$ operator approximately preserves the elegant subspace.

\begin{lemma}
    Let $\ket{\psi}_{\reg{XYP}} \in \algo H_{\reg{XY}} \otimes \hfele$ be a flip elegant quantum state supported on purifications of size at most $t$. Then we have \begin{align}
        \norm{\pinfele_{\reg P} \cdot \cmf_{\reg{XYP}} \ket{\psi}_{\reg{XYP}}} = O\left(\sqrt{\frac{t^2}{2^n}}\right)
    \end{align}
    \label{lem:ele-preserved}
\end{lemma}

\begin{proof}
    Let us analyze the case for forward queries, the proof for inverse queries follows similarly.
    \begin{align}
        \norm{\pinele_{\reg P} \cdot \cmf_{\reg{XYP}} \ket{\psi}_{\reg{XYP}}} &= \norm{\pinele_{\reg P} \cdot \mfc_{\reg{XP}} \cdot \pu_{\reg{XYP}} \cdot \mfc_{\reg{XP}}^\dagger \ket{\psi}_{\reg{XYP}}} \\
        &= \norm{\pinele_{\reg P} \cdot \overline\mfc_{\reg{XP}} \cdot \pu_{\reg{XYP}} \cdot \overline\mfc_{\reg{XP}}^\dagger \ket{\psi}_{\reg{XYP}}} + O\left(\sqrt{\frac{t^2}{2^n}}\right),  \nonumber 
    \end{align}
    by \Cref{cor:ideal-cmf-act}, and the fact that $\pu$ preserves the subspace $\inD\pisoph$. Applying \Cref{rem:ideal-action}, we find that the norm term is $0$, showing the claim.
\end{proof}

Using this lemma, we can show that the compressed masked Feistel oracle approximately preserves the sophisticated subspace.

\begin{lemma}
    Let $\algo A$ be a $q$ query quantum algorithm with intermediate operations $A_0, \dots, A_q$, and denote the final state by \begin{align}
        \ket{\phi_q}_{\reg{AP}} &= A_{q, \reg A} \cmf_{\reg{AP}} \dots A_{1, \reg A} \cmf_{\reg{AP}} A_{0, \reg A} \ket{0}_{\reg A} \ket{\boti}_{\reg P}.
    \end{align}
    Then we have \begin{align}
        \norm{\pinsoph_{\reg P} \ket{\phi_q}_{\reg{AP}}} &= O\left(\sqrt{\frac{q^4}{2^n}}\right).
    \end{align}
    \label{lem:soph-pres}
\end{lemma}

\begin{proof}
    We will prove this statement inductively. Clearly the $A_{t, \reg A}$ operations cannot break validity nor increase the projector onto the non-sophisticated subspace, as operators on different registers commute. Therefore, we need only analyze the action of each query.
    
    Suppose that the statement holds for the first $t$ queries, and denote the state after these queries as \begin{align}
        \ket{\phi_t}_{\reg{AP}} &= A_{t, \reg A} \cmf_{\reg{AP}} \dots A_{1, \reg A} \cmf_{\reg{AP}} A_{0, \reg A} \ket{0}_{\reg A} \ket{\boti}_{\reg P}.
    \end{align}
    We will split the $\reg A$ register into a further $\reg{XY}$ component for the query registers.
    We can write this state as \begin{align}
        \ket{\phi_t}_{\reg{AXYP}} &= \underbrace{\piqv_{\reg{XP}}\pisoph_{\reg P} \ket{\phi_t}_{\reg{AXYP}}}_{\ket{\mu}} + \underbrace{\piqv_{\reg{XP}}\pinsoph_{\reg P} \ket{\phi_t}_{\reg{AXYP}}}_{\ket{\delta}} + \underbrace{\pinqv_{\reg{XP}} \ket{\phi_t}_{\reg{AXYP}}}_{=0}. \label{eqn:phit-sophpres-comps}
    \end{align}
    We have $\norm{\ket{\delta}} \leq \norm{\pinsoph_{\reg P} \ket{\phi_t}_{\reg{AXYP}}}$. Turning to the first term, we have \begin{align}
        \cmf_{\reg{XYP}} \ket{\mu}_{\reg{XYP}} &= \cmf_{\reg{XYP}}\left(\pifele_{\reg{XP}} \ket{\phi_t}_{\reg{AXYP}} + O\left(\sqrt{\frac{t^2}{2^n}}\right)\right) & \text{(By \Cref{lem:ele-valid})} \\
        &= \cmf_{\reg{XYP}}\pifele_{\reg{XP}} \ket{\phi_t}_{\reg{AXYP}} + O\left(\sqrt{\frac{t^2}{2^n}}\right).
    \end{align}
    Recalling from \Cref{lem:ele-preserved} that $\cmf$ approximately preserves the flip elegant subspace, and that the flip elegant subspace is a subspace of sophisticed $\hele \leq \algo H_{\reg X} \otimes \hsoph$, we have $\norm{\pinele\cmf_{\reg{XYP}} \ket{\mu}_{\reg{XYP}}} = O(\sqrt{t^2/2^n})$. Plugging into \Cref{eqn:phit-sophpres-comps} and applying triangle inequality, we find that \begin{align}
        \norm{\pinsoph_{\reg P}\cmf_{\reg{XYP}} \ket{\phi_t}_{\reg{AXYP}}} &= \norm{\pinsoph_{\reg P} \ket{\phi_t}_{\reg{AXYP}}} + O\left(\sqrt{\frac{t^2}{2^n}}\right).
    \end{align}
    The claim now follows by induction.
\end{proof}

Finally, we can show that the $\iso$ operator acts as an approximate intertwiner between the two compressed permutation oracles, so long as the state is valid and close to sophisticated. 

\begin{lemma}
    Let $\ket{\psi}_{\reg{AXYP}}$ be a quantum state supported on purifications of size at most $t$, and which is valid. Let $\delta = \norm{\pinsoph_{\reg{P}} \ket{\psi}_{\reg{AXYP}}}$ be the size of the unsophisticated component.
    Then we have \begin{align}
        \norm{(\iso_{\reg P} \cdot \cmf_{\reg{XYP}} - \cp_{\reg{XYI}} \cdot \iso_{\reg P}) \ket{\psi}_{\reg{AXYP}}} &= O\left(\delta + \sqrt{\frac{t^2}{2^n}}\right).
    \end{align}
    \label{lem:int-valid}
\end{lemma}

\begin{proof}
    We can write $\ket{\psi}_{\reg{AXYP}}$ as \begin{align}
        \ket{\psi}_{\reg{AXYP}} &= \underbrace{\piqv_{\reg{XP}}\pisoph_{\reg{P}} \ket{\psi}_{\reg{AXYP}}}_{\ket{\mu}} + \underbrace{\piqv_{\reg{XP}}\pinsoph_{\reg{P}} \ket{\psi}_{\reg{AXYP}}}_{\ket{\delta}} + \underbrace{\pinqv_{\reg{XP}}\ket{\psi}_{\reg{AXYP}}}_{=0}. \label{eqn:psi-inter-parts}
    \end{align}
    The final term being $0$ follows from the fact that the state is valid.
    Observe that $\norm{\ket{\delta}} \leq \delta$, and the operator acting on the state is norm at most $2$, so this incurs an additive $O(\delta)$. Turning to the remaining component, we can write \begin{align}
        \iso_{\reg P} \cdot \cmf_{\reg{XYP}} \ket{\mu}_{\reg{AXYP}} &= \iso_{\reg P} \cdot \cmf_{\reg{XYP}} \cdot \pifele_{\reg{XP}} \ket{\psi}_{\reg{AXYP}} + O\left(\sqrt{\frac{t^2}{2^n}}\right) & \text{(By \Cref{lem:ele-valid})} \nonumber\\
        &= \cp_{\reg{XYI}} \cdot \iso_{\reg P} \cdot \pifele_{\reg{XP}} \ket{\psi}_{\reg{AXYP}} + O\left(\sqrt{\frac{t^2}{2^n}}\right)& \text{(By \Cref{cor:ideal-cmf-int})}
    \end{align}
    Plugging into \Cref{eqn:psi-inter-parts} and applying triangle inequality shows the claim.
\end{proof}

We are now ready to show the main result of this manuscript, that the compressed permutation oracle construction is sound. We remind the reader that up to this point, we assumed only that the twirling distributions were cromulent. In the next section, we will make an explicit choice for these distributions.

\subsection{Putting the pieces together}
\label{subsec:pieces-main}

Let $\algo A$ be a $q$ query quantum algorithm which queries a random permutation oracle. Suppose that the algorithm has a workspace register $\reg{A}$, and input/output registers $\reg{X}$ and $\reg{Y}$ of dimension $2^{2n+1}$ and $2^{2n}$. The extra input bit indicates the direction of the query. We describe such an algorithm by a sequence of unitary operators $A_0, \dots, A_q$  acting on the $\reg{AXY}$ system, interleaved with oracles $\algo O_{\varphi, \reg{XY}}$ which acts $\forall b \in \bit, x,y \in \bit^{2n}$ as \begin{align}
    \algo O_{\varphi, \reg{XY}} \ket{b, x}_\reg{X} \ket{y}_\reg{Y} &= 
        \ket{b, x}_\reg{X} \ket{y \oplus  \varphi^{1-2b}(x)}_\reg{Y}.
\end{align}
For simplicity, we will treat the registers $\reg{XY}$ as subregisters of the adversary state $\reg{A}$. We can then write the view of such an algorithm as the density matrix \begin{align}
    \rho^{(\algo O)}_\reg{A} &= \mathop{\mathbb{E}}_{\varphi \sim \mathbf S_{2^{2n}}} \left[\rho(A_{q, \reg{A}} \dots \algo O_{\varphi, \reg{A}} A_{0, \reg{A}} \ket{0}_{\reg A})\right].
\end{align}
Let us compare to the compressed oracle experiment, where we answer queries according to our compressed permutation oracle. In this experiment, we have a purifying register $\reg I$ in the Hilbert space $\algo H(\mathbf I)$, spanned by all injective partial functions from $\bit^{2n} \ra \bit^{2n}$. In this case, we answer queries using the unitary $\cp$, as described in \Cref{sec:comp-perms}. We think of the $\reg I$ register as being outside the algorithm's view, so it is traced out. We can write the final state in this experiment as \begin{align}
    \rho^{(\cp)}_\reg{A} &= \Tr_\reg{I} \left[\rho(A_{q, \reg{A}} \dots \cp_\reg{AI} A_{0, \reg{A}} \ket{0}_{\reg{A}} \ket{\boti}_{\reg{I}})\right].
\end{align}
Our main result is that these two views are indistinguishable unless the algorithm has made a number of queries comparable to the twelveth-root of the permutation size.

\begin{theorem}
    Any $q$ query quantum algorithm cannot distinguish the compressed permutation oracle from a uniform random permutation oracle on $[N] \simeq \bit^{2n}$ except with advantage \begin{align}
        \frac{1}{2}\norm{\rho^{(\algo O)}_\reg{A} - \rho^{(\cp)}_\reg{A}}_1 &= O\left(\frac{q^{3}}{N^{1/4}}\right).
    \end{align}
    In particular, constant advantage requires $\Omega(\sqrt[12]{N})$ queries.
    \label{thm:main-indist}
\end{theorem}

\begin{proof}
    We will prove this statement by purifying. The natural purification of a random permutation oracle, which purifies $\rho^{(\algo O)}_\reg{A}$, is to maintain a uniform superposition of permutation truth tables outside the view of the adversary. We will consider the alternate, more complicated purification described in \Cref{sec:Feistel-twirl}, with a uniform twirl. This can be written as follows, where $\reg P$ is the purifying register containing sub-registers $\reg {\Pi\Omega HKF}$, \begin{align}
        \ket{\psi^{(\algo O)}}_\reg{AXYP} &= A_{q,\reg{A}} \dots \mfpu_\reg{AP} A_{0,\reg{A}} \ket{0}_\reg{A} \otimes \sqrt{\frac{1}{(N!)^2 (2^n)^{3 \cdot 2^n}}} \sum_{\pi, \omega \in \mathbf S_N, h, k, f \in \mathbf F_{2^n, 2^n}} \ket{\pi, \omega, h, k, f}_\reg{P}.
    \end{align}
    In other words, we take the twirling distribution $\algo D$ to be uniform over $\mathbf S_N^{\times 2}$. This distribution is cromulent, as shown in \Cref{sec:twirling}. This oracle is perfectly indistinguishable from a random permutation oracle, $\Tr_{\reg P} [\rho(\ket{\psi^{(\algo O)}}_\reg{AXYP})]=\rho^{(\algo O)}_{\reg A}$, as the permutation is a product of a random permutation (say $\pi$) with an independent permutation.
    
    Now, by the standard compressed oracle theory, the state \begin{align}
        \ket{\psi^{(\cmf)}}_\reg{AP} &= A_{q,\reg{A}} \dots \cmf_{\reg{AP}} A_{0,\reg{A}} \ket{0}_\reg{A} \otimes \frac{1}{N!} \sum_{\pi, \omega \in \mathbf S_{N}} \ket{\pi, \omega, \botd, \botd, \botd}_\reg{P}
    \end{align}
    is perfectly indistinguishable after tracing out the $\reg P$ register, in that we have \begin{align}
        \rho^{(\algo O)}_{\reg A} &= \Tr_\reg{P}[\rho(\ket{\psi^{(\algo O)}}_{\reg{AP}})] = \Tr_\reg{P}[\rho(\ket{\psi^{(\cmf)}}_{\reg{AP}})].
    \end{align}
    The purification of the compressed permutation oracle is \begin{align}
        \ket{\psi^{(\cp)}}_\reg{AI} &= A_{q,\reg{A}} \dots \cp_{\reg{AI}} A_{0,\reg{A}} \ket{0}_\reg{A} \ket{\boti}_\reg{I}, & \Tr_{\reg{I}}[\rho(\ket{\psi^{(\cp)}}_\reg{AI})] = \rho_{\reg A}^{(\cp)}.
    \end{align}
    Now we can consider the hybrid states, \begin{align}
        \ket{\psi^{(\mathsf{Hyb})}_0} &= A_{q,\reg{A}} \dots \cp_{\reg{AI}} A_{0,\reg{A}} \iso_\reg{P} \ket{0}_\reg{A} \ket{{\mathbf P(\boti)}}_\reg{P} \\
        \ket{\psi^{(\mathsf{Hyb})}_1} &= A_{q,\reg{A}} \dots \cp_{\reg{AI}} A_{1,\reg{A}}\iso_\reg{P} \cmf_{\reg{AP}} A_{0,\reg{A}} \ket{0}_\reg{A} \ket{{\mathbf P(\boti)}}_\reg{P} \\
        &\dots \nonumber\\
        \ket{\psi^{(\mathsf{Hyb})}_q} &= \iso_\reg{P} A_{q,\reg{A}} \dots \cmf_{\reg{AP}} A_{0,\reg{A}}  \ket{0}_\reg{A} \ket{{\mathbf P(\boti)}}_\reg{P},
    \end{align}
    where $\ket{\psi^{(\cp)}}_\reg{AI} = \ket{\psi^{(\mathsf{Hyb})}_0}$ and $\iso_{\reg P}\ket{\psi^{(\cmf)}}_\reg{AP} =  \ket{\psi^{(\mathsf{Hyb})}_q}$. For $t \leq q$, let us define \begin{align}
        \ket{\phi_t}_{\reg{AP}} = A_{t,\reg{A}} \dots \cmf_{\reg{AP}} A_{0,\reg{A}}  \ket{0}_\reg{A} \ket{{\mathbf P(\boti)}}_\reg{P},
    \end{align}
    observing that $\norm{\pinsoph_{\reg P} \ket{\phi_t}_{\reg{AP}}} = O(\sqrt{t^4/2^n})$ by \Cref{lem:ele-preserved}, and further the $\ket{\phi_t}_{\reg{AP}} \in \piv_{\reg P}$, lying entirely in the valid subspace as $\cmf$ preserves validity. Observing that $\iso_{\reg P}$ clearly commutes with the $A_{t, \reg A}$ operator, we can then compute the successive differences between hybrid states to be \begin{align}
        \norm{\ket{\psi^{(\mathsf{Hyb})}_{t+1}} - \ket{\psi^{(\mathsf{Hyb})}_{t}}} &= \norm{(\iso_{\reg P} \cdot \cmf_{\reg{AP}}-\cp_{\reg{AI}} \cdot \iso_{\reg P})\ket{\phi_t}_{\reg{AP}}} \nonumber\\
        &= O\left(\sqrt{\frac{t^4}{2^n}}\right). & \text{(By \Cref{cor:ideal-cmf-int})}
    \end{align}
    Therefore, we have $\norm{\iso_{\reg P}\ket{\psi^{(\cmf)}}_\reg{AP} - \ket{\psi^{(\cp)}}_\reg{AI}} = O\left(\sqrt{\frac{q^6}{2^n}}\right)$. This equation combined with \Cref{lem:approx-uhlman} proves the claim.
\end{proof}

\section{Feistel security}
\label{sec:feist-sec}

Let $\feist^{(m)}$ denote the distribution induced on permutations by $m$ rounds of Feistel with truly random round functions. This is query-indistinguishable from the case where the round functions are pseudorandom.
The size of all permutations considered is $N=4^n$. Our convention is that the first round is ``left-to-right'', in that input $x$ is transformed to $x_L \Vert (x_R \oplus h(x_L))$ in the first round.

Our proof that the seven round Feistel construction is a qPRP is nearly identical to the soundness proof of our compressed permutation oracle. In particular, we will show that Feistel is indistinguishable from the compressed permutation oracle. To do so, we will apply the machinery we have developed for masked Feistel to the case where the twirling distribution $\algo D$ draws $\pi$ uniformly from $\feist^{(2)}$, and $\omega$ uniformly from $(\feist^{(2)})^{-1}$. In other words, we define \begin{align}
    p_{\algo D}(\pi, \omega) &= p_{\feist^{(2)}\times \feist^{(2)}} (\pi, \omega^{-1}).
\end{align}

Note that this distribution is invariant under the transformation $(\pi, \omega) \mapsto (\omega^{-1}, \pi^{-1})$, maintaining the flip symmetry of all our constructions thus far. The main technical component of this section is proving that this distribution is \emph{cromulent}, which amounts to proving purely classical statistical properties. Once this is in place, we can use the results derived in the previous two sections to give a simple proof of security.

\subsection{Twirling}
\label{sec:outer-feist-twirl}

We will study the statistical properties of the two round Feistel construction. In particular, we will show that $\algo D = \feist^{(2)} \times (\feist^{(2)})^{-1}$ is a cromulent distribution. From the flip symmetry above, it will suffice to show the properties in \Cref{def:cromulent} for the listed forward cases only. We can think of sampling from $\algo D_I$ as a procedure in which we assign inputs to the constituent functions in stages. We use $D^{(\pi)}_1$ and $D^{(\pi)}_2$ to denote the partial databases representing the two Feistel functions composing $\pi$ throughout this procedure. We use $D^{(\omega)}_1$ to represent the last Feistel round, and $D^{(\omega)}_2$ to represent the second to last, in $\omega$. Observe that $D^{(\cdot)}_1$ is left-to-right, and $D^{(\cdot)}_2$ is right-to-left, for both $\pi$ and $\omega$. We stress that this is purely a classical procedure, which is useful as a tool to prove probabilistic properties about this distribution. We use ideas from \Cref{def:non-allowable-dbs} in this description.

\begin{enumerate}
    \item For each $x \in \dom(I)$, sample random values for $D^{(\pi)}_1(x_L)$ and $D^{(\pi)}_2(x_R \oplus D^{(\pi)}_1(x_L))$, and for each $y \in \im(I)$, sample random values for $D^{(\omega)}_1(y_L)$ and $D^{(\omega)}_2(y_R \oplus D^{(\omega)}_1(y_L))$.
    \item Observing that $A=\pi * I * \omega$ is now fully defined, if $A$ is non-allowable, restart the procedure, and otherwise continue.
    \item Sample the remaining undefined points at uniform random.
\end{enumerate}

Now we observe a number of symmetries in this distribution. First, let us define a \emph{right shift} in $\pi$, parameterized by an $s \in \bit^n$, as the transformation $(D^{(\pi)}_1, D^{(\pi)}_2) \mapsto_s (D^{(\pi)'}_1, D^{(\pi)'}_2)$ where \begin{align}
    D^{(\pi)}_1(w_1)=z_1 &\Leftrightarrow D^{(\pi)'}_1(w_1)=z_1 \oplus s, \\
    D^{(\pi)}_2(z_2)=w_2 &\Leftrightarrow D^{(\pi)'}_2(z_2 \oplus s)=w_2.
\end{align}
Similarly, a right shift in $\omega$ parameterized by an $s \in \bit^n$ is the transformation $(D^{(\omega)}_1, D^{(\omega)}_2) \mapsto_s (D^{(\omega)'}_1, D^{(\omega)'}_2)$ where \begin{align}
    D^{(\omega)}_1(w_1)=z_1 &\Leftrightarrow D^{(\omega)'}_1(w_1)=z_1 \oplus s, \\
    D^{(\omega)}_2(z_2)=w_2 &\Leftrightarrow D^{(\omega)'}_2(z_2 \oplus s)=w_2.
\end{align}
We define a left shift as a joint transformation on $(\pi, \omega)$, again parameterized by an $s \in \bit^n$, which maps $(D^{(\pi)}_2, D^{(\omega)}_2) \mapsto_s (D^{(\pi)'}_2, D^{(\omega)'}_2)$ where \begin{align}
    D^{(\pi)}_2(z_{\pi})=w_{\pi} &\Leftrightarrow D^{(\pi)'}_2(z_{\pi})=w_{\pi} \oplus s, \\
    D^{(\omega)}_2(z_\omega)=w_\omega &\Leftrightarrow D^{(\omega)'}_2(z_\omega)=w_\omega \oplus s.
\end{align}

We can equivalently think of sampling from $\algo D_I$ with additional re-randomization, as illustrated by these symmetries. In the following lemma, inside of each step the order does not matter (e.g. the selections in step (1) can be made in any order, as long as they all occur before the selections in step (4)).

\begin{lemma}
    The following procedure is equivalent to sampling $(\pi, \omega)$ from $\algo D_I$, where $\pi$ is defined by the two round Feistel $D^{(\pi)}_1, D^{(\pi)}_2$ and $\omega$ is defined by the inverse two round Feistel of $D^{(\omega)}_1, D^{(\omega)}_2$.
    \begin{enumerate}
        \item For each $x \in \dom(I)$, sample random values for $D^{(\pi)}_1(x_L)$ and $D^{(\pi)}_2(x_R \oplus D^{(\pi)}_1(x_L))$, and for each $y \in \im(I)$, sample random values for $D^{(\omega)}_1(y_L)$ and $D^{(\omega)}_2(y_R \oplus D^{(\omega)}_1(y_L))$. Noting that this suffices to define $A=\pi * I * \omega$, do this conditioned on the fact that $A$ is allowable.
        \item Sample random $s \sim \bit^n$, perform a right shift by $s$ on $\pi$.
        \item Sample random $s' \sim \bit^n$, perform a right shift by $s'$ on $\omega$.
        \item Sample the remaining undefined points of the $D_i^{(\pi)}$ and $D_i^{(\omega)}$ at uniform random.
        \item Sample random $s'' \sim \bit^n$, perform a left shift by $s''$ on $(\pi, \omega)$.
    \end{enumerate}
    \label{lem:shifted-feist-dist}
\end{lemma}

\begin{proof}
    Recall that allowability of $A=\pi * I * \omega$ is a condition on only the left substrings of $A$. A random right shift does not affect the left substrings of $A$, and so the right shifts in steps (2) and (3) do not change the distribution. Furthermore, the allowability conditions are a set of linear equations over $\mathbb Z_2^n$ on the left substrings with $2$ or $4$ variables in each. Shifting all of the variables of such an equation does not change its truth/falsehood, as every shift $s$ has order $2$ with $s \oplus s = 0$, so the last step does not affect allowability.
\end{proof}

Note that we can reason similarly about, e.g., $\algo D_{x, I}$, simply adding the condition of allowing $x$ in step (1). The same argument implies that the random shifts do not change the truth or falsehood of this condition. 
With this in place, we are ready to prove that $\algo D$ satisfies the cromulent conditions.

\begin{lemma}
    Fix some $I \in \mathbf I_t$, $x \in [N] \setminus \dom(I)$. Then we have \begin{align}
        \Pr_{(\pi, \omega) \sim \algo D_I} [(\pi, \omega) \in \mathbf R_{x, I}] &= 1-O\left(\frac{t^2}{2^n}\right)
    \end{align}
    when $t=O(2^{n/3})$.
    \label{lem:twirl-feist-prob}
\end{lemma}

Note the additional condition on $t$ compared to the uniform twirl case; this nonetheless holds in our application.

\begin{proof}
    We will use the characterization in \Cref{lem:shifted-feist-dist}. There are two cases to consider. \begin{enumerate}
        \item There is an $x' \in I$ with $x_L=x_L'$. Let $X \subset I$ be the set of such $x'$'s. Let us imagine selecting the value $D^{(\pi)}_1(x_L)$, and the corresponding values in $D^{(\pi)}_2$, last in step (1). Prior to this selection, define $Z=\{x''_R \oplus D^{(\pi)}_1(x''_L) \, : \, x'' \in I, x''_L \neq x_L\}$ as the set of values on the right wire which have been determined, where $|Z| <t$. We use $Z' = \{z \, : \, \exists x' \in X \cup \{x\} \text{ s.t. } x'_R \oplus z \in Z\}$ as the set of assignments to $D^{(\pi)}_1(x_L)$ which would cause a collision with these values on the right wire. The probability of selecting some $D^{(\pi)}_1(x_L)$ which is not in $Z'$ is uniform, so it only remains to analyze the probability of some selection in $Z'$. Such a selection may be more likely if it causes some $x'$ to immediately satisfy allowability. In such a case, there are a factor $2^n/(2^n-O(t^2)) = 1+O(t^2/2^n)$ more ways to choose the database values, as there would have been $2^n-O(t^2)$ allowable values of $D^{(\pi)}_2(x'_R \oplus D^{(\pi)}_1(x'_L))$ otherwise by \Cref{lem:allow-prefixes}. Over the at most $t$ values in $X$, this gives a factor at most $(1+O(t^2/2^n))^t = 1+O(t^3/2^n)=O(1)$, meaning values in $Z'$ are at most a constant factor more likely than values outside. It follows that values in $Z'$ have probability $O(2^{-n})$. We further have $|Z'|=O(t^2)$ by a simple counting argument.
        
        Over the random choice (conditioned on prior events) of $D^{(\pi)}_1(x_L)$, we then have \begin{align}
            \Pr_{(\pi, \omega) \sim \algo D_I} [D^{(\pi)}_1(x_L) \in Z'] &= O(t^2/2^n)
        \end{align}
        from the union bound. Supposing that this event does not happen, the value $D^{(\pi)}_2(x_R \oplus D^{(\pi)}_1(x_L))$ will be selected at uniform random in step (4). Once again by \Cref{lem:allow-prefixes}, this will lead to an allowable value of $\pi(x)_L$ with probability $1-O(t^2/2^n)$, showing the claim.
        \item There is no $x' \in I$ with $x_L=x_L'$. Then, the value $D^{(\pi)}_1(x_L)$ will be selected in step (4), and will result in a value $D^{(\pi)}_2(x_R \oplus D^{(\pi)}_1(x_L))$ which is not already selected with probability $1-O(t/2^n)$. Similar to the previous case, conditioned on this value being sampled fresh in step (4) there is once again a probability $1-O(t^2/2^n)$ of $x$ being disallowed, proving the claim.
    \end{enumerate}
\end{proof}

\begin{lemma} 
    Fix some $I \in \mathbf I_t$, $y\neq y' \in [N] \setminus \im(I)$, $x \in [N] \setminus \dom(I)$ and $l \in \bit^n$. Then we have \begin{align}
        \Pr_{(\pi, \omega) \sim \algo D_{x, I}} [\omega^{-1}(y)_L = l] &= \frac{1}{2^n} \\
        \Pr_{(\pi, \omega) \sim \algo D_{x, I}} [\omega^{-1}(y)_L = \omega^{-1}(y')_L = l] &= \frac{1}{2^{2n}}\pm O\left(\frac{t^2}{2^{3n}}\right).
    \end{align}
    \label{lem:twirl-feist-second-mom}
\end{lemma}

\begin{proof}
    We will use the characterization in \Cref{lem:shifted-feist-dist}. For the first equation, consider performing steps (1-4). Then, step (5) will re-randomize the value of $\omega^{-1}(y)$, such that it will satisfy $\omega^{-1}(y)_L = l$ with probability $2^{-n}$. 

    For the second equation, let us begin by analyzing the probability that $\omega^{-1}(y)_L=\omega^{-1}(y')_L$ after steps (1-4). There are the following cases. \begin{enumerate}
        \item $y_L = y_L'$. Then, after the selection of $D_1^{(\omega)}(y_L)$, it is guaranteed that $y_R \oplus D_1^{(\omega)}(y_L) \neq y_R' \oplus D_1^{(\omega)}(y_L')$. Let us define \begin{align}
            z&=y_R \oplus D_1^{(\omega)}(y_L), & 
            z'&=y_R \oplus D_1^{(\omega)}(y_L).
        \end{align} 
        Now, consider performing all of the selections which can be performed prior to selecting $D_2^{(\omega)}(z)$. WLOG we can assume that $D_2^{(\omega)}(z')$ is already selected (otherwise, swap $y$ and $y'$). The selection of $D_2^{(\omega)}(z)$ will occur either at the end of step (1) (in which case there are at least $2^n-O(t^2)$ admissible values) or at the end of step (4) (in which case there are $2^n$ admissible values), depending on the prior selections. In all cases, we have shown \begin{align}
            \frac{1}{2^n} \leq \Pr_{(\pi, \omega) \sim \algo D_{x, I}} [\omega^{-1}(y)_L = \omega^{-1}(y')_L \, |\, y_L = y_L'] \leq \frac{1}{2^n-O(t^2)}.
        \end{align}
        \item $y_L \neq y_L'$, and there is  $y'', y''' \in \im(I)$ with $y''_L = y_L$ and $y'''_L=y'_L$. We will call this case $E_2$. Consider in step (1) first selecting all of the values relevant to $\pi$, all of the relevant values for $D_1^{(\omega)}$ except $D_1^{(\omega)}(y_L)$, and none of the values for $D_2^{(\omega)}$. Let us define $z'=y_R' \oplus D_1^{(\omega)}(y_L')$. In the selection of $D_1^{(\omega)}(y_L)$, there are at least $2^n-O(t^2)$ admissible values, i.e. values not immediately violating the allowability condition. If we define $z=y_R \oplus D_1^{(\omega)}(y_L)$, then we will have $z=z'$ with probability no more than $\frac{1}{2^n-O(t^2)}$; call this event $E_{\mathsf{col}}$.

        In the case of $E_{\mathsf{col}}$, we are guaranteed that $\omega^{-1}(y)_L \neq \omega^{-1}(y')_L$, as if there is a right collision there cannot be a left collision. In the other case, we can consider once again selecting $D_2^{(\omega)}(z)$ after $D_2^{(\omega)}(z')$, and after all the (other) selections in step (1) are made. Similar to the previous case, depending on the step in which this occurs, the probability of $\omega^{-1}(y)_L=\omega^{-1}(y')_L$ lies between $\frac{1}{2^n}$ and $\frac{1}{2^n-O(t^2)}$. We then have \begin{align}
            \Pr_{(\pi, \omega) \sim \algo D_{x, I}} [\omega^{-1}(y)_L = \omega^{-1}(y')_L \, |\, E_2] &= \Pr_{(\pi, \omega) \sim \algo D_{x, I}} [\omega^{-1}(y)_L = \omega^{-1}(y')_L \, |\, E_2, E_{\mathsf{col}}] + \nonumber\\&\quad \underbrace{\Pr_{(\pi, \omega) \sim \algo D_{x, I}} [\omega^{-1}(y)_L = \omega^{-1}(y')_L \, |\, E_2, \neg E_{\mathsf{col}}]}_{=0} \\
            &= \left(1-\frac{1}{2^n-O(t^2)}\right) \cdot \left(\frac{1}{2^n \pm O(t^2)}\right) \\
            &= \frac{1}{2^n} \pm O\left(\frac{t^2}{2^{2n}}\right)
        \end{align}
        
        \item There is no $y'' \in \im(I)$ with $y''_L = y_L$, or no such $y''$ with $y''_L=y'_L$. We can WLOG assume the first one. We call this case $E_3$. Now, we imagine performing all of the selections in step (1), as well as selecting $D_1^{(\omega)}(y_L')$ to define $z'=y_R' \oplus D_1^{(\omega)}(y_L')$, and selecting $D_2^{(\omega)}(z')$ to define $w'=y_L' \oplus D_2^{(\omega)}(z')$ (this selection may occur in step (1), or as the first selection in step (4)). Now let $d=y_L \oplus w'$. Let $E_d$ denote the event that $d$ appears in the image of $D_2^{(\omega)}$ at this stage. By the shift invariance of allowability, and the fact that $D_2^{(\omega)}$ is of size at most $t+1$, we have that $E_d$ occurs with probability $O(t/2^n)$.

        Now we consider selecting $D_1^{(\omega)}(y_L)$ and defining $z=y_R \oplus D_1^{(\omega)}(y_L)$. Let $E_z$ denote the event that $z$ is in the domain of $D_2^{(\omega)}$. There are at most $t+1$ such values, so $E_z$ occurs with probability $O(t/2^n)$ over the randomness of the selection. Note that $E_z$ is independent of $E_d$. If $E_z$ occurs, then the value $w=y_L \oplus D_2^{(\omega)}(z)$ is already defined, and we have that $\omega^{-1}(y)_L=\omega^{-1}(y')_L$ can only occur if $D_2^{(\omega)}(z)=d$, i.e. if $E_d$ occurred.

        Turning to the other case, we consider randomly selecting $D_2^{(\omega)}(z)$, which defines $w=y_L \oplus D_2^{(\omega)}(z)$. With probability $2^{-n}$, we will have $w=w'$. This occurs regardless of $E_d$. Assembling all the cases, we can write \begin{align}
            \Pr_{(\pi, \omega) \sim \algo D_{x, I}} [\omega^{-1}(y)_L = \omega^{-1}(y')_L \, |\, E_3] &= \underbrace{\Pr_{(\pi, \omega) \sim \algo D_{x, I}} [\omega^{-1}(y)_L = \omega^{-1}(y')_L \, |\, E_3, E_d, E_z]}_{O(t^2/2^{2n})} + \nonumber\\&\quad\underbrace{\Pr_{(\pi, \omega) \sim \algo D_{x, I}} [\omega^{-1}(y)_L = \omega^{-1}(y')_L \, |\, E_3, \neg E_d, E_z]}_{=0}+ \nonumber\\&\quad
            \underbrace{\Pr_{(\pi, \omega) \sim \algo D_{x, I}} [\omega^{-1}(y)_L = \omega^{-1}(y')_L \, |\, E_3, \neg E_z]}_{(1-O(t/2^n))(2^{-n})} \\
            &= \frac{1}{2^n} \pm O\left(\frac{t^2}{2^{2n}}\right).
        \end{align}
    \end{enumerate}
    Together, these cases imply that we have $\omega^{-1}(y)_L=\omega^{-1}(y')_L$ with probability $2^{-n} \pm O(t^2/2^{2n})$, prior to step (5). Over the random shift in step (5), we will have $\omega^{-1}(y)=l$ with probability $2^{-n}$, independent of the prior event. Multiplying together gives the claim.
\end{proof}

\begin{lemma}
    For any $I \in \mathbf I_t$, $x \in [N] \setminus \dom(I)$, $y \in [N] \setminus \im(I)$, and $\pi, \omega \in \mathbf R_{I[x\ra y]}$ with $A=\pi * I * \omega$, we have \begin{align}
        \sqrt{\frac{p_{\algo D_{I[x\ra y]}}(\pi, \omega)}{(N-t)}} - \sqrt{\frac{p_{\algo D_{x, I}}(\pi, \omega)}{|\mathbf V_{\pi(x), A}|}} &= \pm O\left(\sqrt{\frac{t^2}{2^n}} \cdot \sqrt{\frac{p_{\algo D_{I[x\ra y]}}(\pi, \omega)}{N-t}}\right).
    \end{align}
    \label{lem:twirl-feist-uniform}
\end{lemma}

\begin{proof}
    Observe that we have \begin{align}
        \frac{p_{\algo D_{I[x\ra y]}}(\pi, \omega)}{p_{\algo D_{x, I}}(\pi, \omega)} &= 1 \pm O\left(\frac{t^2}{2^n}\right), & \frac{N-t}{\mathbf V_{\pi(x), \mathbf A}} &= 1 \pm O\left(\frac{t^2}{2^n}\right),
    \end{align}
    the first by an argument analogous to \Cref{lem:twirl-feist-prob} and the second by \Cref{lem:allow-prefixes}. If we let $a = (N-t) / p_{\algo D_{I[x\ra y]}}(\pi, \omega)$, then we can write \begin{align}
        \sqrt{\frac{p_{\algo D_{x, I}}(\pi, \omega)}{|\mathbf V_{\pi(x), A}|}} &= \frac{1}{\sqrt{a(1 \pm O(t^2/2^n)}} \\
        &= \frac{1}{\sqrt{a}} \pm O\left(\frac{t^2}{2^n \cdot \sqrt{a}}\right),
    \end{align}
    where the last line follows from the taylor expansion of $\frac{1}{\sqrt{a+x}}$ about $x=0$.
\end{proof}

\subsection{Putting the pieces together}

As in \Cref{sec:soundness}, let $\algo A$ be a $q$ query quantum algorithm which queries a permutation oracle $\algo O_{\varphi}$ for $\varphi \sim \feist^{(7)}$. Suppose that the algorithm has a workspace register $\reg{A}$, and input/output registers $\reg{X}$ and $\reg{Y}$ of dimension $2^{2n+1}$ and $2^{2n}$. The extra input bit indicates the direction of the query. We describe such an algorithm by a sequence of unitary operators $A_0, \dots, A_q$ acting on the $\reg{AXY}$ system, interleaved with oracles $\algo O_{\varphi, \reg{XY}}$ which acts $\forall b \in \bit, x,y \in \bit^{2n}$ as \begin{align}
    \algo O_{\varphi, \reg{XY}} \ket{b, x}_\reg{X} \ket{y}_\reg{Y} &= 
        \ket{b, x}_\reg{X} \ket{y \oplus  \varphi^{1-2b}(x)}_\reg{Y}.
\end{align}
For simplicity, we will treat the registers $\reg{XY}$ as subregisters of the adversary state $\reg{A}$. We can then write the view of such an algorithm on querying feistel as the density matrix \begin{align}
    \rho^{(\feist)}_\reg{A} &= \mathop{\mathbb{E}}_{\varphi \sim \feist^{(7)}} \left[\rho(A_{q, \reg{A}} \dots \algo O_{\varphi, \reg{A}} A_{0, \reg{A}} \ket{0}_{\reg A})\right].
\end{align}
We can write the view of the same algorithm querying a truly random permutation as \begin{align}
    \rho^{(\algo O)}_\reg{A} &= \mathop{\mathbb{E}}_{\varphi \sim \mathbf S_N} \left[\rho(A_{q, \reg{A}} \dots \algo O_{\varphi, \reg{A}} A_{0, \reg{A}} \ket{0}_{\reg A})\right].
\end{align}

We will show that these two views are indistinguishable unless the algorithm has made a number of queries comparable to the twelveth-root of the permutation size.

\begin{theorem}
    Any $q$ query quantum algorithm cannot the seven round balanced Feistel construction from a uniform random permutation oracle on $[N]$ except with advantage \begin{align}
        \norm{\rho^{(\algo O)}_\reg{A} - \rho^{(\feist)}_\reg{A}}_1 &= O\left(\frac{q^{3}}{N^{1/4}}\right).
    \end{align}
    In particular, constant advantage requires $\Omega(\sqrt[12]{N})$ bidirectional queries.
    \label{thm:feist-indist}
\end{theorem}

\begin{proof}
    We will prove this statement by purifying. We will use our compressed oracle as an intermediary between the two, where we have a purifying register $\reg I$. We can write the final state in this experiment as \begin{align}
        \rho^{(\cp)}_\reg{A} &= \Tr_\reg{I} \left[\rho(A_{q, \reg{A}} \dots \cp_\reg{AI} A_{0, \reg{A}} \ket{0}_{\reg{A}} \ket{\boti}_{\reg{I}})\right].
    \end{align}
    By \Cref{thm:main-indist}, we have \begin{align}
        \norm{\rho^{(\cp)}_\reg{A} - \rho^{(\algo O)}_{\reg A}}_1 &= O\left(\frac{q^3}{N^{1/4}}\right).
    \end{align}
    We will therefore analyze the compresed permutation oracle experiment in place of the ideal permutation oracle.
    
    To purify the seven round Feistel construction, we will follow the idea described in \Cref{sec:Feistel-twirl}. In particular, we will consider the two leftmost Feistel rounds to define the permutation $\pi$, the two rightmost Feistel rounds to define the permutation $\omega$, and the internal three rounds to correspond to the $h, k, f$ functions. This can be written as follows, where $\reg P$ is the purifying register containing sub-registers $\reg {\Pi\Omega HKF}$, \begin{align}
        \ket{\psi^{(\algo O)}}_\reg{AXYP} &= A_{q,\reg{A}} \dots \mfpu_\reg{AP} A_{0,\reg{A}} \ket{0}_\reg{A} \otimes \sqrt{\frac{1}{(2^n)^{3 \cdot 2^n}}} \sum_{\pi, \omega \in \mathbf S_N, h, k, f \in \mathbf F_{2^n, 2^n}} \sqrt{p_{\algo D}(\pi, \omega)}\ket{\pi, \omega, h, k, f}_\reg{P}.
    \end{align}
    where $\algo D$ is defined as the distribution $\feist^{(2)} \times (\feist^{(2)})^{-1}$. This distribution is cromulent, as shown in \Cref{sec:outer-feist-twirl}, meaning we can apply the results of \Cref{sec:soundness}. This oracle is perfectly indistinguishable from the seven round Feistel construction, $\Tr_{\reg P} [\rho(\ket{\psi^{(\algo O)}}_\reg{AXYP})]=\rho^{(\feist)}_{\reg A}$.
    Now, by the standard compressed oracle theory, the state \begin{align}
        \ket{\psi^{(\mfc)}}_\reg{AP} &= A_{q,\reg{A}} \dots \cmf_{\reg{AP}} A_{0,\reg{A}} \ket{0}_\reg{A} \otimes  \sum_{\pi, \omega \in \mathbf S_{N}}\sqrt{p_{\algo D}(\pi, \omega)} \ket{\pi, \omega, \botd, \botd, \botd}_\reg{P}
    \end{align}
    is perfectly indistinguishable after tracing out the $\reg P$ register, in that we have \begin{align}
        \rho^{(\algo O)} &= \Tr_\reg{P}[\outerprod{\psi^{(\algo O)}}{\psi^{(\algo O)}}] \\
        &= \Tr_\reg{P}[\outerprod{\psi^{(\cmf)}}{\psi^{(\cmf)}}].
    \end{align}
    For comparison, we will take the following purification of the compressed oracle, \begin{align}
        \ket{\psi^{(\cp)}}_\reg{AI} &= A_{q,\reg{A}} \dots \cp_{\reg{AI}} A_{0,\reg{A}} \ket{0}_\reg{A} \ket{\boti}_\reg{I}, & \Tr_{\reg{I}}[\rho(\ket{\psi^{(\cp)}}_\reg{AI})] = \rho_{\reg A}^{(\cp)}.
    \end{align}
    Now we can consider the hybrid states, \begin{align}
        \ket{\psi^{(\mathsf{Hyb})}_0} &= A_{q,\reg{A}} \dots \cp_{\reg{AI}} A_{0,\reg{A}} \iso_\reg{P} \ket{0}_\reg{A} \ket{{\mathbf P(\boti)}}_\reg{P} \\
        \ket{\psi^{(\mathsf{Hyb})}_1} &= A_{q,\reg{A}} \dots \cp_{\reg{AI}} A_{1,\reg{A}}\iso_\reg{P} \cmf_{\reg{AP}} A_{0,\reg{A}} \ket{0}_\reg{A} \ket{{\mathbf P(\boti)}}_\reg{P} \\
        &\dots \nonumber\\
        \ket{\psi^{(\mathsf{Hyb})}_q} &= \iso_\reg{P} A_{q,\reg{A}} \dots \cmf_{\reg{AP}} A_{0,\reg{A}}  \ket{0}_\reg{A} \ket{{\mathbf P(\boti)}}_\reg{P},
    \end{align}
    where $\ket{\psi^{(\cp)}}_\reg{AI} = \ket{\psi^{(\mathsf{Hyb})}_0}$ and $\iso_{\reg P}\ket{\psi^{(\cmf)}}_\reg{AP} =  \ket{\psi^{(\mathsf{Hyb})}_q}$. For $t \leq q$, let us define \begin{align}
        \ket{\phi_t}_{\reg{AP}} = A_{t,\reg{A}} \dots \cmf_{\reg{AP}} A_{0,\reg{A}}  \ket{0}_\reg{A} \ket{{\mathbf P(\boti)}}_\reg{P},
    \end{align}
    observing that $\norm{\pinsoph_{\reg P} \ket{\phi_t}_{\reg{AP}}} = O(\sqrt{t^4/2^n})$ by \Cref{lem:ele-preserved}, and further the $\ket{\phi_t}_{\reg{AP}} \in \piv_{\reg P}$, lying entirely in the valid subspace as $\cmf$ preserves validity. Observing that $\iso_{\reg P}$ clearly commutes with the $A_{t, \reg A}$ operator, we can then compute the successive differences between hybrid states to be \begin{align}
        \norm{\ket{\psi^{(\mathsf{Hyb})}_{t+1}} - \ket{\psi^{(\mathsf{Hyb})}_{t}}} &= \norm{(\iso_{\reg P} \cdot \cmf_{\reg{AP}}-\cp_{\reg{AI}} \cdot \iso_{\reg P})\ket{\phi_t}_{\reg{AP}}} \nonumber\\
        &= O\left(\sqrt{\frac{t^4}{2^n}}\right). & \text{(By \Cref{cor:ideal-cmf-int})}
    \end{align}
    Therefore, we have $\norm{\iso_{\reg P}\ket{\psi^{(\cmf)}}_\reg{AP} - \ket{\psi^{(\cp)}}_\reg{AI}} = O\left(\sqrt{\frac{q^6}{2^n}}\right)$. This equation combined with \Cref{lem:approx-uhlman} proves the claim.
\end{proof}

\section{Search bounds}
\label{sec:search}
We can also apply this framework to prove new search type query lower bounds. The most important ingredient in such bounds is the fundamental lemma, which relates the compressed permutation oracle database to points the algorithm ``knows''. This result can be applied to show lower bounds for various search problems in which the goal is to find a set of input output points having some property.
All of our bounds follow the same rough format. \begin{enumerate}
    \item Switch to the compressed permutation oracle.
    \item Show that the database is unlikely to contain a solution after $q$ queries.
    \item Apply \Cref{lem:pf-meaningful} to conclude that the adversary is unlikely to have found a solution.
\end{enumerate}

The first step is straightforward, while the second step varies in complexity depending on the problem. Often, it is enough to reason combinatorially about the databases, analogous to compressed function oracles. We will show how to use the fundamental lemma to apply the final step in general.

Throughout this section, our bounds will incur a loss due to the distinguishing advantage against the compressed permutation oracle. We quantify this by the loss term \begin{align}
    \adv{q}&= O\left(\frac{q^{3}}{N^{1/4}}\right),
\end{align}
and this term is often the loosest. Improving the soundness error of our construction would improve all the lower bounds derived. Most would be tight, ignoring the $\adv{q}$ term.

To formalize this procedure, we will capture search problems by a target predicate $\algo R \subset ([N]^{\times 2})^*$, or a set of finite lists of input-output pairs representing the set of potential solutions. Then, we consider an adversary $\algo A$ making $q$ quantum queries to a permutation oracle $\algo O_{\varphi}$. The predicate $\algo R$ is known to the adversary (meaning its intermediate operations may depend on $\algo R$), although naturally the underlying permutation is not, outside of the information available by querying the oracle. The goal of the adversary is to find some input-output pairs of the permutation which satisfy the predicate. In particular, the game proceeds as follows.
\begin{game}
    Predicate search.
    \begin{enumerate}
        \item A random $\varphi \sim \mathbf S_N$ is selected to construct $\algo O_{\varphi}$.
        \item An adversary $\algo A$ makes $q$ quantum queries to $\algo O_{\varphi}$, and outputs a list $\vec x = [(x_1, y_1), \dots, (x_l, y_l)]$.
        \item The adversary wins if $\varphi(x_i)=y_i$ for all $i$, and $\vec x \in \algo R$.
    \end{enumerate}
    \label{game:rel-search}
\end{game}

Alternatively, we could consider a game in which the adversary queries the compressed permutation oracle, and the final condition involves only analyzing the compressed oracle database. We say an $I \in \mathbf I_t$ satisfies $\algo R$, denoted $I \cap \algo R \neq \emptyset$, if there exists $\vec x = [(x_1, y_1), \dots, (x_l, y_l)]$ such that $I(x_i)=y_i$ for all $i \in [l]$ and $\vec x \in \algo R$. We define the predicate projector \begin{align}
    \Pi_{\algo R} \ket{I} &= \begin{cases}
        \ket{I} & \text{(If $I \cap \algo R \neq \emptyset$)} \\
        0 & \text{(Otherwise)}
    \end{cases}
\end{align}
which projects onto databases which contain a set of  input-output pairs in the predicate.

\begin{game}
    Predicate search with compressed oracles.
    \begin{enumerate}
        \item The compressed databases is initialized to $\ket{\boti}_{\reg I}$.
        \item The adversary $\algo A$ with workspace $\reg A$ makes $q$ quantum queries to $\cp_{\reg{AI}}$.
        \item The projective measurement $\algo M_{\algo R} = \{\Pi_{\algo R}, \Pi_{\algo R}^{\perp}\}$ is applied to register $\reg I$. The adversary wins if the first outcome is achieved.
    \end{enumerate}
    \label{game:rel-search-comp}
\end{game}

It turns out that the winning probability of the second game provides an upper bound on the first. We can compare the same adversary across the two games by simply ignoring the output of $\algo A$ in the second game. This is analogous to the result of \cite{Zhandry2018} relating the compressed function oracle database to the winning probability of a search game.

\begin{lemma}
    Suppose that $\algo A$ wins \Cref{game:rel-search-comp} with probability $p_2$, and wins \Cref{game:rel-search} with probability $p_1$. Then we have the bound \begin{align}
        \sqrt{p_1} \leq \sqrt{p_2} + \frac{l}{\sqrt{N-q-l}} + \adv{q+l}.
    \end{align}
    This bound also implies 
    \begin{align}
        p_1 = O\left(p_2 + \frac{l^2}{N-q-l} + \adv{q+l}^2\right).
    \end{align}
    \label{lem:search-bound}
\end{lemma}

\begin{proof}
    For \Cref{game:rel-search}, we can suppose WLOG that the adversary outputs some $\vec x \in \algo R$. If not, then modify the last step to check whether the output satisfies the predicate, and if needed switch to some output which satisfies the predicate. This cannot decrease the success probability, as the adversary is guaranteed to lose if $\vec x \not\in \algo R$. Call this modified adversary $\algo A'$. Now, consider running \Cref{game:rel-search} with the compressed permutation oracle and adversary $\algo A'$, replacing the final check with the projector $\Pi_{\reg{AI}}^{(\algo O_{\varphi})}$ from \Cref{lem:pf-meaningful}. The resultant state is a distance at most $\adv{q+l}$, by the definition of the distinguishing advantage.
    
    Then, we can replace this projector with $\Pi_{\reg{AI}}^{(\cp)}$ by \Cref{lem:pf-meaningful}, and the resultant state is a distance at most $\frac{l}{\sqrt{N-q-l}}$. Now, this projector can only succeed if there is a tuple of input-output pairs in the database which are also in $\algo R$, i.e. the state is in $\Pi_{\algo R, \reg I}$. The magnitude of this component is then at most $\sqrt{p_2}$.
\end{proof}

\subsection{Simple examples}
\label{sec:search-simple}

The simplest example of a problem for which we can show a lower bound is the one-more problem. In this problem, an adversary makes $q$ queries to a random permutation and is tasked with finding $q+1$ input-output pairs.

\begin{problem}
    Sample $\varphi \sim \mathbf S_N$, to produce two way accessible oracle $\algo O_{\varphi}$. $\algo A$ makes $q$ queries to $\algo O_{\varphi}$, and succeeds if it finds $\{(x_1, y_1), \dots, (x_{q+1}, y_{q+1})\}$ such that $\varphi(x_i)=y_i$ for all $i \in [q+1]$, and $x_i \neq x_{i'}$ for $i \neq i'$.
    \label{prob:one-more}
\end{problem}

This problem does not appear to be in the scope of existing techniques for permutation problems. However, it is a simple corollary of our framework.

\begin{theorem}
    No $q$ query quantum algorithm $\algo A$ can solve the one more problem, \Cref{prob:one-more}, with probability exceeding \begin{align}
        p_{\algo A} &= O\left(\frac{q^2}{N}+ \adv{q}^2\right).
    \end{align}
    In particular, exponential queries are required for constant error.
    \label{thm:one-more}
\end{theorem}

\begin{proof}
    We can analyze this game in the compressed permutation oracle picture, as in \Cref{game:rel-search-comp}. In this context, the predicate $\algo R$ is the set of all length $q+1$ lists of distinct input-output pairs. After $q$ queries, the database is of size at most $q$ by \Cref{lem:pf-bounded}. Therefore, no satisfying list can appear in the database, so the success probability of the adversary in this game is zero. The claim now follows from \Cref{lem:search-bound}.
\end{proof}

Another simple problem outside the scope of existing techniques is the cycle finding problem. In this problem, the adversary is tasked with finding a cycle in the permutation.
\\
\begin{problem}
    Sample $\varphi \sim \mathbf S_N$, to produce two way accessible oracle $\algo O_{\varphi}$. $\algo A$ makes $q$ queries to $\algo O_{\varphi}$, and succeeds if it finds $x_1, \dots, x_l$ such that $\varphi(x_i)=x_{i+1 \mod l}$ for all $i \in [l]$.
    \label{prob:cycle}
\end{problem}

\begin{theorem}
    No $q$ query quantum algorithm $\algo A$ can solve cycle finding, \Cref{prob:cycle}, with probability exceeding \begin{align}
        p_{\algo A} &= O\left(\frac{q^3}{N} + \adv{q}^2\right).
    \end{align}
    In particular, exponential queries are required for constant error.
    \label{thm:cycle}
\end{theorem}

\begin{proof}
    We use the framework of \Cref{game:rel-search-comp}, where the predicate $\algo R$ is the set of all cycles. Let $\algo A$ be an algorithm with Hilbert space $\algo H_{\reg A}$ containing query subregisters $\reg{XY}$, and intermediate operations $A_0, \dots, A_q$ on this space. The success probability of this algorithm is given by \begin{align}
        p_{\algo A} = \norm{\Pi_{\algo R, \reg I} A_{q, \reg A} \cp_{\reg{AI}} \dots A_{1, \reg A} \cp_{\reg{AI}} A_{0, \reg A} \ket{0}_{\reg A} \ket{\boti}_{\reg A}}^2.
    \end{align}
    To bound this quantity, we will show a bound on the commutator norm \begin{align}
        \norm{[\Pi_{\algo R, \reg I}, \cp_{\reg{AI}}]}_{\leq t} &= O\left(\sqrt{\frac{t}{N}}\right). \label{eqn:cycle-comm}
    \end{align}
    To do this, we can modify the compression operator such that it cannot add cycles to the database. In particular, fix an $I \in \mathbf I_t$ and $x \not\in \dom(I)$. We will first define an operator on $\algo H( I|^x)$, writing \begin{align}
        \ket{\widetilde +_{x, I}}_{\reg I} &\coloneqq \frac{1}{\sqrt{N - |\dom(I) \cup \im(I)|}} \sum_{y \in [N] \setminus \dom(I) \setminus \im(I)} \ket{I[x \ra y]} \\
        \widetilde \pc_{x, I} &\coloneqq \Id - \outerprod{I}{I} - \outerprod{\widetilde +_{x, I}}{\widetilde +_{x, I}} + \outerprod{\widetilde +_{x, I}}{I} + \outerprod{I}{\widetilde +_{x, I}} \\
        \widetilde \pc_{\reg{XI}} &= \bigoplus_{x \in [N], I \in \mathbf I^x} \outerprod{x}{x}_{\reg X} \otimes \pc_{x, I, \reg I}.
    \end{align}
    Note that the above are unitary operators so long as the database is of size less than $N/2$, which is anyways beyond where the bound trivializes. 
    Observe that $\norm{\ket{\widetilde +_{x, I}} - \ket{+_{x, I}}}_{\leq t} = O(\sqrt{t/N})$, so $\norm{\pc - \widetilde \pc}_{\leq t} = O(\sqrt{t/N})$ by \Cref{lem:dirsum-norm}. Now we have \begin{align}
        [\Pi_{\algo R, \reg I}, \widetilde \pc_{\reg{XI}}] &= 0, &[\Pi_{\algo R, \reg I}, \flip_{\reg I} \widetilde \pc_{\reg{XI}} \flip_{\reg I}^\dagger ] &= 0, \\
        [\Pi_{\algo R, \reg I}, \pu_{\reg{XYI}}] &= 0, & [\Pi_{\algo R, \reg I}, \flip_{\reg I} \pu_{\reg{XYI}} \flip_{\reg I}^\dagger] &= 0, 
    \end{align}
    as $\widetilde \pc$ cannot add a cycle to the database.
    If we write \begin{align}
        \widetilde \cp_{\reg{XYI}} &= \outerprod{0}{0}_{\reg X} \otimes (\widetilde \pc_{\reg{XI}} \cdot \pu_{\reg{XYI}} \cdot \widetilde \pc_{\reg{XI}}^\dagger) + \outerprod{1}{1}_{\reg X} \otimes (\flip_{\reg I}\cdot \widetilde \pc_{\reg{XI}} \cdot \pu_{\reg{XYI}} \cdot \widetilde \pc_{\reg{XI}}^\dagger \cdot \flip_{\reg I}^\dagger) \\
        &= \outerprod{0}{0}_{\reg X} \otimes (\widetilde \pc_{\reg{XI}} \cdot \pu_{\reg{XYI}} \cdot \widetilde \pc_{\reg{XI}}^\dagger) + \outerprod{1}{1}_{\reg X} \otimes (\flip_{\reg I}\cdot \widetilde \pc_{\reg{XI}} \cdot (\flip^\dagger \flip)_{\reg I} \cdot \pu_{\reg{XYI}} \cdot (\flip^\dagger \flip)_{\reg I} \cdot \widetilde \pc_{\reg{XI}}^\dagger \cdot \flip_{\reg I}^\dagger),
    \end{align}
    then we have $\norm{\cp - \widetilde \cp}_{\leq t} = O(\sqrt{t/N})$, and $\widetilde \cp_{\reg{XYI}} \Pi_{\algo R, \reg I} = \Pi_{\algo R, \reg I} \widetilde \cp_{\reg{XYI}}$, so the claimed \Cref{eqn:cycle-comm} follows. Now we can write \begin{align}
        \sqrt{p_{\algo A}} &= \norm{\Pi_{\algo R, \reg I} A_{q, \reg A} \cp_{\reg{AI}} \dots A_{1, \reg A} \cp_{\reg{AI}} A_{0, \reg A} \ket{0}_{\reg A} \ket{\boti}_{\reg A}} \\
        &= \norm{ A_{q, \reg A} \cp_{\reg{AI}} \Pi_{\algo R, \reg I} A_{q-1, \reg A} \cp_{\reg{AI}} \dots A_{1, \reg A} \cp_{\reg{AI}} A_{0, \reg A} \ket{0}_{\reg A} \ket{\boti}_{\reg A}} + O\left(\sqrt{\frac{q}{N}}\right) \\
        &\dots  \\
        &= \norm{A_{q, \reg A} \cp_{\reg{AI}} \dots A_{1, \reg A} \cp_{\reg{AI}} A_{0, \reg A} \Pi_{\algo R, \reg I} \ket{0}_{\reg A} \ket{\boti}_{\reg A}} + \sum_{t=1}^q O\left(\sqrt{\frac{t}{N}}\right) \\
        &= O\left(\sqrt{\frac{q^3}{N}}\right).
    \end{align}
\end{proof}

\subsection{A general search lower bound}

The proof strategy developed for the cycle finding lower bounding can be applied much more broadly, to give a lower bound for any search problem. This bound is in terms of a quantity we refer to as the \emph{sparsity} of the underlying predicate. Intuitively, this quantity measures the likelihood that each new query results in a solution being added to the database. Recalling that we allow forwards and inverse queries, we define sparsity in a two-sided way that allows both direction of query.

\begin{definition}
    Let $\algo R \subset ([N]^{\times 2})^*$ be a predicate, and $I \in \mathbf I_t$ such that $I \cap \algo R = \emptyset$. Then we define the $I$-sparsity $s_I$ and $t$-sparsity $s_t$ as \begin{align}
        s_I(\algo R) &= \max \Bigg\{\max_{x \in [N] \setminus \dom(I)} |\{y \not\in \im(I) \, : \, I[x\ra y] \cap \algo R \neq \emptyset\}|, \nonumber\\&\qquad\qquad\quad\max_{y \in [N] \setminus \im(I)} |\{x \not\in \dom(I)\, : \, I[x\ra y] \cap \algo R \neq \emptyset\}| \Bigg\} \nonumber \\
        s_t(\algo R) &= \max_{I \in \mathbf I_{\leq t} \, : \, I \cap \algo R = \emptyset} s_I(\algo R).
    \end{align}
    \label{def:sparsity}
\end{definition}

Now, we can give a lower bound for any predicate search problem in terms of this quantity. Note that special cases of this theorem are known for predicates of constant size, meaning subsets of $([N]^{\times 2})^{O(1)}$, due to \cite{MMW24,cojocaru25lifting}.

\begin{theorem}
    Let $\algo R \subset ([N]^{\times 2})^*$ be a predicate, and $\algo A$ be a quantum algorithm making $q$ bi-directional queries to a random permutation of size $[N]$ to output a list $\mathbf x = [(x_1, y_1), \dots, (x_l, y_l)]$ of input-output pairs. Then, $\algo A$ wins \Cref{game:rel-search} under predicate $\algo R$ with probability at most \begin{align}
        p_{\algo A} &= O\left(\adv{q}^2 + \frac{l^2}{N} + \left(\sum_{t=0}^{q-1} \sqrt{\frac{s_t(\algo R)}{N}}\right)^2 \right).
    \end{align}
    \label{thm:sparsity-lb}
\end{theorem}

\begin{proof}
    We use the framework of \Cref{game:rel-search-comp}. Let $\algo A$ act on Hillbert space $\algo H_{\reg A}$ containing query registers $\reg{XY}$ with intermediate operations $A_0, \dots, A_q$ on this space. The success probability of this algorithm is given by \begin{align}
        p_{\algo A} = \norm{\Pi_{\algo R, \reg I} A_{q, \reg A} \cp_{\reg{AI}} \dots A_{1, \reg A} \cp_{\reg{AI}} A_{0, \reg A} \ket{0}_{\reg A} \ket{\boti}_{\reg A}}^2.
    \end{align}
    To bound this quantity, we will show a bound on the commutator norm \begin{align}
        \norm{[\Pi_{\algo R, \reg I}, \cp_{\reg{AI}}]}_{\leq t} &= O\left(\sqrt{\frac{s_t(\algo R)}{N}}\right). \label{eqn:rel-comm}
    \end{align}
    To do this, we can modify the compression operator such that it preserves dissatisfaction of $\algo R$ perfectly. In particular, fix an $I \in \mathbf I_t$ and $x \not\in \dom(I)$. We define the set $\algo R_{x, I}^{\ra}$ to be the set of images whose (forwards) assignment on $x$ would result in the predicate being satisfied, or the empty set if the predicate is already satisfied. Formally, \begin{align}
        \algo R_{x, I}^{\ra} &\coloneqq \begin{cases}
            \{y \, : \, I[x \ra y] \in \algo R\} & \text{(If $I \cap \algo R = \emptyset$)} \\
            \emptyset & \text{(Otherwise)}
        \end{cases}
    \end{align}
    Observe that $|\algo R_{x, I}^\ra| \leq s_I(\algo R)$. Then we write \begin{align}
        \ket{\widetilde +_{x, I}^\ra}_{\reg I} &\coloneqq \frac{1}{\sqrt{N - |\algo R_{x, I}^\ra \cup \im(I)|}} \sum_{y \in [N] \setminus \im(I) \setminus \algo R_{x, I}^\ra} \ket{I[x \ra y]} \\
        \widetilde \pc_{x, I}^\ra &\coloneqq \Id - \outerprod{I}{I} - \outerprod{\widetilde +_{x, I}^\ra}{\widetilde +_{x, I}^\ra} + \outerprod{\widetilde +_{x, I}^\ra}{I} + \outerprod{I}{\widetilde +_{x, I}^\ra} \\
        \widetilde \pc_{\reg{XI}}^\ra &= \bigoplus_{x \in [N], I \in \mathbf I^x} \outerprod{x}{x}_{\reg X} \otimes \widetilde \pc_{x, I}^\ra.
    \end{align}
    In the other direction, we define the set $\algo R_{y, I}^\la$ to be the set of preimages whose (inverse) assignment on $y$ would result in the predicate being satisfied, and otherwise the empty set. Formally, \begin{align}
        \algo R_{y, I}^{\la} &\coloneqq \begin{cases}
            \{x \, : \, (I[y \ra x])^{-1} \in \algo R\} & \text{(If $I^{-1} \cap \algo R = \emptyset$)} \\
            \emptyset & \text{(Otherwise)}
        \end{cases}
    \end{align}
    Observe that $|\algo R_{y, I}^\la| \leq s_{I^{-1}}(\algo R)$. Then we write \begin{align}
        \ket{\widetilde +_{y, I}^\la}_{\reg I} &\coloneqq \frac{1}{\sqrt{N - |\algo R_{y, I}^\la \cup \im(I)|}} \sum_{x \in [N] \setminus \im(I) \setminus \algo R_{y, I}^\la} \ket{I[y \ra x]} \\
        \widetilde \pc_{y, I}^\la &\coloneqq \Id - \outerprod{I}{I} - \outerprod{\widetilde +_{y, I}^\la}{\widetilde +_{y, I}^\la} + \outerprod{\widetilde +_{y, I}^\la}{I} + \outerprod{I}{\widetilde +_{y, I}^\la} \\
        \widetilde \pc_{\reg{XI}}^\la &= \bigoplus_{y \in [N], I \in \mathbf I^y} \outerprod{y}{y}_{\reg X} \otimes \widetilde \pc_{y, I}^\la.
    \end{align}
    
    Observe that $\norm{\ket{\widetilde +_{x, I}^\ra} - \ket{+_{x, I}}} = O(\sqrt{|\algo R_{x, I}^\ra|/N})$, and $\norm{\ket{\widetilde +_{y, I}^\la} - \ket{+_{y, I}}} = O(\sqrt{|\algo R_{y, I}^\la|/N})$. It follows that $\norm{\pc - \widetilde \pc^\ra}_{\leq t} = O(\sqrt{s_t(\algo R)/N})$ and $\norm{\pc - \widetilde \pc^\la}_{\leq t} = O(\sqrt{s_t(\algo R)/N})$ by \Cref{lem:dirsum-norm} and the prior observation.
    Now we have \begin{align}
        [\Pi_{\algo R, \reg I}, \widetilde \pc_{\reg{XI}}^\ra] &= 0, &[\Pi_{\algo R, \reg I}, \flip_{\reg I} \cdot \widetilde \pc_{\reg{XI}}^\la \cdot \flip_{\reg I}^\dagger ] &= 0, \\
        [\Pi_{\algo R, \reg I}, \pu_{\reg{XYI}}] &= 0, & [\Pi_{\algo R, \reg I}, \flip_{\reg I} \cdot \pu_{\reg{XYI}} \cdot \flip_{\reg I}^\dagger] &= 0,
    \end{align}
    as the modified operators cannot add any input-output pair which causes $\algo R$ to be satisfied.
    If we write \begin{align}
        \widetilde \cp_{\reg XYI} &= \outerprod{0}{0}_{\reg X} \otimes (\widetilde \pc_{\reg{XI}}^\ra \cdot \pu_{\reg{XYI}} \cdot \widetilde \pc_{\reg{XI}}^{\ra, \dagger}) + \outerprod{1}{1}_{\reg X} \otimes (\flip_{\reg I}\cdot \widetilde \pc_{\reg{XI}}^\la \cdot \pu_{\reg{XYI}} \cdot \widetilde \pc_{\reg{XI}}^{\la, \dagger} \cdot \flip_{\reg I}^\dagger) \\
        &= \outerprod{0}{0}_{\reg X} \otimes (\widetilde \pc_{\reg{XI}}^\ra \cdot \pu_{\reg{XYI}} \cdot \widetilde \pc_{\reg{XI}}^{\ra, \dagger}) + \outerprod{1}{1}_{\reg X} \otimes (\flip_{\reg I}\cdot \widetilde \pc_{\reg{XI}}^\la \cdot (\flip^\dagger \flip)_{\reg I} \cdot \pu_{\reg{XYI}} \cdot (\flip^\dagger \flip)_{\reg I} \cdot \widetilde \pc_{\reg{XI}}^{\la, \dagger} \cdot \flip_{\reg I}^\dagger),\nonumber
    \end{align}
    then we have $\norm{\cp - \widetilde \cp}_{\leq t} = O(\sqrt{s_t(\algo R)/N})$, and $\widetilde \cp_{\reg{XYI}} \Pi_{\algo R, \reg I} = \Pi_{\algo R, \reg I} \widetilde \cp_{\reg{XYI}}$, so the claimed \Cref{eqn:rel-comm} follows. Now we can write \begin{align}
        \sqrt{p_{\algo A}} &= \norm{\Pi_{\algo R, \reg I} A_{q, \reg A} \cp_{\reg{AI}} \dots A_{1, \reg A} \cp_{\reg{AI}} A_{0, \reg A} \ket{0}_{\reg A} \ket{\boti}_{\reg A}} \\
        &= \norm{ A_{q, \reg A} \cp_{\reg{AI}} \Pi_{\algo R, \reg I} A_{q-1, \reg A} \cp_{\reg{AI}} \dots A_{1, \reg A} \cp_{\reg{AI}} A_{0, \reg A} \ket{0}_{\reg A} \ket{\boti}_{\reg A}} + O\left(\sqrt{\frac{s_{q-1}(\algo R)}{N}}\right) \\
        &\dots  \\
        &= \norm{A_{q, \reg A} \cp_{\reg{AI}} \dots A_{1, \reg A} \cp_{\reg{AI}} A_{0, \reg A} \Pi_{\algo R, \reg I} \ket{0}_{\reg A} \ket{\boti}_{\reg A}} + \sum_{t=0}^{q-1} O\left(\sqrt{\frac{s_t(\algo R)}{N}}\right) \\
        &= \sum_{t=0}^{q-1} O\left(\sqrt{\frac{s_t(\algo R)}{N}}\right).
    \end{align}
    This, combined with \Cref{lem:search-bound} implies the claim.
\end{proof}

One of the simplest examples from the literature of a predicate search problem is the double-sided zero search problem, due to Unruh \cite{Unruh2021,Unruh2023}.

\begin{problem}
    Sample $\varphi \sim \mathbf S_N$ for $N=2^{2n}$, to produce two way accessible oracle $\algo O_{\varphi}$. $\algo A$ makes $q$ queries to $\algo O_{\varphi}$, and succeeds if it finds $(x \Vert 0^n, y \Vert 0^n)$ such that $\varphi(x \Vert 0^n) = y \Vert 0^n)$.
    \label{prob:DSZS}
\end{problem}

This is a search problem with predicate $\algo R = \{[(x\Vert 0^n, y \Vert 0^n)] \, : \, x, y \in \bit^n\}$. It is easy to compute $s_t(\algo R) = O(2^{n})$ for $t = o(2^n)$, giving us the following bound.

\begin{corollary}[of \Cref{thm:sparsity-lb}]
    Any algorithm $\algo A$ making $q$ quantum queries solves \Cref{prob:DSZS} with probability at most \begin{align}
        p_{\algo A} &= O\left(\frac{q^2}{2^n} + \adv{q}^2\right).
    \end{align}
\end{corollary}

Lower bounds for this problem are known already \cite{carolan24oneway,MMW24,cojocaru25lifting}, even tight ones. Observe that our bound would be tight, ignoring that $\adv{q}^2$ term. The problems we will look at next are less well understood, and for many no tight bound is known. The bounds we show as applications of our general theorem will in some cases improve on known results and/or be tight. 

Even more intriguingly, all of them would be tight if the $\adv{q}^2$ term were ignored. This suggests that our construction accurately captures the difficulty of these problems, and improving our soundness analysis provides a direct path towards showing tight quantum query lower bounds for many problems in the random permutation model.

\subsection{Application to sponge}
The sponge construction lifts an invertible permutation to a hash function with arbitrary domain. It is parameterized by integers $r, c \in \mathbb N$ and the permutation $\varphi : \bit^{r+c} \ra \bit^{r+c}$. To hash a message $m \in (\bit^r)^*$ of the form $m=m_1 \Vert m_2 \Vert \dots \Vert m_l$, we initialize the internal state to $u_0 = 0^{r+c}$. Then, for every $i$ from $1$ up to $l$, we XOR $m_i$ into the first $r$ bits of $u$ and apply $\varphi$, obtaining $u_{i+1} = \varphi(u_i \oplus (m_i \Vert 0^c))$. The output is then the first $r$ bits of the final state $u_l$.\footnote{For simplicity, we consider only one round of squeezing.} This construction is due to Bertoni et al \cite{BDPvA07,BDPvA08}, and when instantiated with the Keccak permutation underlies the NIST standardized hash function SHA3 \cite{KeccakSub3,KeccakSponge3}.

It is known that this construction is quantum indifferentiable from a random oracle \cite{alagic2025sponge}, which in particular implies preimage and collision resistance. The best known bound for these is due to the same work, which shows $\tilde \Omega(2^{\min(r, c)/5})$ queries are required for constant success probability at either collision or preimage finding. This result required a long and detailed proof tailored to this specific construction, and is not tight. Using our framework, we can improve these bounds in a direct manner.

We will use the notation $x[a : b]$ to denote the $a$-th (inclusive) through $b$-th (exclusive) bits of string $x$ with zero based indexing. Then, we can define the preimage predicate of the sponge for image $w \in \bit^r$ as \begin{align}
    \algo R^{w-\mathsf{pre}} &\coloneqq \left\{[(x_1, y_1), \dots, (x_l, y_l)] \, : \, \substack{
        x_1[r:r+c] = 0^c, \text{ and }\\
        y_i[r:r+c] = x_{i+1}[r:r+c], \text{ and }\\
        y_l[0:r] = z.}\right\} \label{eqn:sponge-w-pre}
\end{align}

\begin{problem}[Sponge inversion.]
    Given quantum query access to an ideal random permutation $\algo O_{\varphi}$, find a sponge pre-image of some fixed $w \in \bit^r$. Equivalently, solve the search problem defined by predicate $\algo R^{w-\mathsf{pre}}$.
    \label{prob:sponge-inv}
\end{problem}

\begin{theorem}
    A $q$ query quantum algorithm $\algo A$ solves sponge inversion, \Cref{prob:sponge-inv}, with probability no more than \begin{align}
        p_{\algo A} = O\left(\frac{q^2}{2^r} + \frac{q^3}{2^c} + \adv{q}^2 \right)
    \end{align}
    \label{thm:sponge-inv}
\end{theorem}

\begin{proof}
    To show this bound, we will change the predicate slightly. In particular, define the internal $z$-preimage predicate as \begin{align}
        \algo R^{z-\mathsf{ipre}} &\coloneqq \left\{[(x_1, y_1), \dots, (x_l, y_l)] \, : \substack{x_1[r:r+c] = 0^c, \text{ and} \\
        y_i[r:r+c] = x_{i+1}[r:r+c], \text{ and } \\
        y_l[r:r+c] = z.}\right\} \label{eqn:sponge-z-opre},
    \end{align}
    which can be used to construct the internal collision predicate \begin{align}
        \algo R^{\mathsf{icol}} &\coloneqq \algo R^{0^c-\mathsf{ipre}} \cup \bigcup_{z \in \bit^c \setminus 0^c} \algo R^{z-\mathsf{ipre}} \times \algo R^{z-\mathsf{ipre}}. \label{eqn:sponge-icol}
    \end{align}
    We then let $\algo R' = \algo R^{\mathsf{icol}} \cup \algo R^{w-\mathsf{pre}}$; in other words, we allow the adversary to find either a pre-image of $w$ or an internal collision. Clearly, this task is no harder than finding a sponge preimage. It is known that $s_t(\algo R')=O\left(2^{c} + t\cdot 2^{r}\right)$ \cite{LM22}, from which the claim follows.
\end{proof}

We can also define the sponge collision predicate, \begin{align}
    \algo R^{\mathsf{col}} &\coloneqq \bigcup_{w \in \bit^r} \algo R^{w-\mathsf{pre}} \times \algo R^{w-\mathsf{pre}}. \label{eqn:sponge-col}
\end{align}

This gives rise to the following search task.

\begin{problem}[Sponge collision.]
    Given quantum query access to a random permutation oracle $\algo O_{\varphi}$, find a sponge collision. Equivalently, solve the search problem defined by predicate $\algo R^{\mathsf{col}}$.
    \label{prob:sponge-col}
\end{problem}

\begin{theorem}
    A $q$ query quantum algorithm $\algo A$ solves sponge collision, \Cref{prob:sponge-inv}, with probability no more than \begin{align}
        p_{\algo A} = O\left(\frac{q^3}{2^{\min(r, c)}} + \adv{q}^2 \right)
    \end{align}
    \label{thm:sponge-col}
\end{theorem}

\begin{proof}
    To show this bound, we will change the predicate slightly. In particular, we once again allow the adversary to find either a sponge collision or an internal collision, letting $\algo R' = \algo R^{\mathsf{icol}} \cup \algo R^{\mathsf{col}}$. It is known that $s_t(\algo R')=O\left(t\cdot 2^r+t\cdot 2^c\right)$ \cite{Unruh2021}, from which the claim follows.
\end{proof}

Observe that using quantum collision and claw finding \cite{BHT97}, there is an algorithm solving \Cref{prob:sponge-col} with probability $O(q^3 2^{-\min(r, c)})$. It follows that \Cref{thm:sponge-col} is tight whenever $\min(r, c)/3 \leq (r+c)/12$, and improves on the prior art \cite{alagic2025sponge} when $\min(r, c)/5 \leq (r+c)/12$. Furthermore, ignoring the $\adv{q}^2$ term the bound would be tight for any $r$ and $c$.

Similarly using Grover search \cite{Gro96} or claw finding, there is an algorithm solving \Cref{prob:sponge-inv} with probability $O(q^3 / 2^r+q^2/2^{c})$. Then \Cref{thm:sponge-inv} is tight whenever $\min(r/2, c/3) \leq (r+c)/12$, and improves on the prior art \cite{alagic2025sponge} when $\min(r, c)/5 \leq (r+c)/12$. Again, ignoring the $\adv{q}^2$ term the bound would be tight for any $r$ and $c$.

\subsection{Application to Davies-Meyer}

Davies-Meyer is a construction that turns an ideal cipher into a fixed length compressing hash function \cite{winternetz84dm}. It underlies SHA1, SHA2, and MD5. A cipher is a keyed permutation, and the quantum ideal cipher model allows quantum query access across both the key and the resulting permutation. We will state our results for the key-less case for simplicity.\footnote{Our framework can be applied to the ideal cipher model, though we leave a detailed exploration to future work.}

The Davies-Meyer function $\mathsf{DM} : \bit^n \ra \bit^n$ is defined by a permutation $\varphi : \bit^n \ra \bit^n$, and is \begin{align}
    \mathsf{DM}^\varphi(x) &\coloneqq \varphi(x) \oplus x.
\end{align}

\begin{problem}[Davies-Meyer inversion.]
    Given bidirectional quantum query access to a random permutation $\varphi : \bit^n \ra \bit^n$, find an $x \in \bit^n$ such that $\mathsf{DM}^\varphi(x) = 0^n$.
    \label{prob:dm-inv}
\end{problem}

\begin{theorem}
    Any $q$ query quantum algorithm $\algo A$ solves Davies-Meyer inversion, \Cref{prob:dm-inv}, with probability bounded by \begin{align}
        p_{\algo A} &= O\left(\frac{q^2}{2^n} + \adv{q}^2\right).
    \end{align}
    \label{thm:dm-inv}
\end{theorem}

\begin{proof}
    We can define the zero preimage predicate $\algo R = \{[(x, x)] \, : \, x \in \bit^n\}$, which satisfies $s_t(\algo R) = 1$. Then the claim follows from \Cref{thm:sparsity-lb}.
\end{proof}

\begin{problem}[Davies-Meyer collision.]
    Given bidirectional quantum query access to a random permutation $\varphi : \bit^n \ra \bit^n$, find a pair $x \neq x' \in \bit^n$ such that $\mathsf{DM}^{\varphi}(x) = \mathsf{DM}^\varphi(x')$.
    \label{prob:dm-col}
\end{problem}

\begin{theorem}
    Any $q$ query quantum algorithm $\algo A$ solves Davies-Meyer collision, \Cref{prob:dm-col}, with probability bounded by \begin{align}
        p_{\algo A} &= O\left(\frac{q^3}{2^n} + \adv{q}^2\right).
    \end{align}
    \label{thm:dm-col}
\end{theorem}

\begin{proof}
    We can define the collision predicate \begin{align}
        \algo R = \left\{[(x, y), (x', y')] \, : \, \substack{x,y,x',y' \in \bit^n, \text{ and } \\
        x\oplus y = x' \oplus y', \text{ and }\\
        x \neq x', y \neq y'.}\right\},
    \end{align}
    which satisfies $s_t(\algo R) = t$. Then the claim follows from \Cref{thm:sparsity-lb}.
\end{proof}

Note that \cite{hosoyamanda18dm,cojocaru25lifting} prove bounds for preimage and collision resistance of the Davies-Meyer construction, and these bounds are tighter than the ones given here. However, no tight bound is known for collision resistance: the best known is $\Omega(2^{n/4})$ for constant success probability due to \cite{cojocaru25lifting}.

Ignoring the $\adv{q}^2$ term, our bound would be $\Omega(2^{n/3})$, matching quantum collision finding. Tightening our framework then provides a path towards deriving a tight collision bound for the Davies-Meyer construction.

\pagebreak

\printbibliography

@misc{hosoyamanda18dm,
      author = {Akinori Hosoyamada and Kan Yasuda},
      title = {Building Quantum-One-Way Functions from Block Ciphers: Davies-Meyer and Merkle-Damgård Constructions},
      howpublished = {Cryptology {ePrint} Archive, Paper 2018/841},
      year = {2018},
      url = {https://eprint.iacr.org/2018/841}
}

@misc{alagic2025sponge,
      author = {Gorjan Alagic and Joseph Carolan and Christian Majenz and Saliha Tokat},
      title = {The Sponge is Quantum Indifferentiable},
      howpublished = {Cryptology {ePrint} Archive, Paper 2025/731},
      year = {2025},
      url = {https://eprint.iacr.org/2025/731}
}

@INPROCEEDINGS{winternetz84dm,
  author={Winternitz, Robert S.},
  booktitle={1984 IEEE Symposium on Security and Privacy}, 
  title={A Secure One-Way Hash Function Built from DES}, 
  year={1984},
  volume={},
  number={},
  pages={88-88},
  doi={10.1109/SP.1984.10027}}

@misc{cojocaru25lifting,
      author = {Alexandru Cojocaru and Minki Hhan and Qipeng Liu and Takashi Yamakawa and Aaram Yun},
      title = {Quantum Lifting for Invertible Permutations and Ideal Ciphers},
      howpublished = {Cryptology {ePrint} Archive, Paper 2025/738},
      year = {2025},
      url = {https://eprint.iacr.org/2025/738}
}

@misc{yuen2013quantumlowerbounddistinguishing,
      title={A quantum lower bound for distinguishing random functions from random permutations}, 
      author={Henry Yuen},
      year={2013},
      eprint={1310.2885},
      archivePrefix={arXiv},
      primaryClass={cs.CC},
      url={https://arxiv.org/abs/1310.2885}, 
}

@inproceedings{blowfish,
  title={Description of a New Variable-Length Key, 64-bit Block Cipher (Blowfish)},
  booktitle={Fast Software Encryption, Cambridge Security Workshop, Cambridge, UK, December 9-11, 1993, Proceedings},
  series={Lecture Notes in Computer Science},
  publisher={Springer},
  volume={809},
  pages={191-204},
  doi={10.1007/3-540-58108-1_24},
  author={Bruce Schneier},
  year=1993
}

@techreport{fips46,
  author       = {National Institute of Standards, U.S. Department of Commerce},
  title        = {Data Encryption Standard (DES)},
  institution  = {National Bureau of Standards},
  number       = {Federal Information Processing Standards Publication 46},
  year         = {1977},
  address      = {Washington, D.C.},
  url          = {https://csrc.nist.gov/pubs/fips/46/final}
}

@InProceedings{liu19multicols,
author="Liu, Qipeng
and Zhandry, Mark",
editor="Ishai, Yuval
and Rijmen, Vincent",
title="On Finding Quantum Multi-collisions",
booktitle="Advances in Cryptology -- EUROCRYPT 2019",
year="2019",
publisher="Springer International Publishing",
address="Cham",
pages="189--218",
isbn="978-3-030-17659-4"
}

@InProceedings{carolan24oneway,
author="Carolan, Joseph
and Poremba, Alexander",
editor="Reyzin, Leonid
and Stebila, Douglas",
title="Quantum One-Wayness of the Single-Round Sponge with Invertible Permutations",
booktitle="Advances in Cryptology -- CRYPTO 2024",
year="2024",
publisher="Springer Nature Switzerland",
address="Cham",
pages="218--252",
isbn="978-3-031-68391-6"
}

@InProceedings{unruh14revocable,
author="Unruh, Dominique",
editor="Nguyen, Phong Q.
and Oswald, Elisabeth",
title="Revocable Quantum Timed-Release Encryption",
booktitle="Advances in Cryptology -- EUROCRYPT 2014",
year="2014",
publisher="Springer Berlin Heidelberg",
address="Berlin, Heidelberg",
pages="129--146",
isbn="978-3-642-55220-5"
}

@inproceedings{Zhandry12,
	Acmid = {2417838},
	Address = {Washington, DC, USA},
	Author = {Zhandry, Mark},
	Booktitle = {Proceedings of the 53rd Annual Symposium on Foundations of Computer Science},
	Doi = {10.1109/FOCS.2012.37},
	Isbn = {978-0-7695-4874-6},
	Numpages = {9},
	Pages = {679--687},
	Publisher = {IEEE Computer Society},
	Series = {FOCS '12},
	Title = {How to Construct Quantum Random Functions},
	Year = {2012}}

@inproceedings{ito19feistcca,
author = {Ito, Gembu and Hosoyamada, Akinori and Matsumoto, Ryutaroh and Sasaki, Yu and Iwata, Tetsu},
title = {Quantum Chosen-Ciphertext Attacks Against Feistel Ciphers},
year = {2019},
isbn = {978-3-030-12611-7},
publisher = {Springer-Verlag},
address = {Berlin, Heidelberg},
url = {https://doi.org/10.1007/978-3-030-12612-4_20},
doi = {10.1007/978-3-030-12612-4_20},
booktitle = {Topics in Cryptology – CT-RSA 2019: The Cryptographers' Track at the RSA Conference 2019, San Francisco, CA, USA, March 4–8, 2019, Proceedings},
pages = {391–411},
numpages = {21},
location = {San Francisco, CA, USA}
}

@inproceedings{bhaumik24badnorm,
author = {Bhaumik, Ritam and Cogliati, Beno{\^{i}}t and Ethan, Jordan and Jha, Ashwin},
title = {Mind the Bad Norms: Revisiting Compressed Oracle-Based Quantum Indistinguishability Proofs},
year = {2024},
publisher = {Springer-Verlag},
address = {Berlin, Heidelberg},
url = {https://doi.org/10.1007/978-981-96-0947-5_8},
doi = {10.1007/978-981-96-0947-5_8},
booktitle = {Advances in Cryptology – ASIACRYPT 2024: 30th International Conference on the Theory and Application of Cryptology and Information Security, Kolkata, India, December 9–13, 2024, Proceedings, Part IX},
pages = {215–247},
numpages = {33},
location = {Kolkata, India}
}

@InProceedings{YZ21,
  author       = {Yamakawa, Takashi and Zhandry, Mark},
  booktitle    = {Annual International Conference on the Theory and Applications of Cryptographic Techniques},
  title        = {Classical vs quantum random oracles},
  year         = {2021},
  organization = {Springer},
  pages        = {568--597},
}

@inproceedings{KM10,
	Author = {Hidenori Kuwakado and Masakatu Morii},
	Booktitle = {Proc.\ IEEE International Symposium on
                  Information Theory},
	Pages = {2682-2685},
	Title = "Quantum distinguisher between the 3-round {Feistel} cipher and the random permutation",
	publisher = "IEEE",
	Year = {2010}}

@inproceedings{hosoyamada19lr,
author = {Hosoyamada, Akinori and Iwata, Tetsu},
title = {4-Round Luby-Rackoff Construction is a qPRP},
year = {2019},
isbn = {978-3-030-34577-8},
publisher = {Springer-Verlag},
address = {Berlin, Heidelberg},
url = {https://doi.org/10.1007/978-3-030-34578-5_6},
doi = {10.1007/978-3-030-34578-5_6},
booktitle = {Advances in Cryptology – ASIACRYPT 2019: 25th International Conference on the Theory and Application of Cryptology and Information Security, Kobe, Japan, December 8–12, 2019, Proceedings, Part I},
pages = {145–174},
numpages = {30},
location = {Kobe, Japan}
}

@article{Zhandry15,
	Acmid = {2871413},
	Address = {Paramus, NJ},
	Author = {Zhandry, Mark},
	Issue_Date = {May 2015},
	Journal = {Quantum Info. Comput.},
	Month = may,
	Number = {7-8},
	Numpages = {11},
	Pages = {557--567},
	Publisher = {Rinton Press, Incorporated},
	Title = {A Note on the Quantum Collision and Set Equality Problems},
	Volume = {15},
	Year = {2015}}

@misc{czajkowski19lazy,
      author = {Jan Czajkowski and Christian Majenz and Christian Schaffner and Sebastian Zur},
      title = {Quantum Lazy Sampling and Game-Playing Proofs for Quantum Indifferentiability},
      howpublished = {Cryptology ePrint Archive, Paper 2019/428},
      year = {2019},
      note = {\url{https://eprint.iacr.org/2019/428}},
      url = {https://eprint.iacr.org/2019/428}
}

@article{Hamoudi_2023,
   title={Quantum Time–Space Tradeoff for Finding Multiple Collision Pairs},
   volume={15},
   ISSN={1942-3462},
   url={http://dx.doi.org/10.1145/3589986},
   DOI={10.1145/3589986},
   number={1–2},
   journal={ACM Transactions on Computation Theory},
   publisher={Association for Computing Machinery (ACM)},
   author={Hamoudi, Yassine and Magniez, Frédéric},
   year={2023},
   month=jun, pages={1–22} }

@misc{Unruh2021,
      author = {Dominique Unruh},
      title = {Compressed Permutation Oracles (And the Collision-Resistance of Sponge/SHA3)},
      howpublished = {Cryptology ePrint Archive, Paper 2021/062},
      year = {2021},
      note = {\url{https://eprint.iacr.org/2021/062}},
      url = {https://eprint.iacr.org/2021/062}
}

@inproceedings{Unruh2023,
author = {Unruh, Dominique},
title = {Towards Compressed Permutation Oracles},
year = {2023},
isbn = {978-981-99-8729-0},
publisher = {Springer-Verlag},
address = {Berlin, Heidelberg},
url = {https://doi.org/10.1007/978-981-99-8730-6_12},
doi = {10.1007/978-981-99-8730-6_12},
booktitle = {Advances in Cryptology – ASIACRYPT 2023: 29th International Conference on the Theory and Application of Cryptology and Information Security, Guangzhou, China, December 4–8, 2023, Proceedings, Part IV},
pages = {369–400},
numpages = {32},
location = {Guangzhou, China}
}

@misc{Czajkowski2017,
      author = {Jan Czajkowski and Leon Groot Bruinderink and Andreas Hülsing and Christian Schaffner and Dominique Unruh},
      title = {Post-quantum security of the sponge construction},
      howpublished = {Cryptology ePrint Archive, Paper 2017/771},
      year = {2017},
      note = {\url{https://eprint.iacr.org/2017/771}},
      url = {https://eprint.iacr.org/2017/771}
}

@misc{rosmanis2022tight,
      title={Tight Bounds for Inverting Permutations via Compressed Oracle Arguments}, 
      author={Ansis Rosmanis},
      year={2022},
      eprint={2103.08975},
      archivePrefix={arXiv},
      primaryClass={quant-ph}
}

@InProceedings{BDPvA08,
author="Bertoni, Guido
and Daemen, Joan
and Peeters, Micha{\"e}l
and Van Assche, Gilles",
editor="Smart, Nigel",
title="On the Indifferentiability of the Sponge Construction",
booktitle="Advances in Cryptology -- EUROCRYPT 2008",
year="2008",
publisher="Springer Berlin Heidelberg",
address="Berlin, Heidelberg",
pages="181--197",
isbn="978-3-540-78967-3"
}

@InProceedings{10.1007/978-3-642-25385-0_3,
author="Boneh, Dan
and Dagdelen, {\"O}zg{\"u}r
and Fischlin, Marc
and Lehmann, Anja
and Schaffner, Christian
and Zhandry, Mark",
editor="Lee, Dong Hoon
and Wang, Xiaoyun",
title="Random Oracles in a Quantum World",
booktitle="Advances in Cryptology -- ASIACRYPT 2011",
year="2011",
publisher="Springer Berlin Heidelberg",
address="Berlin, Heidelberg",
pages="41--69",
isbn="978-3-642-25385-0"
}

@misc{KeccakSub3,
  author    = {G. Bertoni and J. Daemen and M. Peeters and G. Van Assche},
  title     = {The Keccak SHA-3 submission},
  url        = {http://keccak.noekeon.org/Keccak-submission-3.pdf},
  howpublished = {Submission to NIST (Round 3)},
  year      = {2011},
}

@misc{KeccakSponge3,
  author    = {G. Bertoni and J. Daemen and M. Peeters and G. Van Assche},
  title     = {Cryptographic sponge functions},
  url        = {http://sponge.noekeon.org/CSF-0.1.pdf},
  howpublished = {Submission to NIST (Round 3)},
  year      = {2011},
}

@inproceedings{Gro96,
	Address = {New York},
	Author = {Lov K. Grover},
	Booktitle = {Proceedings of the Twenty-Eighth Annual Symposium on the Theory of Computing},
	Pages = {212--219},
	Publisher = {ACM Press},
	Title = "A fast quantum mechanical algorithm for database search",
	Year = {1996}}

@misc{BHT97,
  author       = {Gilles Brassard and Peter H{\o}yer and Alain Tapp},
  title        = {Quantum Algorithm for the Collision Problem},
  year         = {1997},
  note = "Avalable at \url{https://arxiv.org/abs/quant-ph/9705002}",
}

@misc{alagic2023twosided,
      title={On the Two-sided Permutation Inversion Problem}, 
      author={Gorjan Alagic and Chen Bai and Alexander Poremba and Kaiyan Shi},
      year={2023},
      eprint={2306.13729},
      archivePrefix={arXiv},
      primaryClass={quant-ph}
}

@misc{don2021extract,
      author = {Jelle Don and Serge Fehr and Christian Majenz and Christian Schaffner},
      title = {Online-Extractability in the Quantum Random-Oracle Model},
      howpublished = {Cryptology {ePrint} Archive, Paper 2021/280},
      year = {2021},
      url = {https://eprint.iacr.org/2021/280}
}

@misc{cpz24precomp,
      author = {Joseph Carolan and Alexander Poremba and Mark Zhandry},
      title = {(Quantum) Indifferentiability and Pre-Computation},
      howpublished = {Cryptology {ePrint} Archive, Paper 2024/1727},
      year = {2024},
      url = {https://eprint.iacr.org/2024/1727}
}

@InProceedings{Zhandry2018,
 author="Zhandry, Mark",
 editor="Boldyreva, Alexandra
 and Micciancio, Daniele",
 title="How to Record Quantum Queries, and Applications to Quantum Indifferentiability",
 booktitle="Advances in Cryptology -- CRYPTO 2019",
 year="2019",
 publisher="Springer International Publishing",
 address="Cham",
 pages="239--268",
 isbn="978-3-030-26951-7"
 }

@inproceedings{chung2021compressed,
  title={On the compressed-oracle technique, and post-quantum security of proofs of sequential work},
  author={Chung, Kai-Min and Fehr, Serge and Huang, Yu-Hsuan and Liao, Tai-Ning},
  booktitle={Annual International Conference on the Theory and Applications of Cryptographic Techniques},
  pages={598--629},
  year={2021},
  organization={Springer}
}

@InProceedings{LM22,
	author="Lefevre, Charlotte
	and Mennink, Bart",
	editor="Dodis, Yevgeniy
	and Shrimpton, Thomas",
	title="Tight Preimage Resistance of the Sponge Construction",
	booktitle="Advances in Cryptology -- CRYPTO 2022",
	year="2022",
	publisher="Springer Nature Switzerland",
	address="Cham",
	pages="185--204",
	isbn="978-3-031-15985-5"
}

@inproceedings{BDPvA07,
	title={Sponge functions},
	author={Bertoni, Guido and Daemen, Joan and Peeters, Micha{\"e}l and van Assche, Gilles},
	booktitle="ECRYPT Hash Workhsop",
	year=2007
}

@InProceedings{Zhandry21,
	author="Zhandry, Mark",
	editor="Tibouchi, Mehdi
	and Wang, Huaxiong",
	title="Redeeming Reset Indifferentiability and Applications to Post-quantum Security",
	booktitle="Advances in Cryptology -- ASIACRYPT 2021",
	year="2021",
	publisher="Springer International Publishing",
	address="Cham",
	pages="518--548",
	isbn="978-3-030-92062-3"
}

@misc{MMW24,
	author = {Christian Majenz and Giulio Malavolta and Michael Walter},
	title = {Permutation Superposition Oracles for Quantum Query Lower Bounds},
	howpublished = {Cryptology {ePrint} Archive, Paper 2024/1140},
	year = {2024},
	url = {https://eprint.iacr.org/2024/1140}
}

@misc{zhandry2016notequantumsecureprps,
      title={A Note on Quantum-Secure PRPs}, 
      author={Mark Zhandry},
      year={2016},
      eprint={1611.05564},
      archivePrefix={arXiv},
      primaryClass={cs.CR},
      url={https://arxiv.org/abs/1611.05564}, 
}

@article{LR88,
	author = {Luby, Michael and Rackoff, Charles},
	title = {How to Construct Pseudorandom Permutations from Pseudorandom Functions},
	journal = {SIAM Journal on Computing},
	volume = {17},
	number = {2},
	pages = {373-386},
	year = {1988},
	doi = {10.1137/0217022},
	
	URL = { 
	
	https://doi.org/10.1137/0217022
	},
	eprint = { 
	
	https://doi.org/10.1137/0217022
	
	
	
	}
}
\end{document}